\documentclass[letter,onecolumn,12pt]{IEEEtran}
\usepackage{cite,amsmath,amssymb,url}
\usepackage{subfigure}
\usepackage{citesort}
\usepackage{fancyhdr}
\usepackage{dsfont}
\usepackage{array}
\usepackage[pdftex]{graphicx}
\newtheorem{theorem}{Theorem}

\newtheorem{lemma}{Lemma}

\newtheorem{corollary}{Corollary}

\newtheorem{proposition}{Proposition}

\newtheorem{conjecture}{Conjecture}

\newtheorem{sketch}{Sketch of Proof}

\newcommand{\dis}{\stackrel{d}{\sim}}

\newcommand{\papertitle}{Open-Loop Spatial Multiplexing and Diversity Communications in Ad Hoc Networks}

\begin{document}

\title{\papertitle}

\author{Raymond~H.~Y.~Louie,~\IEEEmembership{Student~Member,~IEEE,}
        Matthew~R.~McKay,~\IEEEmembership{Member,~IEEE,}
        and~Iain~B.~Collings,~\IEEEmembership{Senior~Member,~IEEE}
}

\maketitle
 \setcounter{page}{1} \thispagestyle{fancyplain}


\begin{abstract}
This paper investigates the performance of open-loop multi-antenna point-to-point links in ad hoc networks
with slotted ALOHA medium access control (MAC). We consider spatial multiplexing transmission with linear maximum ratio combining and zero
forcing receivers, as well as orthogonal space time block coded transmission. New closed-form expressions are derived for the outage probability, throughput and transmission capacity. Our results demonstrate that both the best performing scheme and the optimum number of transmit antennas depend on different network parameters, such as the node intensity and the signal-to-interference-and-noise ratio operating value. We then compare the performance to a network consisting of single-antenna devices and an idealized fully centrally coordinated MAC. These results show that multi-antenna schemes with a simple decentralized slotted ALOHA MAC can outperform even idealized single-antenna networks in various practical scenarios.
\end{abstract}

\begin{keywords}
Ad hoc networks, linear receivers, multiple antennas, OSTBC, spatial multiplexing, throughput, transmission capacity.
\end{keywords}

\vspace*{2cm}

\noindent \small{The work of M.\ R.\ McKay was supported by the Hong Kong Research Grants Council (RGC) General Research Fund under grant no. 617809. }

\noindent \small{The material in this paper was presented in part at the IEEE International Conference on Communications, Beijing, China, May 2008, and the IEEE International Conference on Communications, Dresden, Germany, July 2009.}

\newpage

%
%

\section{Introduction}

Multiple antennas offer the potential for significant performance improvements
in wireless communication systems by providing higher data rates and
more reliable links. A practical method which can achieve high data
rates is to employ spatial multiplexing transmission in conjunction
with low complexity linear receivers, such as maximum ratio
combining (MRC) or zero forcing (ZF) receivers. In the context of
point-to-point systems operating in the absence of interference, the
performance of such techniques has now been well studied (see e.g.,
\cite{tse05,louie08}). Multiple antennas can also be used to provide
more reliable links through spatial diversity techniques.  Of the
various spatial diversity techniques which have been proposed,
orthogonal space time block codes (OSTBC) have emerged as one of the
most important practical approaches, since they offer high diversity
gains, whilst requiring only very low computational complexity.  The
performance of OSTBC techniques have also been well studied in the
context of point-to-point systems, operating in the absence of
interference (see e.g.,\ \cite{larsson03,forenza06b}, and references
therein).

In this paper, we investigate spatial multiplexing techniques and
OSTBC in the context of wireless ad hoc networks using a slotted
ALOHA medium access control (MAC) protocol.  The interference model
we use includes the spatial distribution of nodes, with the nodes
distributed as a homogeneous Poisson point process (PPP) on the
plane. Besides approximating realistic network scenarios,
modeling the nodes according to a PPP has the benefit of allowing
network performance measures to be obtained.  This model has been
used previously in \cite{weber05} which considered code division
multiple access (CDMA) systems, and in \cite{weber07} which
considered single-antenna systems with threshold scheduling and
power control.  It was also used in \cite{hasan07}, where network
performance measures for a single antenna coordinated MAC protocol
were derived.

The PPP model has also been extended to ad hoc networks employing
multiple antennas at each transmitting and receiving node. In
particular, the schemes in
\cite{govindasamy07,huang08,jindal09,vaze09} considered the use of
the receive antennas to cancel interference from the $L$ closest
nodes to the receiver. However, a key requirement of these schemes
is that each receiver must measure the channels of the $L$ closest
interfering nodes, or at least the corresponding short-term
covariance information, the practical feasibility of which is still not clear. In contrast, the schemes in
\cite{hunter07,hunter08conf,kountouris09,stamatiou07} considered the
use of the receive antennas for either canceling self-interference
or increasing the signal power from the corresponding transmitter
only, and as such, require the receiver to estimate only the
transmitter-receiver channel (something which can be done with
standard channel estimation techniques). For this scenario,
\cite{hunter07} considered various spatial diversity techniques, and
derived transmission capacity and outage probability expressions. In
\cite{hunter08conf,kountouris09}, the multiple antennas were used
for spatial multiplexing, assuming a ``closed-loop'' scenario where channel knowledge is available at the  transmitter. In particular, \cite{hunter08conf}
considered multiple-input multiple-output (MIMO) singular value
decomposition systems, whereas \cite{kountouris09} considered MIMO broadcast transmission,
for which each transmitting node communicated independently with
multiple receiver nodes. 
Note that all of these papers, with the exception of
\cite{hunter07,stamatiou07}, assumed that either the receiver can acquire
knowledge about the network interference, or that there is
sufficient capabilities for each receiver to feedback channel
information back to their corresponding transmitter. In practice, it
seems reasonable to investigate simpler techniques which \emph{do
not} require the receiver to constantly measure the network
interference, nor require the transmitter to have channel
information. We refer to such schemes as point-to-point ``open-loop'' schemes.

One previous contribution which focuses on such schemes is provided
in \cite{stamatiou07}, where OSTBC and spatial multiplexing with ZF
receivers were considered, and approximations were derived for the
frame-error probability. In general, however, the fundamental
question as to whether to use open-loop spatial multiplexing or
diversity-based transmission in ad hoc networks, particularly from a
capacity-based network performance point of view, is not well
understood.

In this paper, we derive new network performance measures for
point-to-point open-loop spatial multiplexing with MRC and ZF
receivers\footnote{Note that we consider linear receivers in the
spatial multiplexing mode since non-linear receiver structures, such
as joint decoding, come at the expense of prohibitively higher
complexity, especially with
increasing number of antennas. We consider MRC and ZF receivers, but
not the better performing minimum-mean-squared-error (MMSE) receivers, since the analytic
expressions for MMSE receivers have proved intractable at this
stage. However, insights into the MMSE receiver can be directly
gained from the MRC and ZF results by noting that in both the
interference and non-interference scenario, the performance of MRC
and ZF receivers converges to the performance of MMSE receivers at
low and high signal-to-noise ratios respectively.}, and OSTBC, using
a slotted ALOHA MAC, and the same PPP ad hoc network model as in
\cite{weber05,hunter07}.  In all cases, we do not require the
receiver to know the network interference, and the transmitter does
not have any channel state information. (Note that for the ZF case,
each receiver cancels interference from its corresponding
transmitter only.) We derive new closed-form expressions for the
outage probability, network throughput, and transmission capacity.
For spatial multiplexing, these results are exact; whereas, for
OSTBC, they are accurate approximations, which we show to be
significantly more accurate than previous corresponding
approximations.

Our analytical results reveal key insights into  the relative
throughput performance of spatial multiplexing with MRC and ZF
receivers, as well as OSTBC, demonstrating that each scheme may
outperform the other, depending on different network parameters such
as the node intensity and the chosen signal-to-interference noise
ratio (SINR) operating value $\beta$. We also gain insights into the
optimal number of transmit antennas to employ for each scheme. For
certain scenarios, e.g., sufficiently dense networks or networks
with high operating SINR values, these insights are obtained
analytically. For example,  we show that single-stream transmission
is throughput-optimal in dense networks and for networks operating
with high $\beta$ requirements, and that transmitting the maximum
number of data streams is throughput-optimal for systems operating
with low $\beta$ requirements.

We then analyze the transmission capacity of the spatial
multiplexing and OSTBC systems. Among other things, we prove that for spatial multiplexing, the transmission capacity can scale \emph{linearly} with the number of antennas, and we derive precise conditions which must be met to achieve this scaling, using tools from large-dimensional random matrix theory. We also provide concrete design
guidelines for selecting the number of transmit streams for
achieving (and maximizing) this scaling, and contrast our results
with those in \cite{jindal09}, which derived a similar scaling
result with receivers employing (network) interference cancelation.
In addition, we prove that the transmission capacity of OSTBC scales
only sub-linearly with the number of antennas, regardless of the
particular code employed.

Finally, we turn from the analysis and comparison of different MIMO
schemes to look closer at the benefits of multiple antennas with a decentralized MAC, compared with the benefits of a tightly
coordinated MAC.  In particular, we compare the performance of a
multi-antenna system employing a simple slotted ALOHA MAC
protocol with a baseline scheme involving single-antenna devices and
a \emph{fully coordinated access} (CA) MAC which enforces guard
zones around each receiver. This scheme is similar to that proposed
in \cite{hasan07}, however we also consider a time division multiple
access (TDMA) scheme where only one transmitting node around each
receiver within the guard zone is scheduled to transmit. It is worth
noting that this is an idealized protocol, since the overhead
involved in achieving full coordination is prohibitive in practice
for ad hoc networks. To compare with this baseline scheme, we first
derive new throughput expressions for the CA MAC protocol. These
expressions allow us to show the important result that the slotted
ALOHA approach with multiple antennas can actually \emph{perform
better} than the idealized CA MAC, for a wide range of system
parameters. This shows that not only does the use of multiple
antennas compensate for the inherent performance loss
caused by using a simple decentralized MAC (i.e.,\ slotted ALOHA),
but it can actually yield an overall \emph{performance improvement}.
For sparse networks, we show that this is true for most slotted
ALOHA transmission probabilities, and for dense networks when there
is a sufficient number of antennas.

\section{System Model}\label{sec:sysmodel}

We consider an ad hoc network comprising of transmitter-receiver
pairs, where each transmitter communicates to its corresponding
receiver in a point-to-point manner, treating all other
transmissions as interference. In addition, each
transmitter-receiver pair are separated by a distance $r_{\rm tr}$ meters. The transmitting nodes are distributed spatially according to a
homogeneous PPP of intensity $\lambda$ nodes per meter squared in $\mathds{R}^2$. Each
transmitting node transmits with probability $p$ according to a
slotted ALOHA MAC protocol. As such, the effective intensity of
actual transmitting nodes is $p \lambda$.

In this paper, we investigate network performance measures. To
obtain such measures, it is sufficient to focus on a typical
transmitter-receiver pair, with the typical receiver located at the
origin. This can be done due to the Palm probabilities of a PPP,
which states that conditioning on a typical receiver located at the
origin does not affect the statistics of the rest of the process
\cite{stoyan95}. In addition, the stationarity property of the PPP
indicates that the statistics of the signal received at the typical
receiver is the same as for every other receiver. Note that the
typical transmitter, i.e.,\ the transmitter associated with the
typical transmitter-receiver pair, is not considered part of the
PPP.

We consider a network where each node is equipped with $N$ antennas. Each transmitting node communicates using $M$ out of its $N$ antennas, whereas each receiver operates employing all $N$ antennas. As we will show, transmitting with less than $N$ antennas may, in fact, lead to an
increased network throughput due to the reduced interference
in the network. The transmitting nodes, with the exception of the typical transmitter, constitute a marked PPP. This is denoted by $\Phi =\{ ( D_{\ell}, \mathbf{H}_{\ell}), \ell \in \mathds{N}^+ \}$,  where $D_{\ell}$ and\footnote{The notation $X \dis Y$ means that $X$ \emph{is distributed as} $Y$.}  $\mathbf{H}_{\ell} \dis \mathcal{CN}_{N,M} \left(\mathbf{0}_{N \times M}, \mathbf{I}_{N \times M} \right)$ model the location and channel matrix respectively of the $\ell$th transmitting node with respect to (w.r.t.) the typical receiver.  Further, we denote the typical transmitter as the $0$th transmitting node, and the channel matrix of the typical transmitter-receiver pair given by $\mathbf{H}_{0} \dis \mathcal{CN}_{N,M} \left(\mathbf{0}_{N \times M}, \mathbf{I}_{N \times M} \right)$. Each transmitting node is assumed to use
the same transmission power $P$, and the transmitted signals are
attenuated by a factor $1/r^{\alpha}$ with distance $r$ where $\alpha>2$ is the path loss
exponent\footnote{Note that there are more
accurate, but more complicated, path loss models particularly suited for
dense networks; e.g.,\ $1/(1+r^{\alpha})$. We have considered such a
model, but have found that it leads to more cumbersome expressions, without changing the fundamental insights gained based on the simpler model employed in this paper.}.

We consider the practical scenario where each receiver has perfect knowledge of the channel to its corresponding transmitter, but does not know the channel to the other transmitting nodes. Moreover, we assume an open-loop scenario where there are no feedback links between the transmitters and receivers, and as such, all transmitters have no channel state information (CSI). This is particularly relevant to ad hoc networks, where obtaining CSI may be difficult due to their changing nature.

In this paper, we consider the use of multiple antennas for either open-loop point-to-point spatial multiplexing with linear receivers, or OSTBC.  There has been little work, to the author's knowledge, done on analyzing and comparing these systems in ad hoc networks.  We analyze important network performance measures for these spatial multiplexing and OSTBC systems, and draw key insights into their relative performance.

\subsection{Spatial Multiplexing with Linear Receivers}

For spatial multiplexing transmission, we assume that each
transmitting node sends $M\le N$ independent data streams to its
corresponding receiver.  Focusing on the $k$th stream, the received
$N \times 1$ signal vector at the typical receiver can be written as
\begin{align}
\mathbf{y}_{0,k} &= \sqrt{\frac{1}{r_{\rm tr}^\alpha}} \mathbf{h}_{0,k} x_{0,k} + \sqrt{\frac{1}{r_{\rm tr}^\alpha}} \sum_{q=1, q \neq k}^{M}
\mathbf{h}_{0,q} x_{0,q} + \sum_{D_{\ell} \in \Phi}\sqrt{\frac{1}{|D_{\ell}|^\alpha}} \sum_{q=1}^M \mathbf{h}_{\ell,q} x_{\ell,q}  +\mathbf{n}_{0,k}
\end{align}
where $\mathbf{h}_{p,q}$ is the $q$th column of $\mathbf{H}_{p}$, $x_{p,q}$ is the symbol sent from the $q$th transmit antenna of the $p$th transmitting node satisfying ${\rm E}[|x_{p,q}|^2]=\frac{P}{M}$, and $\mathbf{n}_{0,k}\dis \mathcal{CN}_{N, 1} \left(\mathbf{0}_{N \times 1}, N_0 \mathbf{I}_{N} \right)$ is the complex additive white Gaussian noise (AWGN) vector. To obtain an estimate for $x_{0,k}$, we consider the use of low complexity MRC and ZF linear receivers.


For MRC, the data estimate is formed via $\hat{x}_{0,k} =
\mathbf{h}_{0,k}^\dagger \mathbf{y}_{0,k}$, where $(\cdot)^\dagger$
denotes conjugate transpose, from which the SINR can be written as
\begin{align}\label{eq:snr_bf}
\gamma_{{\rm MRC},0,k} =
\frac{\frac{\rho}{M r_{\rm tr}^\alpha} ||
\mathbf{h}_{0,k}||^2}{\frac{\rho}{M r_{\rm tr}^\alpha} \sum_{q=1, q \neq k}^{M}
\frac{|\mathbf{h}_{0,k}^\dagger \mathbf{h}_{0,q}|^2}{||
\mathbf{h}_{0,k}||^2} +
\frac{\rho}{M} \sum_{D_{\ell} \in \Phi}\frac{\sum_{q=1}^M  \frac{|\mathbf{h}_{0,k}^\dagger \mathbf{h}_{\ell,q}|^2}{||
\mathbf{h}_{0,k}||^2}}{|D_{\ell}|^\alpha}  + 1} \notag
\end{align}
where $\rho = \frac{P}{N_0}$ is the transmit signal-to-noise ratio (SNR). Note that
\begin{align}
& \frac{\rho}{M r_{\rm tr}^\alpha} || \mathbf{h}_{0,k}||^2\dis {\rm
Gamma}\left(N,\frac{\rho}{M r_{\rm tr}^\alpha}\right), \\
& \frac{\rho}{M r_{\rm tr}^\alpha} \sum_{q=1, q \neq k}^{M}
\frac{|\mathbf{h}_{0,k}^\dagger \mathbf{h}_{0,q}|^2}{||
\mathbf{h}_{0,k}||^2} \dis {\rm Gamma}\left(M-1,\frac{\rho}{M r_{\rm
tr}^\alpha}  \right) \notag  ,
\end{align}
and
\begin{align}
\frac{\rho}{M} \sum_{q=1}^M \frac{|\mathbf{h}_{0,k}^\dagger
\mathbf{h}_{\ell,q}|^2}{|| \mathbf{h}_{0,k}||^2} \dis {\rm
Gamma}\left(M, \frac{\rho}{M} \right) \, ,
\end{align}
where ${\rm Gamma}(v,\theta)$ denotes a gamma random variable with
shape parameter $v$ and scale parameter $\theta$, with probability density function (p.d.f.)
\begin{align}
f_X(x) = \frac{x^{v-1} e^{-\frac{x}{\theta}}}{\Gamma(v) \theta^v} \,
, \quad x \geq 0 \; .
\end{align}

For ZF, since each receiver only knows the CSI of its corresponding
transmitter, the ZF weight vector is designed to cancel interference
due to the other data streams (i.e.,\ other than the one being
detected) originating from the corresponding transmitter only. The
data estimate can thus be written as
\begin{align}\label{eq:rec_zf}
\hat{x}_{0,k} & =\mathbf{g}_{0,k}^\dagger \mathbf{y}_{0,k} \notag \\
&= \sqrt{\frac{1}{r_{\rm tr}^\alpha}} {x}_{0,k} + \sum_{D_{\ell} \in
\Phi}\sqrt{\frac{1}{|D_{\ell}|^\alpha}} \mathbf{g}_{0,k}^\dagger
\mathbf{H}_{\ell} \mathbf{x}_{\ell} + \mathbf{g}_{0,k}^\dagger
\mathbf{n}_{0,k}
\end{align}
where $\mathbf{g}_{0,k}^\dagger$ is the $k$th row of $\left(\mathbf{H}_0^\dagger \mathbf{H}_0 \right)^{-1} \mathbf{H}_0^\dagger$. The SINR for ZF follows from (\ref{eq:rec_zf}), and is given by
\begin{align}\label{eq:snr_zf}
\gamma_{{\rm ZF},0,k}
&= \frac{\frac{\rho}{M r_{\rm tr}^{\alpha} \left[ (\mathbf{H}_0^\dagger \mathbf{H}_0 )^{-1} \right]_{k}}}{
\frac{\rho\left[(\mathbf{H}_0^\dagger \mathbf{H}_0 )^{-1} \mathbf{H}_0^\dagger\left(\sum_{D_{\ell} \in \Phi} \frac{1}{|D_{\ell}|^\alpha}\mathbf{H}_{\ell} \mathbf{H}_{\ell}^\dagger  \right) \mathbf{H}_0 (\mathbf{H}_0^\dagger \mathbf{H}_0 )^{-1}  \right]_{k}}{M\left[ (\mathbf{H}_0^\dagger \mathbf{H}_0 )^{-1} \right]_{k}} +1}
\end{align}
where $[ \cdot ]_{k}$ denotes the $(k,k)$th element. Note that
\begin{align}
\frac{\rho}{M r_{\rm tr}^{\alpha} \left[ (\mathbf{H}_0^\dagger
\mathbf{H}_0 )^{-1} \right]_{k}}\dis {\rm
Gamma}\left(N-M+1,\frac{\rho}{M r_{\rm tr}^{\alpha}}  \right)
\end{align}
and, from  \cite[Eq. (2.47)]{mckay06_thesis} and \cite{james64}, it
can be shown that
\begin{align}
& \frac{\rho}{M} \frac{\left[(\mathbf{H}_0^\dagger \mathbf{H}_0 )^{-1}
\mathbf{H}_0^\dagger\left(\mathbf{H}_{\ell}
\mathbf{H}_{\ell}^\dagger  \right) \mathbf{H}_0
(\mathbf{H}_0^\dagger \mathbf{H}_0 )^{-1}  \right]_{k}}{\left[
(\mathbf{H}_0^\dagger \mathbf{H}_0 )^{-1} \right]_{k}} \dis {\rm
Gamma}\left(M,\frac{\rho}{M}  \right) \, .
\end{align}

\subsection{OSTBC}

For OSTBC, $N_s$ different symbols are transmitted over $\tau$ time slots using $M\le N$ antennas. Different codes have been proposed for OSTBC (see e.g.,\ \cite{alamouti98,paulraj03,tarokh99,liang03}), each being characterized by different values of $N_s$, $M$ and $\tau$. Associated with each code is a code rate, which is defined by $R := \frac{N_s}{\tau}$.

The transmitted OSTBC $M \times \tau$ code matrix is given by
\begin{align}\label{eq:codematrix}
\mathbf{X}_{\ell} = \sum_{q=1}^{N_s} \left(x_{\ell,q} \mathbf{A}_{q} + x_{\ell,q}^* \mathbf{B}_{q} \right)
\end{align}
where $x_{\ell,q}$ is the $q$th transmitted symbol of the $\ell$th transmitting node, and $\mathbf{A}_{q}$ and $\mathbf{B}_{q}$ are $M \times \tau$ matrices, both of which are dependent on the particular code employed. The received $N \times \tau$ signal matrix at the typical receiver can be written as
\begin{align}
\mathbf{Y}_0 = \sqrt{\frac{1}{r_{\rm tr}^\alpha}} \mathbf{H}_0 \mathbf{X}_0 + \sum_{D_{\ell} \in \Phi} \sqrt{\frac{1}{|D_{\ell}|^\alpha}} \mathbf{H}_{\ell} \mathbf{X}_{\ell} + \mathbf{N}_{0}
\end{align}
where $\mathbf{N}_{0}\dis \mathcal{CN}_{N,\tau}\left(\mathbf{0}_{N \times \tau}, N_0 \mathbf{I}_{N \times \tau} \right)$ is the AWGN matrix. We assume that the channels $\mathbf{H}_0$ and $\mathbf{H}_{\ell}$ are constant during the $\tau$ time slots used for transmission. To obtain an expression for the data estimate for the $k$th symbol, it is convenient to introduce the matrix function $\varsigma_{k}(\cdot) : \mathcal{C}^{N \times \tau} \to \mathcal{C}^{N \times \tau}$, which, for a given input matrix $\mathbf{V}$ with $(p,q)$th element $v_{p,q}$, produces the matrix $\mathbf{Z}_k=\varsigma_k(\mathbf{V})$  with $(p,q)$th element
\begin{align}
z_{k,p,q}=
\begin{cases}
v_{p,q}^* \hspace{1cm} {\rm if} \hspace{1cm} (p,q) \in \underline{\varphi}_k \notag \\
-v_{p,q} \hspace{0.7cm} {\rm if} \hspace{1cm} (p,q) \in \underline{\chi}_k \notag \\
v_{p,q} \hspace{1cm} {\rm if} \hspace{1cm} (p,q) \in \underline{\varrho}  \backslash (\underline{\varphi}_k \cup \underline{\chi}_k)
\end{cases}
\end{align}
where $(\cdot)^*$ denotes conjugate, and $\underline{\varrho}$ denotes the entire matrix index set given by
\begin{align}\label{eq:elmement_set2}
\underline{\varrho} = \left\{(p,q): p \in 1,\ldots,N \cap q \in 1,\ldots, \tau \right\} \; ,
\end{align}
with $\underline{\varphi}_k \subseteq \underline{\varrho}$ and $\underline{\chi}_k\subseteq \underline{\varrho}$.
Here, the mapping function $\varsigma_k(\cdot)$, and index sets $\underline{\varphi}_k$ and $\underline{\chi}_k$, depend, once again, on the specific OSTBC code employed.  Given these code-specific parameters, the data estimate for the $k$th symbol is then obtained via the following operation
\begin{align} \label{eq:DataEst}
\hat{x}_{0,k} = ||\mathbf{M}_k \odot \varsigma_k\left( \mathbf{Y}_0 \right)||_1
\end{align}
where $|| \cdot ||_1$ denotes 1-norm (i.e.,\ the sum of all entries), $\odot$ denotes Hadamard product (i.e.,\ elementwise product), and $\mathbf{M}_k$ is an $N \times \tau$ matrix chosen according to the principle of MRC to satisfy
\begin{align}\label{eq:ostbc_req_sys}
||\mathbf{M}_k \odot \varsigma_k\left(\mathbf{H}_0 \mathbf{X}_0\right)||_1  = ||\mathbf{H}_0||_F^2 x_{0,k}
\end{align}
where $|| \cdot ||_F$ denotes Frobenius norm.

To further illustrate the code-specific parameters $\mathbf{A}_{q}$, $\mathbf{B}_{q}$, $\mathbf{X}_{\ell}$, $\varsigma_k(\cdot)$ and $\mathbf{M}_k$, in the general OSTBC model presented above, let us consider the following concrete example.

\emph{Example: [Alamouti code with $N=2$]:} In this case, we have
\begin{align}\label{eq:code_construct_alamouti}
& \mathbf{A}_{1}=\left[ \begin{array}{cc}
1  & 0  \\
0  & 0  \end{array} \right] \; , \;
\mathbf{B}_{1}=\left[ \begin{array}{cc}
0  & 0  \\
0  & 1  \end{array} \right] \notag \\
& \mathbf{A}_{2}=\left[ \begin{array}{cc}
0  & 0  \\
1  & 0  \end{array} \right] \; \; {\rm and} \; \;
\mathbf{B}_{2}=\left[ \begin{array}{cc}
0  & -1  \\
0  & 0  \end{array} \right] \; .
\end{align}
Substituting (\ref{eq:code_construct_alamouti}) into (\ref{eq:codematrix}), the OSTBC code matrix $\mathbf{X}_{\ell}$ can be written as
\begin{equation}\label{eq:alamouti_param}
\mathbf{X}_{\ell}=\left[ \begin{array}{cc}
x_{\ell,1}  & -x_{\ell,2}^*  \\
x_{\ell,2}  & x_{\ell,1}^*  \end{array} \right] \; .
\end{equation}
Let us focus on decoding $x_{0,1}$ (i.e.,\ $k=1$). Then the mapping function $\varsigma_1(\cdot)$ is given by
\begin{equation}\label{eq:alamouti_param2}
\varsigma_1(\mathbf{V}) = \left[ \begin{array}{cc}
v_{1,1}  & v_{1,2}^*  \\
v_{2,1}  & v_{2,2}^*  \end{array} \right] \; .
\end{equation}
Substituting (\ref{eq:alamouti_param}) and (\ref{eq:alamouti_param2}) with $\ell=0$ into (\ref{eq:ostbc_req_sys}), we find that the matrix $\mathbf{M}_1$ which solves the resulting equation is given by
\begin{equation}\label{eq:alamouti_param3}
\mathbf{M}_1 = \left[ \begin{array}{cc}
h_{0,1,1}^*  & h_{0,1,2}  \\
h_{0,2,1}^*  & h_{0,2,2}  \end{array} \right] \; .
\end{equation}
\hfill \interlinepenalty500 $\Box$

Returning to the general case, by noting that $\varsigma_k(\cdot)$ and $\odot$ are both linear functions, the data estimate (\ref{eq:DataEst}) for the $k$th symbol can be written as
\begin{align}
\hat{x}_{0,k}
&= \frac{|\mathbf{H}_0||_F^2 x_{0,k} }{\sqrt{r_{\rm tr}^\alpha}} |+ \sum_{D_{\ell} \in \Phi} \sqrt{\frac{1}{|D_{\ell}|^\alpha}} ||\mathbf{M}_k \odot \varsigma_k \left(\mathbf{H}_{\ell} \mathbf{X}_{\ell} \right)||_1  + ||\mathbf{M}_k \odot \varsigma_k\left(\mathbf{N}_{0}\right)||_1 \; ,
\label{eq:data_est_ostbca}
\end{align}
from which the SINR is obtained as
\begin{align}\label{eq:ostbc_sinr_exact}
\gamma_{{\rm OSTBC},0,k} 
&= \frac{\frac{\rho}{R M r_{\rm tr}^\alpha} ||\mathbf{H}_0||_F^2}{\frac{\rho}{RM} \sum_{D_{\ell} \in \Phi} \frac{1}{|D_{\ell}|^\alpha} \mathcal{K}_{\ell,\sum} + 1 }
\end{align}
where
\begin{align}
\mathcal{K}_{\ell, \sum} = \frac{{\rm E}_{\mathbf{X}_{\ell}}\left[||\mathbf{M}_k \odot \varsigma_k \left(\mathbf{H}_{\ell} \mathbf{X}_{\ell} \right)||_1^2\right]}{||\mathbf{H}_0||_F^2}
\end{align}
is the normalized interference power for the $\ell$th transmitting node.

Deriving exact closed-form expressions for the SINR distribution is
difficult, based on the exact SINR expression in
(\ref{eq:ostbc_sinr_exact}). To proceed, we focus on deriving an approximation for the SINR,
based on assumptions relating to the distribution of the
$\mathcal{K}_{\ell, \sum}$ terms, and the independence between
different random variables. We explain and justify these assumptions in Appendix
\ref{app_ostbc_sinr_bound}, some of which were also used in
\cite{hunter07,choi07}. Our approximation for $\gamma_{{\rm OSTBC},0,k}$,
derived in Appendix \ref{app_ostbc_sinr_bound},  is given by
\begin{align}\label{eq:ostbc_sinr_approximate}
\tilde{\gamma}_{{\rm OSTBC},0,k} & = \frac{\frac{\rho}{R M r_{\rm tr}^\alpha} ||\mathbf{H}_0||_F^2}{\frac{\rho}{RM} \sum_{D_{\ell} \in \Phi}\frac{1}{|D_{\ell}|^\alpha}  \tilde{\mathcal{K}}_{\ell,\sum} + 1 }
\end{align}
where
\begin{align}
\frac{\rho}{RM} \tilde{\mathcal{K}}_{\ell,\sum} \dis {\rm
Gamma}\left(\frac{N_{I}}{M},\frac{\rho}{RM}  \right) \quad \; , \frac{\rho}{R M r_{\rm tr}^\alpha} ||\mathbf{H}_0||_F^2 \dis {\rm
Gamma}\left(MN,\frac{\rho}{R M r_{\rm tr}^\alpha} \right) \;
\end{align}
are independent random variables, and $N_{I}$ is the total number of non-zero elements in the columns of the OSTBC code matrix $\mathbf{X}_{\ell}$ containing either $\pm x_{\ell,k}$ or
$\pm x_{\ell,k}^*$. We will show in Section
\ref{sec:per} that the approximation in
(\ref{eq:ostbc_sinr_approximate}) is very accurate, significantly
more so than a previous approximation presented in
\cite{hunter07}.

\subsection{Performance Measures}

In this paper, we consider three main performance measures; namely, outage probability, network throughput, and transmission capacity. The outage probability is defined as the probability that the SINR falls below a certain threshold $\beta$, i.e.\footnote{Note that an alternative outage definition, adopted in \cite{vaze09}, would be to consider the probability that the \emph{overall sum-rate} (summed over all data streams) lies below a certain threshold. With this definition, different insights than those presented in this paper may be obtained, and this is the subject of ongoing work.},
\begin{align}\label{eq:SINR_cdf}
{\rm F}(\beta) = {\rm Pr}({\rm SINR} \le \beta)\; .
\end{align}
The network throughput is defined as the total number of successful
transmitted symbols/channel use/unit area. This is given by
\begin{align}\label{eq:through_def}
{\rm T} = \zeta p \lambda \left(1-{\rm F}(\beta)\right)
\end{align}
where $\zeta$ is the average number of transmitted symbols per node per channel use, and is given in Table \ref{table:gamma_param} for spatial multiplexing and OSTBC.

Although throughput is an important performance measure, it may be obtained at the expense of unacceptably high outage levels, which is undesirable for some applications due to, for example, significant delays caused by data retransmission. This has motivated the introduction of the \emph{transmission capacity} \cite{weber05}, defined as the maximum throughput subject to an outage constraint $\epsilon$, where the maximization is performed over all intensities $\lambda p$.  The transmission capacity is thus given by
\begin{align}\label{eq:tc_def}
{\rm c}(\epsilon) =  \zeta \lambda(\epsilon)(1 - \epsilon)
\end{align}
where $\lambda(\epsilon)$ is the \emph{contention density}, defined as the inverse of $\epsilon={\rm F}(\beta;\lambda p)$ taken w.r.t.\ $\lambda p$.  Note that here we have made explicit the dependence of the outage probability on $\lambda p$. In the following three sections, we will investigate each of these three performance measures for spatial multiplexing and OSTBC systems.

\section{Outage Probability and Network Throughput: Exact Analysis}\label{sec:per}

In this section, we derive new exact closed-form expressions for the
outage probability and network throughput for the spatial
multiplexing and OSTBC systems. To facilitate the derivations, we
first note that the received SINR for the spatial multiplexing and
OSTBC systems can be written in the general form
\begin{align}\label{eq:SINR_general}
\gamma = \frac{W}{Y + \sum_{\ell \in \Phi} |X_\ell|^{-\alpha} \Psi_{\ell 0}+ 1}
\end{align}
with the generic random variables
\begin{align}
& W \dis {\rm Gamma} \left(m,\theta \right) , \quad \quad  Y \dis {\rm
Gamma} \left(u,\Upsilon \right)  \Psi_{\ell 0} \dis
{\rm Gamma} \left(n, \Omega\right) \, .
\end{align}
Here $W$ represents the effective signal power from the desired
transmitter, $Y$ represents the effective self-interference power
from the desired transmitter, and $\Psi_{\ell 0}$ represents the
effective interference power from the interfering transmitting
nodes. These are all mutually independent. The particularizations of
the shape and scale parameters of $W$, $Y$, and $\Psi_{\ell 0}$ for
the spatial multiplexing and OSTBC systems are summarized in Table
\ref{table:gamma_param}.

\begin{table*}[!t]
\caption{Gamma parameters and $\zeta$ values for the spatial multiplexing and OSTBC systems.}\label{table:gamma_param}
\centering
\begin{tabular}{|c|c|c|c|}
  \hline
   & SM with MRC receivers  &  SM with ZF receivers & OSTBC \\
  \hline
Desired signal, $W \dis {\rm Gamma}(m,\theta)$ & $m=N$, $\theta=\frac{\rho}{M r_{\rm tr}^\alpha}$ & $m=N-M+1$, $\theta=\frac{\rho}{M r_{\rm tr}^\alpha}$ & $m=MN$, $\theta=\frac{\rho}{R M r_{\rm tr}^\alpha}$\\
Multi-node interference, $\Psi_{\ell0} \dis {\rm Gamma}(n, \Omega)$ & $n=M$, $\Omega=\frac{\rho}{M }$ & $n=M$, $\Omega=\frac{\rho}{M }$ & $n= \frac{N_I}{M}$, $\Omega=\frac{\rho}{RM}$\\
Self-interference, $Y \dis {\rm Gamma}(u,\Upsilon)$ &  $u=M-1$, $\Upsilon=\frac{\rho}{M r_{\rm tr}^\alpha}$ & N/A & N/A \\
Effective multiplexing gain, $\zeta$ &  $M$ & $M$ & $R$ \\
  \hline
\end{tabular}
\end{table*}

From Table \ref{table:gamma_param}, we can directly compare the SINR
distributions of each system and make some important observations.
Focusing on the spatial multiplexing systems, it is evident that the distribution of the effective
interference $\Psi_{\ell 0}$ caused by the interfering transmitting
nodes is the same for both MRC and ZF receivers. The difference lies in the effective signal power $W$ and the self-interference $Y$. For ZF receivers, the receive d.o.f.\ is split between boosting the signal power and canceling the total self-interference. For MRC receivers, the total receive d.o.f.\ is used to boost the signal power, but no effort is made to cancel self-interference..

Based on (\ref{eq:SINR_general}), we present the following general
theorem which, after substituting the parameters in Table
\ref{table:gamma_param}, yields exact closed-form expressions for
the outage probability of the spatial multiplexing and OSTBC
schemes.

\begin{theorem}\label{lemm_general_cdf}
If the SINR $\gamma$ takes the general form (\ref{eq:SINR_general}),
then its cumulative distribution function (c.d.f.) is
\begin{align}\label{eq:cdf_general}
{\rm F}_{\gamma}(\beta) &=  1 - \frac{(-1)^{m-1} e^{- \lambda
\left(\frac{\beta \Omega}{\theta} \right)^{\frac{2}{\alpha}} \eta(n)
} e^{-\frac{\beta}{\theta}}}{\Gamma(m) } \sum_{\ell=0}^{m-1}
\binom{m-1}{\ell} \left(-\frac{\beta}{\theta}\right)^{\ell}
\sum_{i=0}^{m-\ell-1} s(m-\ell,i+1) \left(\frac{2}{\alpha}\right)^i  \notag \\
& \hspace{1cm} \times \sum_{j=0}^i S(i,j) \left(-\lambda \left(
\frac{\beta \Omega}{\theta} \right)^{\frac{2}{\alpha}} \eta (n)
\right)^j   \sum_{\tau=0}^{\ell} \binom{\ell}{\tau} {\rm
E}_Y\left[e^{-\frac{ \beta Y}{\theta}} Y^\tau\right]
\end{align}
where
\begin{align} \label{eq:etaDefn}
\eta (n) :=  {\frac{\pi p \Gamma\left(n + \frac{2}{\alpha} \right)
\Gamma\left(1-\frac{2}{\alpha} \right)
 }{\Gamma(n)}}
\end{align}
and $s(n,m)$ and $S(n,m)$ are Stirling numbers of the first and second kind respectively \cite[pp. 824]{abramowitz70}.
Further, $ {\rm E}_Y\left[e^{-\frac{ \beta Y}{\theta}}
Y^\tau\right]$ captures the effects of the self-interference on the
outage probability, and is given by
\begin{align} \label{eq:exp_poly_selfint_term}
{\rm E}_Y\left[e^{-\frac{ \beta Y}{\theta}} Y^\tau\right] &=
\frac{\Gamma(\tau+u) }{\Gamma(u) \Upsilon^u } \left(\frac{\beta
}{\theta}+\frac{1}{\Upsilon}\right)^{-\tau-u} \; .
\end{align}
\end{theorem}
\begin{proof}
See Appendix \ref{app:general_cdf_proof}.
\end{proof}

Note that the outage probability and transmission capacity of
systems with SINRs of the general form (\ref{eq:SINR_general}) were
also considered previously in \cite{weber07} and \cite{hunter07},
focusing specifically on the special case $Y=0$ (i.e., no
self-interference). However, in contrast to our results,
\cite{weber07} presented bounds, rather than exact expressions;
whereas the results in \cite{hunter07}, whilst exact, were derived
using different methods and were expressed in a more complicated
form involving summations over subsets.

\begin{corollary}
For $m = 1$ and $Y=0$, the SINR c.d.f.\ (\ref{eq:cdf_general})
becomes
\begin{align}\label{eq:cdf_general_simple}
{\rm F}_{\gamma}(\beta) &=  1 - e^{- \lambda \left( \frac{\beta
\Omega}{\theta}\right)^{\frac{2}{\alpha}} \eta(n) }
e^{-\frac{\beta}{\theta}}  \; .
\end{align}
\end{corollary}
Note that this very simple expression can be used to give the outage
probability of spatial multiplexing with ZF receivers when $M=N$.

\subsection{Throughput of Spatial Multiplexing}

To compute the throughput achieved by spatial multiplexing with MRC
and ZF receivers, we substitute the relevant parameters from Table
\ref{table:gamma_param} into (\ref{eq:cdf_general}), and substitute
the resulting expression into (\ref{eq:through_def}).

Figs. \ref{fig:mrc_throughput_changingM} and
\ref{fig:zf_throughput_changingM} show the throughputs achieved by
spatial multiplexing with MRC and ZF receivers for different SINR
operating values $\beta$, and different numbers of transmission streams
$M$.  In both cases, results are shown for $N=4$ antennas. We see
that for all curves, the throughput increases monotonically with
$\lambda$ up to a certain value, after which it decreases
monotonically. This behavior is intuitive, since increasing
$\lambda$ results in a higher number of transmissions in the
network, however it also yields more interference for a given link.
The inherent trade-off between these competing factors is clear from
the figures.

For both receiver
structures, we see that in many cases, the throughput is significantly higher when less
antennas are used for transmission.  This is particularly
significant if the SINR operating value $\beta$ is high, and when the spatial intensity
$\lambda$ is large.  This behavior can be explained by noting that
in these regimes, the additional interference in the network caused
by each transmitting node employing more antennas outweighs the
benefits of an increased multiplexing gain. As we may expect,
however, for the contrasting scenario where $\beta$ and $\lambda$
are low, it is beneficial (in terms of network throughput) to use
more transmit antennas. The results in Figs. \ref{fig:mrc_throughput_changingM} and
\ref{fig:zf_throughput_changingM} demonstrate various important
tradeoffs which arise in multi-antenna ad hoc networks when using
spatial multiplexing transmission. We will examine these further in
the following section, where we derive simplified expressions for
various asymptotic regimes.

\begin{figure}[tb!]
\centerline{\includegraphics[width=0.7\columnwidth]{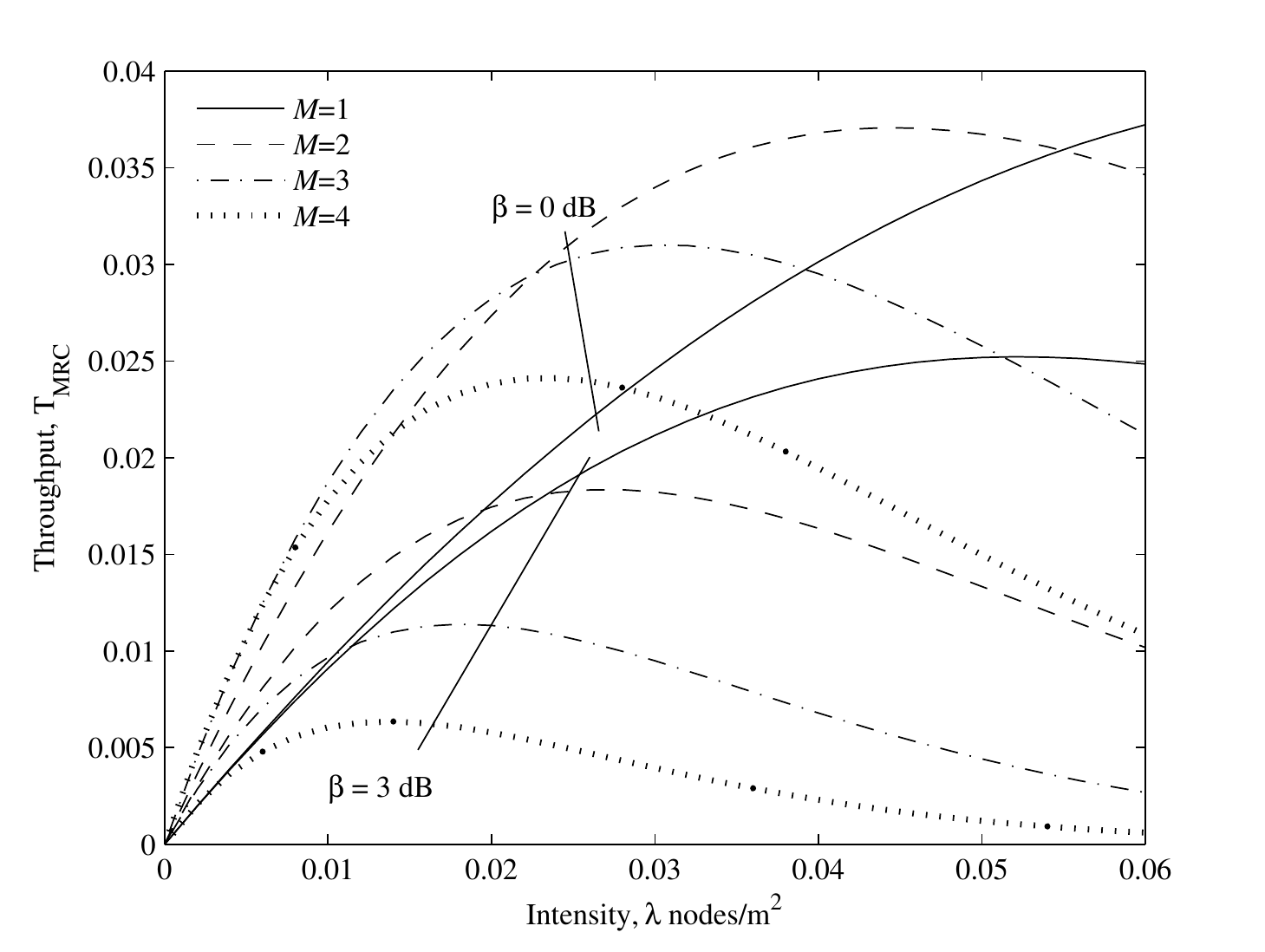}}
\caption{Throughput vs intensity of slotted ALOHA with spatial multiplexing and MRC receivers, and with $N=4$, $\alpha=3.1$, ${r_{\rm tr}}=2$ m, $\rho=25$ dB and $p=1$.}
\label{fig:mrc_throughput_changingM}
\end{figure}

\begin{figure}[tb!]
\centerline{\includegraphics[width=0.7\columnwidth]{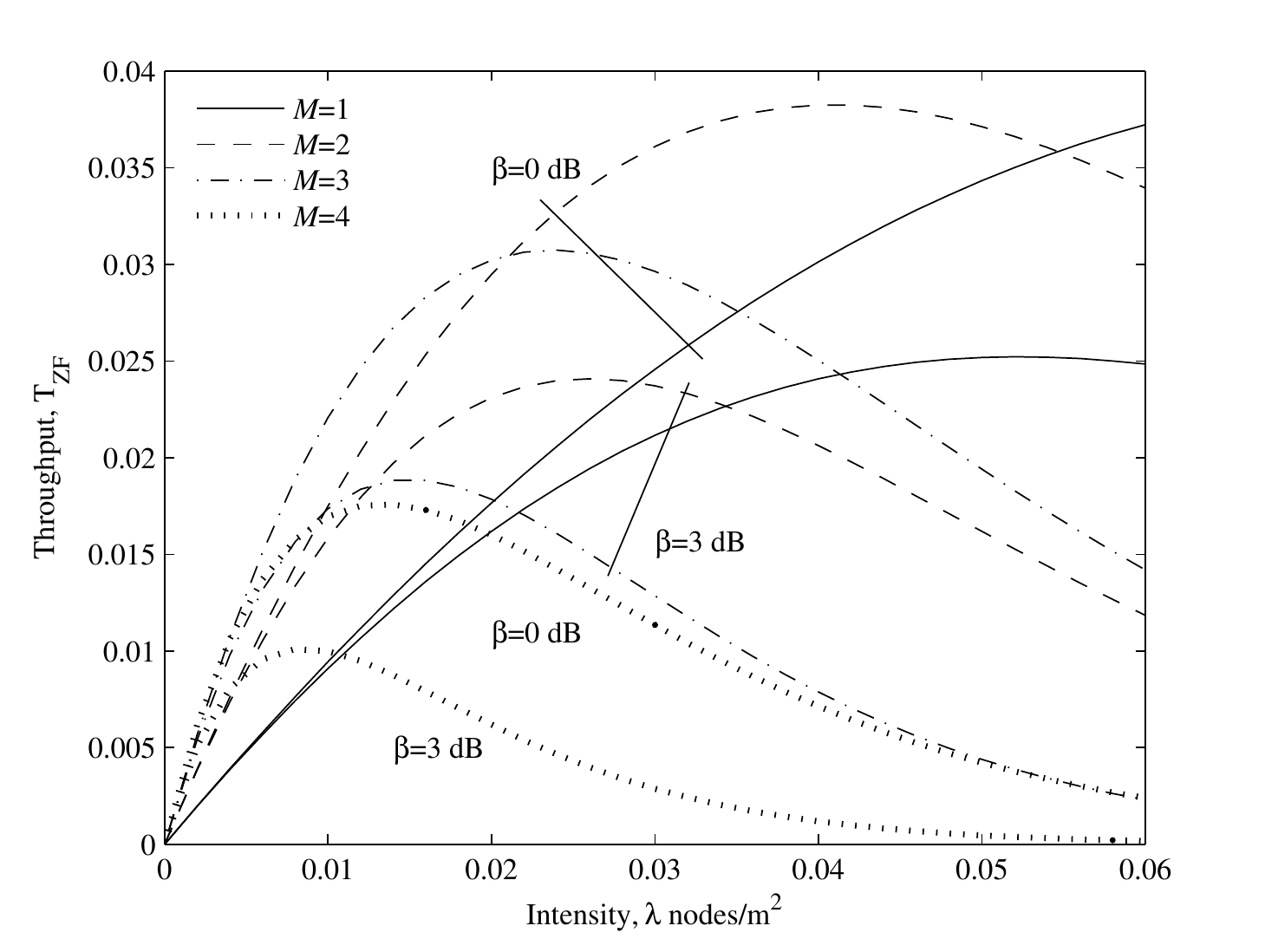}}
\caption{Throughput vs intensity of slotted ALOHA with spatial multiplexing and ZF receivers, and with $N=4$, $\alpha=3.1$, ${r_{\rm tr}}=2$ m, $\rho=25$ dB and $p=1$.}
\label{fig:zf_throughput_changingM}
\end{figure}

\subsubsection{Special Case: Spatial Multiplexing with ZF ($M = N$)}

For this special case, as evident from
(\ref{eq:cdf_general_simple}), the exact outage probability and
throughput expressions reduce to very simple forms, and we can gain
some interesting analytical insights.  In particular, taking the
derivative of ${\rm T}_{\rm ZF}$, we find that with all other
parameters fixed, the throughput is maximized if $\lambda$ is chosen
to satisfy:
\begin{align}\label{eq:cond_ZF_lambdaincrease}
\lambda^{\rm opt} = \frac{\Gamma(N)}{\pi p \Gamma\left(N +
\frac{2}{\alpha} \right) \Gamma\left(1-\frac{2}{\alpha} \right)
\beta^{\frac{2}{\alpha}}  r_{\rm tr}^{2} } \; .
\end{align}
This point is analogous to the ``peaks'' identified in Figs.
\ref{fig:mrc_throughput_changingM} and
\ref{fig:zf_throughput_changingM}, and gives the optimal tradeoff in
terms of network interference and spatial multiplexing gain.
As we may expect, we see that this optimal spatial intensity
decreases when $N$ increases (since there is more network
interference, whilst the per-link multiplexing gain is unchanged), or
when $\beta$ or $r_{\rm tr}$ increases (since this puts higher
reliability requirements on the per-link performance, whilst not
changing the network interference). Note also that the probability
of successful transmission, $1-{\rm F}_{\gamma_{\rm ZF}}(\beta)$,
obtained by using $\lambda^{\rm opt}$, is $e^{-1}$. This is
interesting since it coincides precisely with the probability of
successful transmission for \emph{single} antenna systems, presented
in \cite{baccelli06,weber07}.

In addition to considering the optimal $\lambda$, it is also of
interest to study the optimal number of antennas.  In particular, if
all other parameters are kept fixed, then we find that the
throughput is maximized by choosing $N$ as 
\begin{align}
N^{\rm opt} = \max ( \lfloor x\rfloor, 1 )
\end{align}
with $x$ the solution to
\begin{align}\label{eq:cond_ZF_equalant}
\frac{\ln \left(1+\frac{1}{x}\right) \Gamma(x+1)}{\Gamma\left(x +
\frac{2}{\alpha} \right)} = \frac{2 \pi p \lambda
\Gamma\left(1-\frac{2}{\alpha} \right) \beta^{\frac{2}{\alpha}}
r_{\rm tr}^{2} }{\alpha }+ \frac{\beta r_{\rm tr}^\alpha}{\rho} \; .
\end{align}
This is illustrated in Figs.\ \ref{fig:maxantenna_lambda}, which
plots the optimal number of antennas for different node intensities.
We see that $N^{\rm opt}$ is decreasing in $\lambda$, in
line with the observations in Figs.\
\ref{fig:mrc_throughput_changingM} and
\ref{fig:zf_throughput_changingM}.

\begin{figure}[tb!]
\centerline{\includegraphics[width=0.7\columnwidth]{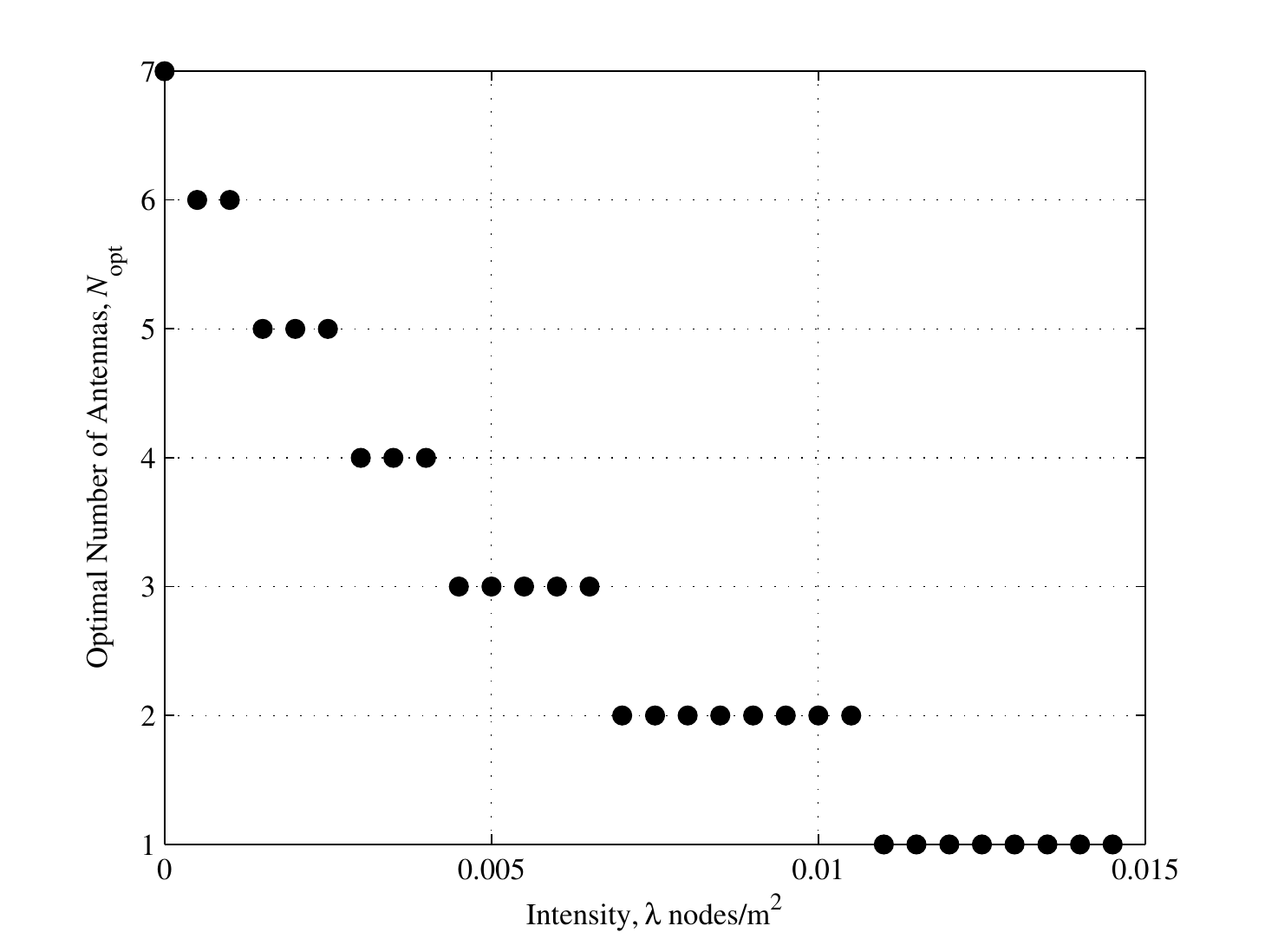}}
\caption{Optimal number of antennas vs.\ intensity with ${r_{\rm tr}}=3$ m, $\rho=25$ dB, $\beta=0$ dB, $\alpha=4$ and $p=1$.}
\label{fig:maxantenna_lambda}
\end{figure}

\subsection{Throughput of OSTBC}

To compute the throughput achieved by OSTBC, we substitute the
relevant parameters from Table \ref{table:gamma_param} into
(\ref{eq:cdf_general}), and substitute the resulting expression into
(\ref{eq:through_def}).  Recall that the parameters from Table
\ref{table:gamma_param} give approximations for the outage
probability and throughput of OSTBC, rather than exact results.

Figs.\ \ref{fig:ostbc_outage_compare_nt4},
\ref{fig:ostbc_outage_compare_nt3}, and
\ref{fig:ostbc_outage_compare_nt4Ns2} plot the outage probability
vs.\ SINR threshold for three different OSTBC codes, with $N=M$.
Fig.\ \ref{fig:ostbc_outage_compare_nt4} is based on the code
\begin{equation}\label{eq:ostbc_code3}
\mathbf{X}_\ell=\left[ \begin{array}{cccc}
x_{\ell,1}  & -x_{\ell,2}^* & x_{\ell,3}^* & 0 \\
x_{\ell,2} & x_{\ell,1}^* & 0 & x_{\ell,3}^* \\
x_{\ell,3} & 0 & -x_{\ell,1}^* & -x_{\ell,2}^* \\
0 & x_{\ell,3} & x_{\ell,2} &  -x_{\ell,1}\end{array} \right] \; ,
\end{equation}
Fig.\ \ref{fig:ostbc_outage_compare_nt3} is based on the code
\begin{equation}\label{eq:ostbc_code2}
 \mathbf{X}_\ell= \left[ \begin{array}{cccc}
x_{\ell,1}  & 0  & x_{\ell,2} & -x_{\ell,3}\\
0 & x_{\ell,1} & x_{\ell,3}^* & x_{\ell,2}^* \\
-x_{\ell,2}^* & -x_{\ell,3} & x_{\ell,1}^* & 0\end{array}  \right]
\; ,
\end{equation}
and Fig. \ref{fig:ostbc_outage_compare_nt4Ns2} is based on the code
\begin{equation}\label{eq:ostbc_code_M4_2}
\mathbf{X}_\ell=\left[ \begin{array}{cccc}
x_{\ell,1}  & -x_{\ell,2}^*  & 0 & 0 \\
x_{\ell,2}  & x_{\ell,1}^*  & 0 & 0 \\
0 & 0 & x_{\ell,1}  & -x_{\ell,2}^*\\
0 & 0 & x_{\ell,2}  & x_{\ell,1}^*\end{array} \right] \; .
\end{equation}
The analytical curves are seen to accurately approximate the Monte
Carlo simulated curves for all SINR thresholds. For further
comparison, the outage probability approximation given in
\cite{hunter07} is also shown. The increased accuracy of our
approximation is clearly evident.

\begin{figure}[tb!]
\centerline{\includegraphics[width=0.7\columnwidth]{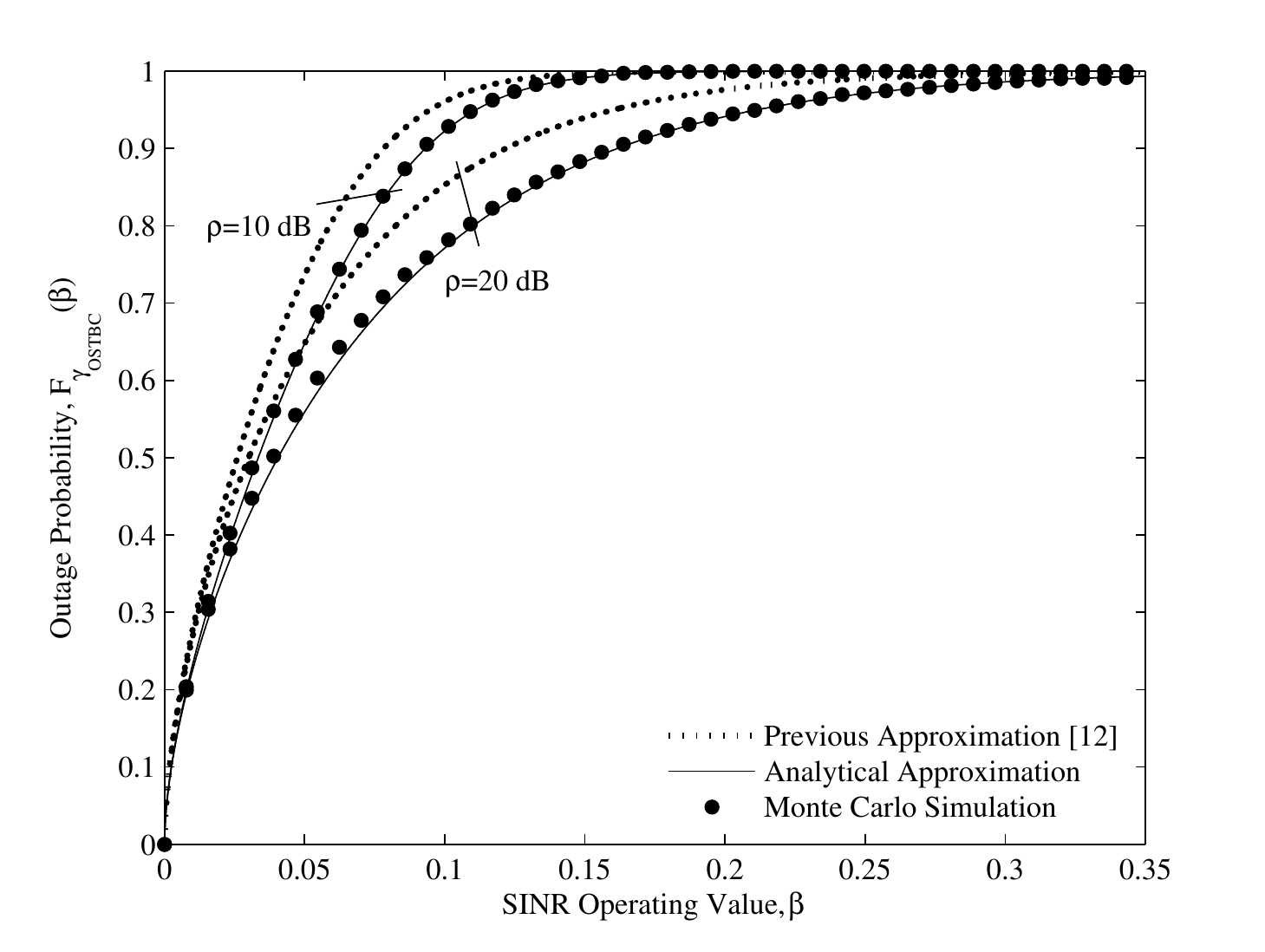}}
\caption{Outage probability vs SINR of slotted ALOHA with OSTBC using the code in (\ref{eq:ostbc_code3}), and with $\alpha=3.5$, ${r_{\rm tr}}=5$ m, $\lambda=0.1$ nodes/${\rm m}^2$, $N=4$ and $p=1$.}
\label{fig:ostbc_outage_compare_nt4}
\end{figure}

\begin{figure}[tb!]
\centerline{\includegraphics[width=0.7\columnwidth]{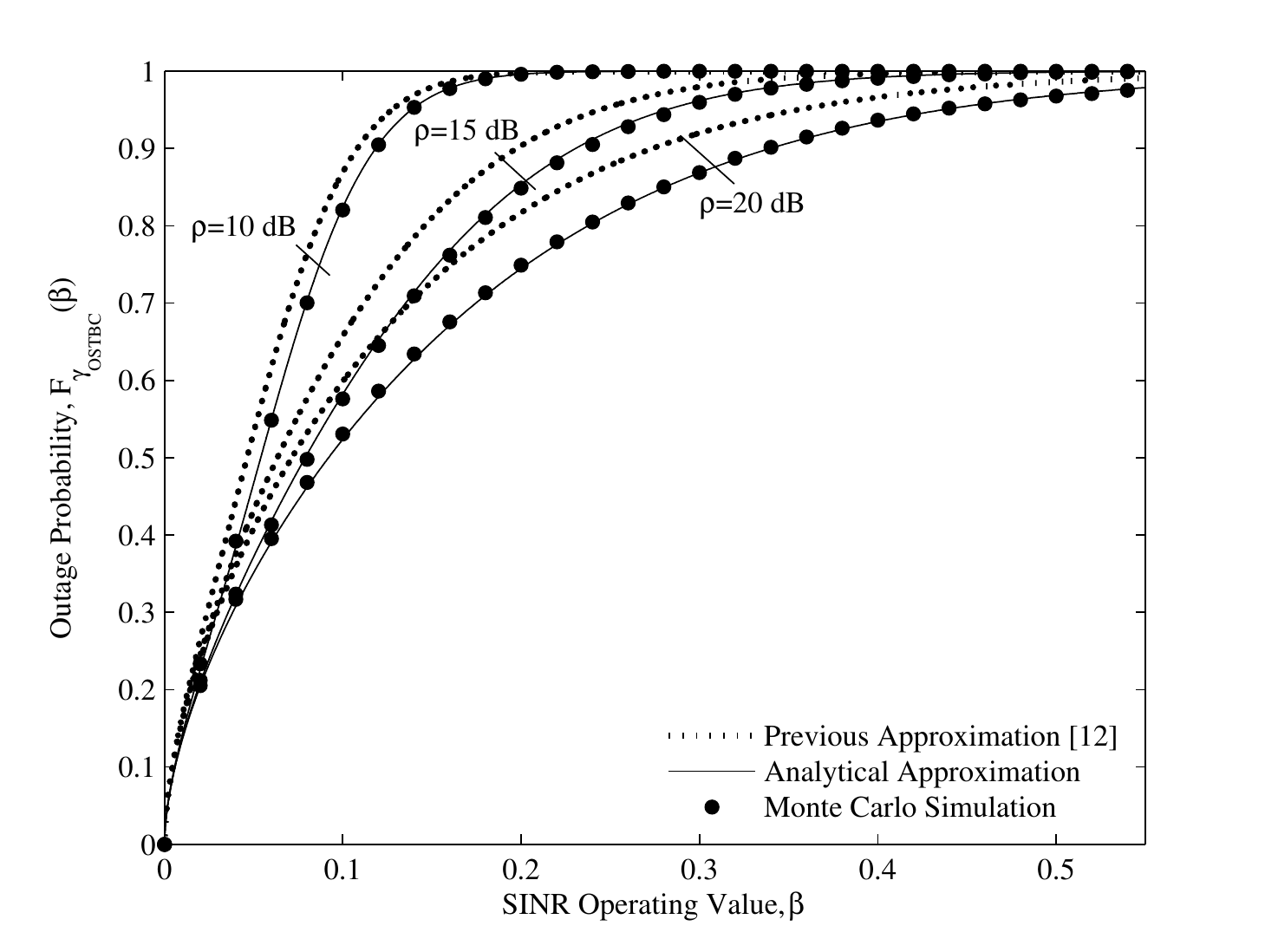}}
\caption{Outage probability vs SINR of slotted ALOHA with OSTBC using the code in (\ref{eq:ostbc_code2}), and with $\alpha=3.5$, ${r_{\rm tr}}=5$ m, $\lambda=0.05$ nodes/${\rm m}^2$, $N=3$ and $p=1$.}
\label{fig:ostbc_outage_compare_nt3}
\end{figure}

\begin{figure}[tb!]
\centerline{\includegraphics[width=0.7\columnwidth]{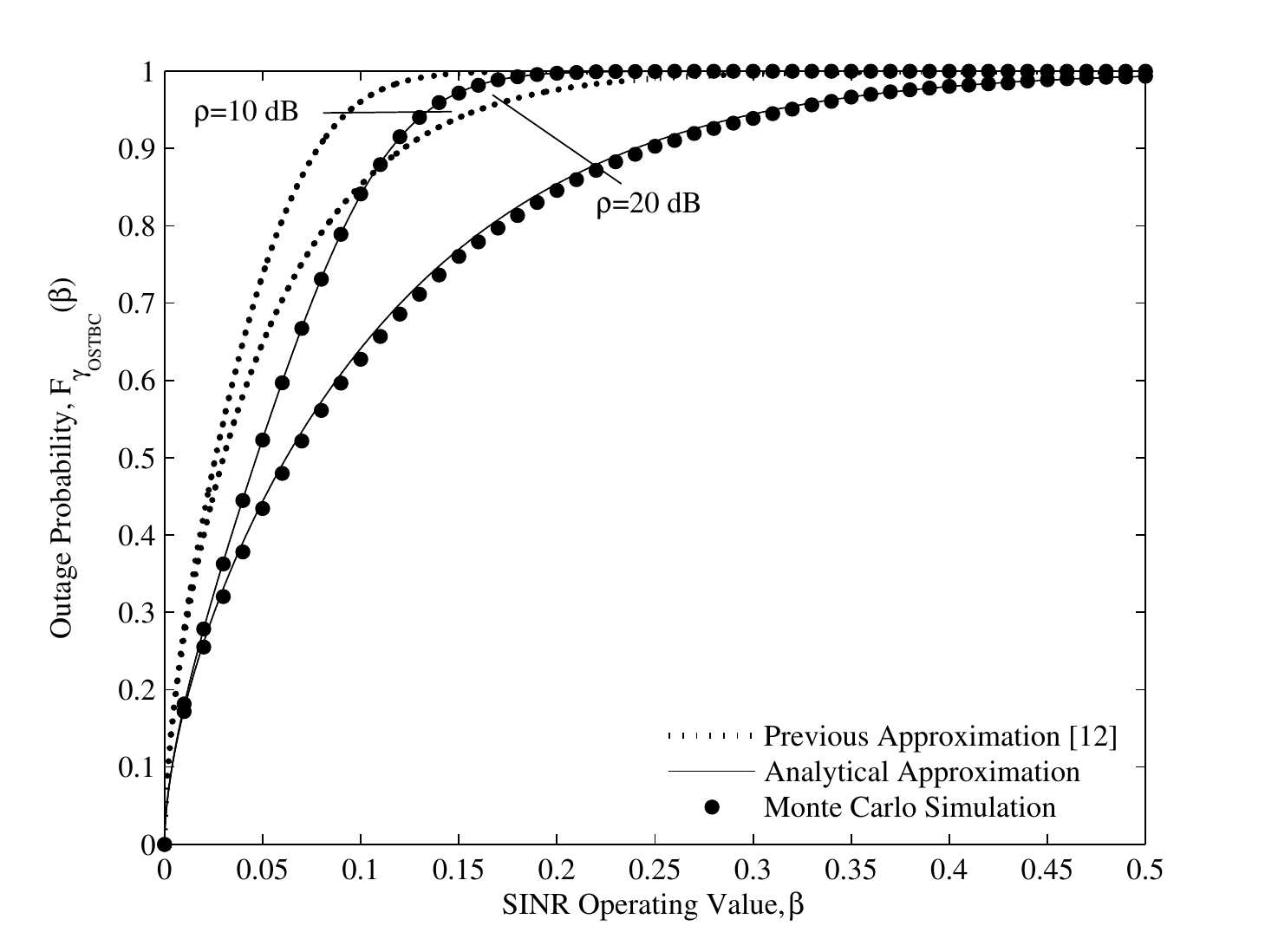}}
\caption{Outage probability vs SINR of slotted ALOHA with OSTBC using the code in (\ref{eq:ostbc_code_M4_2}), and with $\alpha=3.5$, ${r_{\rm tr}}=5$ m, $\lambda=0.1$ nodes/${\rm m}^2$, $N=4$ and $p=1$.}
\label{fig:ostbc_outage_compare_nt4Ns2}
\end{figure}

Note that our outage probability approximation becomes an
\emph{exact} result for the important class of cyclic antenna
diversity codes, which transmit only one symbol per OSTBC codeword,
with the symbol being sent out of a different antenna during each
channel use. For example, for the case $M=4$, the codeword matrix
$\mathbf{X}_\ell$ for this coding scheme would have the form
\begin{align}\label{eq:ostbc_code_M4_1}
\mathbf{X}_\ell = x_{\ell,1} \mathbf{I}_4 \; .
\end{align}
This type of code, whilst achieving full spatial diversity order,
results in the lowest code rate among all OSTBC codes, under the
assumption that at least one symbol is transmitted per time slot.
However, as we show in Section \ref{sec:analysis_compare}, if the
network is sufficiently dense, then this type of coding scheme
becomes \emph{optimal} in terms of maximizing throughput, due to the
minimal network interference it yields compared with other higher
rate codes.

Fig.\ \ref{fig:ostbc_throughput_changingM} shows the throughput achieved by OSTBC, based on (\ref{eq:through_def}), for different SINR operating values $\beta$, and different numbers of antennas used for transmission $M$.  The codes used for $M=2$ and $M=3$ are given by  (\ref{eq:alamouti_param}) and (\ref{eq:ostbc_code2}) respectively. Note that both of these codes correspond to maximum-rate codes for the particular antenna configuration employed \cite{liang03}. As with spatial multiplexing, we see that for all curves, the throughput increases monotonically with $\lambda$ up to a certain peak point, after which the throughput decreases monotonically.

Fig. \ref{fig:ostbc_throughput_changingM} also reveals the interesting fact that the throughput can be significantly higher if less antennas are used for transmission.  This can be explained by first noting that for $M=2$ when using the code in (\ref{eq:alamouti_param}), we have $\frac{N_I}{M}=2$, and for $M=3$ when using the code in (\ref{eq:ostbc_code2}), we have $\frac{N_I}{M}=\frac{7}{3}$ . Hence, whilst increasing $M$ increases the spatial diversity order, it also increases $\frac{N_I}{M}$. (Note that for maximum-rate codes, such as the ones used in Fig.\ \ref{fig:ostbc_throughput_changingM}, it can easily be shown that $N_I$ is always an increasing function of $M$.) Further, by invoking \emph{Lemma \ref{lem:gamma_increase_shape}}, we see that the moments of the approximate normalized interference term $\tilde{\mathcal{K}}_{\ell}$ in (\ref{eq:ostbc_sinr_approximate}) are increasing functions of $\frac{N_I}{M}$. This suggests that the additional interference in the network caused by each transmitting node employing $M = 3$ antennas using the code in (\ref{eq:ostbc_code2}), compared with $M = 2$ antennas using the code in (\ref{eq:alamouti_param}), outweighs the benefits in terms of an increased spatial diversity order.

Fig. \ref{fig:ostbc_throughput_changingNI} shows the throughput achieved by OSTBC, based on (\ref{eq:through_def}), for different SINR operating values $\beta$, and different
values of $N_I$. In particular, for $N_I=4$, $N_I=8$, and $N_I=12$, we use the codes in (\ref{eq:ostbc_code_M4_1}), (\ref{eq:ostbc_code_M4_2}), and (\ref{eq:ostbc_code3}) respectively. Note that each of these codes have the same diversity order, but different code rates $R$.  Specifically, for the case $N_I=4$, $R=\frac{1}{4}$; for the case $N_I=8$, $R=\frac{1}{2}$; and for the case $N_I=12$, $R=\frac{3}{4}$. Clearly the code rate is an increasing function of $N_I$.  Fig.\ \ref{fig:ostbc_throughput_changingNI} reveals the interesting fact that in many cases, the throughput is higher when a lower code rate is used. This is because the additional interference caused by more simultaneous symbol transmissions in the network outweighs the benefits of an increased code rate.

\begin{figure}[tb!]
\centerline{\includegraphics[width=0.7\columnwidth]{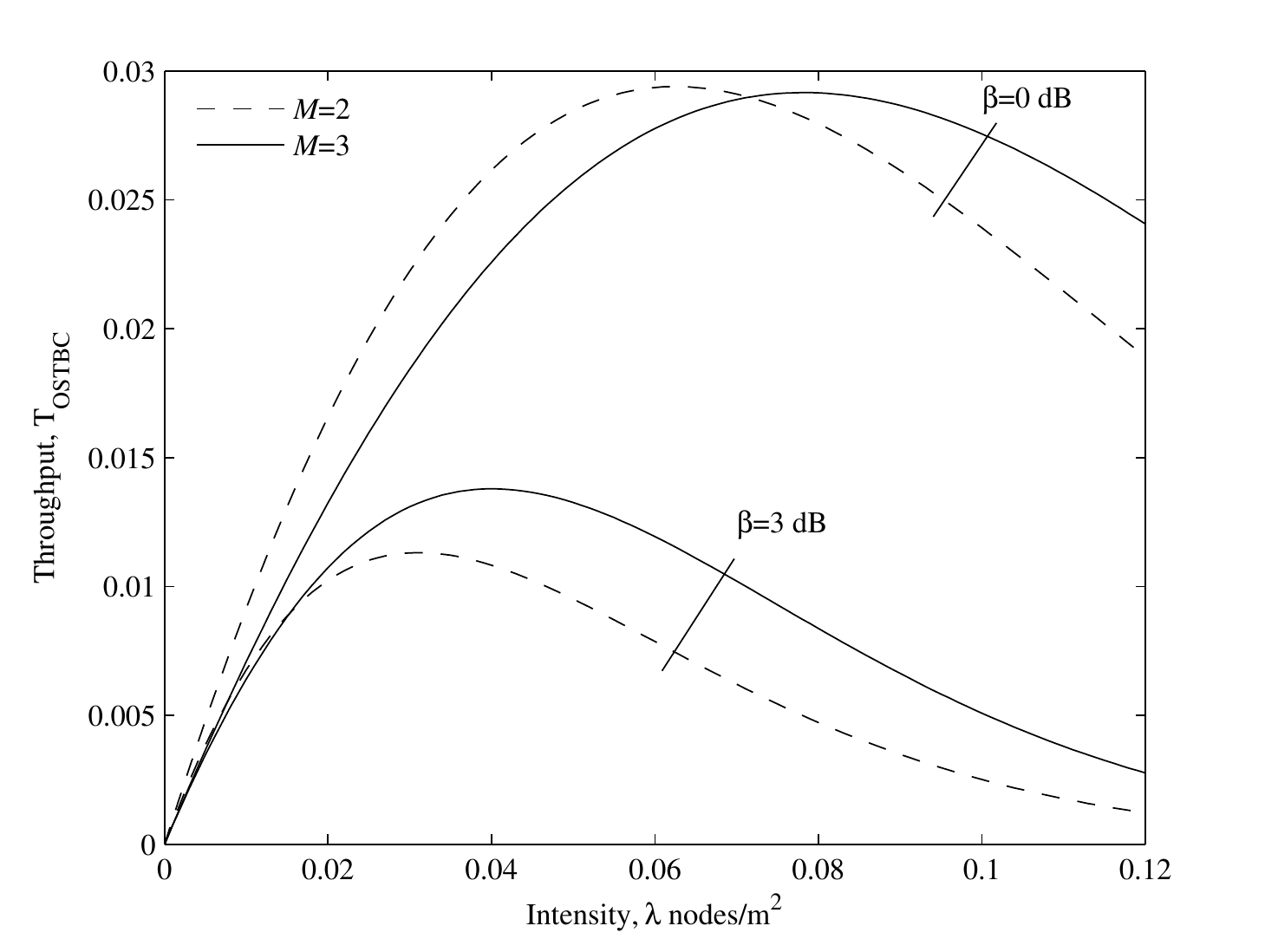}}
\caption{Throughput vs intensity of slotted ALOHA with OSTBC, and with $N=4$, $\alpha=3.1$, ${r_{\rm tr}}=2$ m, $\rho=25$ dB and $p=1$.}
\label{fig:ostbc_throughput_changingM}
\end{figure}

\begin{figure}[tb!]
\centerline{\includegraphics[width=0.7\columnwidth]{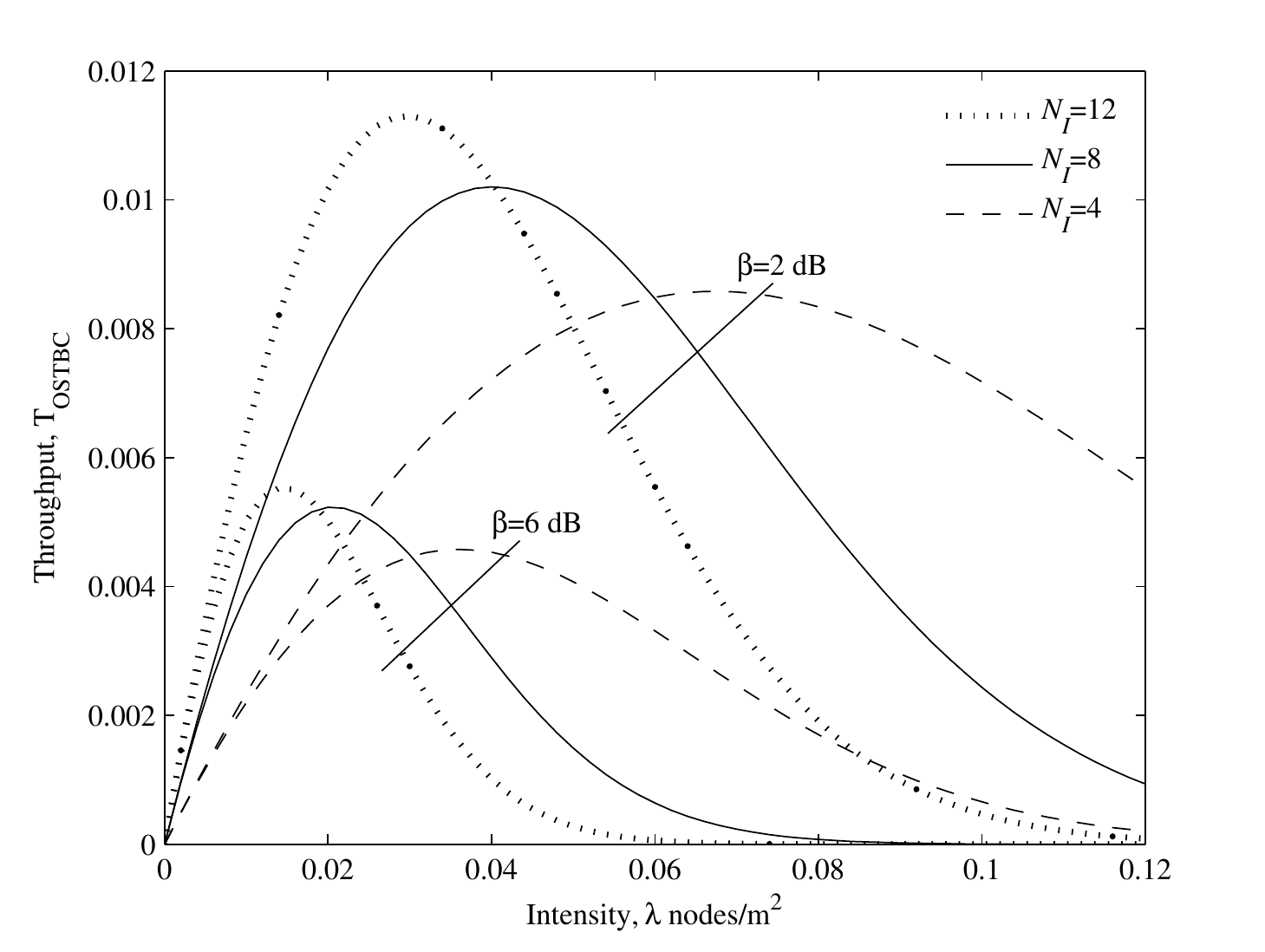}}
\caption{Throughput vs intensity of slotted ALOHA with OSTBC, and with $N=4$, $\alpha=3.1$, ${r_{\rm tr}}=3$ m, $\rho=20$ dB and $p=1$.}
\label{fig:ostbc_throughput_changingNI}
\end{figure}

\subsection{Comparison}

We see in Table \ref{table:gamma_param} that the signal and interference powers for the spatial multiplexing and OSTBC
schemes are different, and are dependent on the system parameters $M$, $N$, $R$ and $N_I$. As such, it is not straightforward to determine which scheme performs the best, and under which scenario. For example, when $M=2$, it can be shown that the signal power of OSTBC using the Alamouti code is greater than spatial multiplexing with ZF receivers, while the interference power of both OSTBC and spatial multiplexing schemes are the same. However, for a fixed $\beta$, the per-link data rate for spatial multiplexing is twice that of OSTBC, due to the $M=2$ transmitted streams per node for spatial multiplexing. Thus although the SINR of OSTBC is greater than spatial multiplexing with ZF receivers for $M=2$, the overall network throughput may not be greater.

The best performing receiver used for spatial multiplexing is also dependent on the particular scenario.  For example, it is well known that for a single user MIMO system with no interfering nodes, spatial multiplexing with ZF receivers performs better than spatial multiplexing with MRC receivers at high SNR, hence ZF receivers is expected to perform better than MRC receivers in sparse networks. However, for MRC receivers, as the node density increases, the impact of the interference from other transmitting nodes becomes more dominant than the self-interference. By noting that the distribution of the interference power from these interfering transmitting nodes for both MRC and ZF receivers are equal, the distribution of the \emph{total interference} for both MRC and ZF receivers thus converge with increasing node density. The key factor which distinguishes between these two schemes in dense networks is the signal power, which is greater for MRC receivers. This suggests that the best performing receiver used for spatial multiplexing is dependent on the particular node density.

The preceding discussion suggests that the best performing scheme is dependent on the system and network parameters. This can be seen in Fig.\ \ref{fig:compare_throughput_changingN}, which plots the network
throughput vs.\ node intensity $\lambda$, of spatial multiplexing
with MRC and ZF receivers, and OSTBC, based on (\ref{eq:through_def}), for different numbers
of antennas used for transmission, $M$. We observe that spatial
multiplexing with ZF performs the best in sparse network
configurations (i.e., small $\lambda$), while OSTBC performs the
best in dense networks (i.e., high $\lambda$). Further, spatial
multiplexing with MRC receivers is seen to perform better than ZF
receivers at high $\lambda$, which agrees with previous discussion. We will also prove this analytically in Section \ref{sec:analysis_compare}.

\begin{figure}[tb!]
\centerline{\includegraphics[width=0.7\columnwidth]{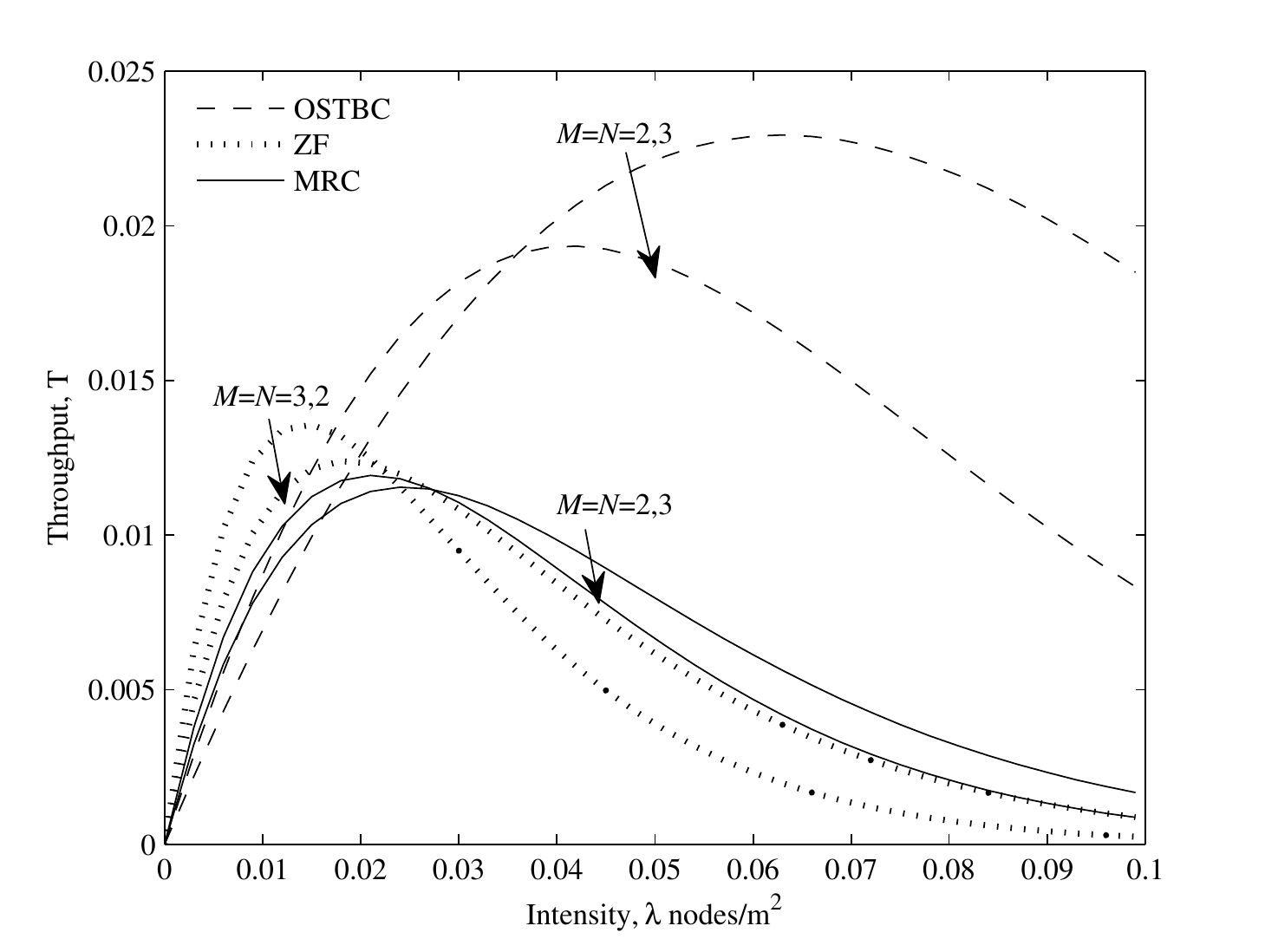}}
\caption{Throughput vs intensity of slotted ALOHA with spatial multiplexing and OSTBC, and with $\alpha=3.3$, ${r_{\rm tr}}=2$ m, $\rho=25$ dB, $\beta=2$ dB and $p=1$.}
\label{fig:compare_throughput_changingN}
\end{figure}

The best performing scheme can also be shown to be dependent on the SINR $\beta$ operating value, as shown in Fig.\ \ref{fig:compare_throughput_changeSINR}, which plots the network
throughput vs.\ SINR $\beta$ for all three schemes, again
for different $M$. We see that spatial multiplexing with MRC
receivers performs the best for low $\beta$, while OSTBC performs
the best for high $\beta$.

\begin{figure}[tb!]
\centerline{\includegraphics[width=0.7\columnwidth]{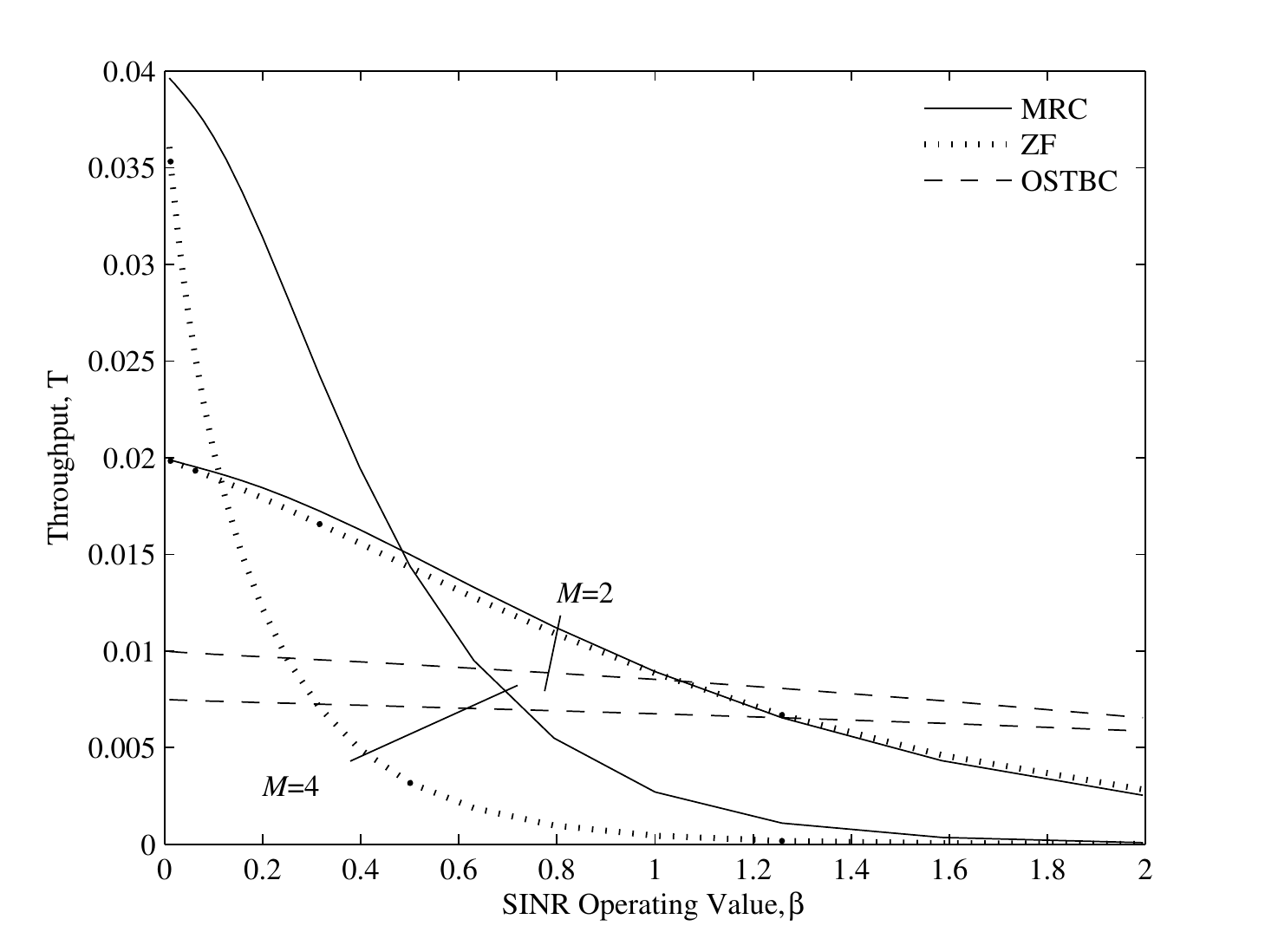}}
\caption{Throughput vs SINR operating value $\beta$ of slotted ALOHA with spatial multiplexing and OSTBC, and with $N=4$, $\alpha=2.5$, ${r_{\rm tr}}=3$ m, $\rho=20$ dB, $\lambda=0.01$ nodes/${\rm m}^2$ and $p=1$.}
\label{fig:compare_throughput_changeSINR}
\end{figure}

As can be seen from Figs.
\ref{fig:mrc_throughput_changingM}--\ref{fig:compare_throughput_changeSINR},
the impact of the number of antennas used for transmission, and the
relative throughput of the spatial multiplexing and OSTBC systems
are dependent on different network parameters. To explore this
further, it is convenient to analyze the throughput of spatial
multiplexing and OSTBC in asymptotic regimes.  This is considered in
the following section.

\section{Network Throughput: Asymptotic Analysis}\label{sec:analysis_compare}

In this section, we analyze the performance of spatial multiplexing
and OSTBC in dense networks, and for low and high $\beta$ operating
values. We note that in sparse networks, the performance of spatial
multiplexing and OSTBC approaches the performance of single user MIMO
systems, for which the performance is well
known (see e.g.,\ \cite{tse05,louie08,larsson03,forenza06b}). Thus,
although we investigate dense network scenarios, we do not consider
the opposite case of sparse networks in this paper. Note also, that for $M=1$, corresponding to single-input multiple-output (SIMO) transmission, the performance of all schemes are the same.

\subsection{Dense Networks (Large $\lambda$)}

Under dense network conditions, we have the following corollary.
  \begin{corollary}
As $\lambda \to \infty$, the SINR
c.d.f.\ in (\ref{eq:cdf_general}) behaves as\footnote{The notation
$f(x) =o( g(x))$ as $x \to \infty$ means that for every $M>0$, there exists a constant $x_0$ such that $|f(x)| \le M |g(x)|$ for all $x > x_0$.}
\begin{align}\label{eq:cdf_general_dense}
{\rm F}_{\gamma}(\beta)  =  {\rm F}_{\gamma,\lambda \to \infty}(\beta) + o\left(e^{- \lambda \left( \frac{\beta
\Omega}{\theta}\right)^{\frac{2}{\alpha}} \eta(n) }
  \lambda^{m-1}\right)
\end{align}
where
\begin{align}\label{eq:cdf_general_dense2}
{\rm F}_{\gamma,\lambda \to \infty}(\beta)  = 1 - \frac{2 \eta (n)}{\alpha \Gamma(m)}
\left( \frac{ \beta \Omega}{\theta} \right)^{\frac{2}{\alpha}}
  {\rm E}_Y\left[e^{-\frac{ \beta Y}{\theta}} \right]
e^{-\frac{\beta}{\theta}} e^{- \lambda \left( \frac{\beta
\Omega}{\theta}\right)^{\frac{2}{\alpha}} \eta(n) }
  \lambda^{m-1} \notag
\end{align}
and
\begin{align}
{\rm E}_Y\left[e^{-\frac{ \beta Y}{\theta}} \right] 
&=  \left(\frac{\theta}{\Upsilon \beta +\theta}\right)^{u} \; .
\end{align}
\end{corollary}

To give an indication as to when a network is ``sufficiently dense''
such that the expansion (\ref{eq:cdf_general_dense2}) is accurate,
Table \ref{table:dense_verify} tabulates the quantity
$1/\lambda_{\rm min}$, with $\lambda_{\min}$ representing the
minimum node density required such that (\ref{eq:cdf_general_dense2})
is within at least $85 \%$ of the true non-asymptotic value. Intuitively, the
quantity $1/\lambda_{\rm min}$ gives a measure of the maximum
allowable separation between adjacent transmitting nodes (on
average), in order for the asymptotic expansion
(\ref{eq:cdf_general_dense2}) to serve as a good approximation. Note
that a scenario with $M = 1$ data streams is chosen, because as we will discuss later, single-stream transmission is throughput-optimal in dense networks. Moreover, a
transmit-receive distance of $5$ ${\rm m}$ is chosen, which is practically
relevant (e.g., for wireless local area networks).

\begin{table*}[!t]
\caption{Average maximum area around each transmitting node, such that $\frac{\biggr|{\rm F}_{\gamma,\lambda \to \infty}(\beta)-{\rm F}_{\gamma}(\beta)\biggr|}{1-{\rm F}_{\gamma}(\beta)}>0.15$, for various SINR thresholds $\beta$ and transmit SNRs $\rho$ with $M=1$, $N=3$, $r_{\rm tr}=5$ m and $\alpha=2.1$.}\label{table:dense_verify}
\centering
\begin{tabular}{|c|c|c|c|c|c|}
  \hline
 $\rho$ dB  $\setminus \beta$ dB & $1$  & $2$  & $5$  &  $10$    &  $20$   \\ \hline
$2$ & 6.500&  6.486&  6.383&  6.114&  5.506\\
$5$ & 12.447&  12.519&  12.520&  12.133&  10.981\\
$10$   & 33.478&  34.674&  36.900&  37.481&  34.638\\
$20$   & 111.483 &  130.208&  193.050&  286.532&  335.569 \\
  \hline
\end{tabular}
\end{table*}

The results in the table show that the large--$\lambda$ expansion
(\ref{eq:cdf_general_dense2}) serves as an accurate performance
measure for practical network configurations; in some cases,
applying even when the networks are \emph{relatively sparse}.  This
is particularly true for moderate to large SNRs. For example, at
$10$ dB SNR, a transmitter spacing of roughly $35$ ${\rm m}^2$ on
average is sufficient, which is relatively large compared with the
transmit-receive distance of $5$ m.  As the SNR is reduced, e.g., to
$2$ dB, a closer transmitter spacing of roughly $6$ ${\rm m}^2$ is
needed.  This behavior is intuitive, since dense networking
conditions are representative of interference-limited scenarios, in
which case the interference in the network is much more significant
compared with the noise.  If the SNR is reduced, then the
noise has greater relative effect, and there must be more
interference (i.e., a greater density $\lambda$) in order for the
interference-limited behavior to be apparent.

From (\ref{eq:through_def}) and
(\ref{eq:cdf_general_dense}), and recalling that $\Omega / \theta =
r_{\rm tr}^{\alpha}$ for each scheme (c.f.\ Table
\ref{table:gamma_param}), it follows that for large $\lambda$ the
throughput becomes
\begin{align}\label{eq:through_dense_general}
{\rm T} \propto e^{- \lambda \beta^{\frac{2}{\alpha}} r_{\rm tr}^2
\eta(n) } \;  \lambda^{m} \; .
\end{align}
From this, we can obtain some useful insights into the network
performance and optimization:

\begin{itemize}

\item The throughput decays exponentially in $\lambda$. This loss of throughput indicates that in dense networks, the negative effects of interference will dominate any positive throughput gains obtained by an increase in the number of communication links.
\item Since the exponential in (\ref{eq:through_dense_general}) dominates for large
$\lambda$, the component $\eta(n)$ is a critical factor which
determines performance. Moreover, since $n$ is proportional to the
effective interference power caused by each interfering transmitter,
and $\eta(n)$ increases with $n$, for all three transmission schemes
it is best to choose $n$ as small as possible. Thus, recalling the
parameters in Table \ref{table:gamma_param}, we have the following
design criteria for dense networks:
\begin{itemize}
\item For spatial multiplexing with either MRC or ZF receivers, it is optimal to use only a \emph{single transmit stream} (i.e., $M=1$).
\item For OSTBC, it is optimal to use a \emph{cyclic antenna diversity} coding scheme.
This code, illustrated in (\ref{eq:ostbc_code_M4_1}) for $M=4$,  minimizes $n$
by maximizing the coding parameter $N_I$ (i.e., $N_I=M$). To determine the optimal $M$ for cyclic antenna diversity codes, we consider the leading order factor $\lambda^m$. As $m$ is proportional to the spatial diversity order, we note that the optimal choice is $M=N$.
\end{itemize}
\item For systems with\footnote{Although $M > 1$ is sub-optimal for dense
networks, for other networking scenarios this configuration will
become important.  Thus, it is still of interest to study the
throughput of such configurations under dense network conditions.}
$M
> 1$, OSTBC codes with $N_I < M^2$ will
yield a higher throughput than both spatial multiplexing schemes;
whereas, if $N_I > M^2$, the throughput will be worse than both.
This can be seen by again noting that the exponential in
(\ref{eq:through_dense_general}) dominates for large $\lambda$. As $N_I=M$ for cyclic antenna diversity code, we see that these codes always perform better than spatial multiplexing. For
the remaining case $N_I = M^2$, which only occurs for the Alamouti
code and therefore $M=2$, the throughputs of all three schemes have
the same exponential decay in (\ref{eq:through_dense_general});
however, OSTBC has the largest leading order factor $\lambda^m$, and
therefore achieves the highest throughput. Finally, for all $M > 1$,
spatial multiplexing with MRC achieves a higher throughput than ZF.
\end{itemize}
In general, these results reveal that OSTBC with a $M=N$ cyclic antenna diversity code is the optimal scheme in dense networks.

\subsection{Networks with High SINR Operating Values (Large $\beta$)}

For networks with high SINR operating values, we have the following
corollary.
\begin{corollary}
As $\beta \to \infty$, the SINR
c.d.f.\ (\ref{eq:cdf_general}) behaves as
\begin{align}\label{eq:cdf_general_highbeta}
 {\rm F}_{\gamma}(\beta)   &=  {\rm F}_{\gamma,\beta \to \infty}(\beta) + o\left(  {\rm E}_Y\left[e^{-\frac{ \beta Y}{\theta}} \right]
e^{-\frac{\beta}{\theta}} e^{- \lambda \left( \frac{\beta
\Omega}{\theta}\right)^{\frac{2}{\alpha}} \eta(n) } \beta^{m-1} \right)
 \end{align}
 where
 \begin{align}\label{eq:cdf_general_highbeta2}
{\rm F}_{\gamma,\beta \to \infty}(\beta)  = 1 - \frac{1}{\Gamma(m) \theta^{m-1}} {\rm E}_Y\left[e^{-\frac{ \beta Y}{\theta}} \right] e^{-\frac{\beta}{\theta}} e^{- \lambda \left( \frac{\beta \Omega}{\theta}\right)^{\frac{2}{\alpha}} \eta(n) } \beta^{m-1}  \; .
\end{align}
\end{corollary}
To check the accuracy of this expansion, Table
\ref{table:highbeta_verify} tabulates for various scenarios the
minimum SINR operating value $\beta$ such that
(\ref{eq:cdf_general_highbeta2}) is within $85 \%$ of the true non-asymptotic value.
These results demonstrate that the large--$\beta$ expansion
(\ref{eq:cdf_general_highbeta2}) serves as an accurate performance
measure for practical network configurations; in some cases,
applying even when $\beta$ is low. This is particularly true for low
to moderate SNRs. For example, at $2$ dB SNR and $\lambda$ chosen
such that the average transmitter spacing is $20$ ${\rm m}^2$, a
SINR operating value of $-12.832$ dB or above is
sufficient---something which is expected to be true for most
wireless applications. As the SNR is increased, e.g., to $10$ dB, a
higher SINR operating value of $-1.726$ dB or above is needed. This
matches with intuition, since, if the network is relatively sparse
(as in the example above), then for the outage probability to remain
constant as the SINR is increased, the SNR must increase
accordingly.

\begin{table*}[t!]
\caption{Minimum SINR $\beta$ (dB) operating value such that $\frac{\biggr|{\rm F}_{\gamma,\beta \to \infty}(\beta)-{\rm F}_{\gamma}(\beta)\biggr|}{1-{\rm F}_{\gamma}(\beta)}>0.15$, for various $1/\lambda$, and transmit SNRs $\rho$ with $M=1$, $N=3$, $r_{\rm tr}=5$ m and $\alpha=4$.}\label{table:highbeta_verify}
\centering
\begin{tabular}{|c|c|c|c|c|c|}
  \hline
 $\rho$ dB  $\setminus \frac{1}{\lambda} {\rm m}^2$ & $100$  & $50$  & $20$  &  $12.5$    &  $10$   \\ \hline
$2$ & -14.547 & -14.078  &  -12.832 &  -11.543 &  -10.701\\
$5$ & -11.421 &  -10.804 &  -8.925 &  -7.211 &  -6.142 \\
$10$ & -5.933 &  -4.827 &  -1.726 &  0.885 &  2.366\\
$15$   & -0.078 &  1.833 &  6.569 &  9.948 &  11.699 \\
  \hline
\end{tabular}
\end{table*}

From (\ref{eq:through_def}) and
(\ref{eq:cdf_general_highbeta}), and recalling that $\Omega / \theta =
r_{\rm tr}^{\alpha}$ for each scheme (c.f.\ Table
\ref{table:gamma_param}), it follows that for large $\beta$ the
throughput becomes
\begin{align}\label{eq:through_highbeta_general}
{\rm T} \propto  \frac{1}{\Gamma(m) \theta^{m-1} \Upsilon^u} e^{-\frac{\beta}{\theta}} e^{- \lambda r_{\rm tr}^2 \beta
^{\frac{2}{\alpha}} \eta(n) } \beta^{m-1-u} \; .
\end{align}
From this, we can obtain some useful insights into the network
performance and optimization:

\begin{itemize}
\item Recalling that $\alpha > 2$, we see that the exponential $e^{-\frac{\beta}{\theta}}$ in (\ref{eq:through_highbeta_general}) dominates for large
$\beta$.  This is intuitive, by recalling the direct correspondence
between the SINR operating value $\beta$ and the outage probability.
\item The component $\theta$, which represents the transmit SNR, is a critical factor, and this
quantity should be maximized in order to maximize the throughput. When $\theta$ is independent of $M$, as is the case for cyclic antenna diversity codes, the $\beta^{m-1-u} $ polynomial term should be maximized in order to maximize the throughput.
Recalling the parameters from Table \ref{table:gamma_param},
we arrive at the following design criteria for networks with high
SINR operating values:
\begin{itemize}
\item The optimal transmission scheme is the \emph{same as for dense networks}. That is,
for spatial multiplexing with either MRC or ZF receivers it is
optimal to use only a single transmit stream, whereas for OSTBC it
is optimal to use the $M=N$  cyclic antenna diversity coding scheme.
\end{itemize}
Note that if the SNR is also sufficiently high, then the exponential
$e^{- \lambda r_{\rm tr}^2 \beta ^{\frac{2}{\alpha}} \eta(n) }$ may
dominate $e^{-\frac{\beta}{\theta}}$; however, it is easy to see
that the same optimality criteria still applies.
\item For systems with $M > 1$, OSTBC codes yield a higher throughput than each of the
spatial multiplexing schemes. Moreover, since the throughput of both
spatial multiplexing schemes have the same exponential decay and
also the same polynomial factor $\beta^{m-1-u} $ in
(\ref{eq:through_highbeta_general}), their relative performance is
determined by the constant factor $\frac{1}{\Gamma(m) \theta^{m-1}
\Upsilon^u}$. Thus, by substituting the relevant parameters from
Table \ref{table:gamma_param}, we can show that ZF will deliver a
higher throughput than MRC if the following condition is met:
\begin{align}\label{eq:zf_mrc_highbeta}
\rho  \ge M r_{\rm tr}^\alpha
\left(\frac{\Gamma(N-M+1)}{\Gamma(N)}\right)^{\frac{1}{M-1}} ,
\end{align}
otherwise MRC will perform better.
\end{itemize}
In general, these results indicate that, as for dense networks, OSTBC with a $M=N$ cyclic antenna diversity code is the preferable scheme for high SINR operating values.

\subsection{Networks with Low SINR Operating Values (Small $\beta$)}

For networks with low SINR operating values, we have the following
corollary.
\begin{corollary}
As $\beta \to 0^+$, the SINR c.d.f.\ (\ref{eq:cdf_general}) behaves
as\footnote{The notation
$f(x) =O( g(x))$ as $x \to 0^{+}$ means there exists positive numbers $\delta$ and $M$ such that $|f(x)| \le M |g(x)|$ for $|x|< \delta$.}
\begin{align}\label{eq:cdf_general_lowbeta}
{\rm F}_{\gamma}(\beta)  &=     {\rm F}_{\gamma,\beta \to 0^{+}}(\beta) + O\left(\beta\right)
\end{align}
where
\begin{align}\label{eq:cdf_general_lowbeta2}
{\rm F}_{\gamma,\beta \to 0^{+}}(\beta)=  \lambda \frac{
\Gamma\left(m-\frac{2}{\alpha}\right)}{ \Gamma(m) }  \frac{ \eta (n)
}{\Gamma\left(1-\frac{2}{\alpha}\right) }  \left( \frac{\beta
\Omega}{\theta}\right)^{\frac{2}{\alpha}} \; .
\end{align}
\end{corollary}
To give an indication of when the expansion
(\ref{eq:cdf_general_lowbeta2}) is accurate, Table
\ref{table:lowbeta_verify} tabulates the maximum SINR operating
values such that (\ref{eq:cdf_general_lowbeta2}) is within $85 \%$ of
the true non-asymptotic value. From the table, we see that quite low SINR operating
values are required, and these may or may not be practical.
Nevertheless, as we discuss shortly, the insights we will obtain
from (\ref{eq:cdf_general_lowbeta2}) also turn out to be valid for
much higher SINR operating values than those tabulated in Table
\ref{table:lowbeta_verify}.

\begin{table*}[t!]
\caption{Maximum SINR $\beta$ (dB) operating value such that $\frac{\biggr|{\rm F}_{\gamma,\beta \to 0^{+}}(\beta)-{\rm F}_{\gamma}(\beta)\biggr|}{1-{\rm F}_{\gamma}(\beta)}>0.15$, for various $1/\lambda$ and transmit SNRs $\rho$. Spatial multiplexing with MRC receivers is considered with $M=4$, $N=4$, $r_{\rm tr}=5$ m and $\alpha=3$.}\label{table:lowbeta_verify}
\centering
\begin{tabular}{|c|c|c|c|c|c|}
  \hline
 $\rho$ dB  $\setminus \frac{1}{\lambda} {\rm m}^2$ & $100$  & $50$  & $20$  &  $12.5$    &  $10$   \\ \hline
$10$ & -22.055 & -22.660  &  -24.559 &  -27.033 &  -30.362 \\
$15$ & -18.164 &  -19.462 &  -22.518 &  -25.850 &  -29.788 \\
$20$ & -15.476 &  -17.467 &  -21.487 &  -25.346 &  -29.586\\
$25$   & -14.0351 &  -16.501 &  -21.068 &  -25.171 &  -29.547 \\
$30$   & -13.435 &  -16.128 &  -20.921 &  -25.100 &  -29.508 \\
  \hline
\end{tabular}
\end{table*}

From (\ref{eq:through_def}) and
(\ref{eq:cdf_general_lowbeta}), and recalling that $\Omega / \theta =
r_{\rm tr}^{\alpha}$ for each scheme (c.f.\ Table
\ref{table:gamma_param}), it follows that for small $\beta$ the
throughput becomes
\begin{align}\label{eq:through_lowbeta_general}
{\rm T} 
  \propto \zeta p \lambda\left(1- \frac{
\Gamma\left(m-\frac{2}{\alpha}\right)}{ \Gamma(m) }  \frac{ \eta (n)
}{\Gamma\left(1-\frac{2}{\alpha}\right) }  \lambda  r_{\rm tr}^2
\beta^{\frac{2}{\alpha}}\right) \; .
\end{align}

From the preceding equations, we can obtain the following useful
insights:
\begin{itemize}

\item If $\beta$ is ``very'' small (i.e., such that the $\beta$-dependent term in
(\ref{eq:through_lowbeta_general}) is negligible), then the
throughput approaches $\zeta \lambda p$, which trivially gives the
following design criteria:
\begin{itemize}
\item For spatial multiplexing with either MRC or ZF receivers, it is optimal
to use the \emph{maximum number of transmit streams}.
\item For OSTBC, it is optimal to use maximum-rate codes.
This follows by noting that the maximum code rate is a decreasing function of the number of
antennas used for transmission when $M \ge 2$ \cite{liang03}, and thus increasing the number of antennas for
transmission leads to a lower throughput. The codes with the maximum-rate corresponds to either SIMO transmission with $M=1$ or the Alamouti code with $M=2$. From (\ref{eq:cdf_general_lowbeta}), the Alamouti code performs better than SIMO if
\begin{align}
\frac{\Gamma\left(N-\frac{2}{\alpha}\right)\Gamma(2N)}{\Gamma(N)\Gamma\left(2N-\frac{2}{\alpha}\right)} > \left(1+\frac{2}{\alpha}\right) \; .
\end{align}
We observe that this occurs for low path loss exponents, i.e., as $\alpha \to 2$.
\end{itemize}
Intuitively, the outage probability is very close to zero, and
therefore it makes sense to send as much information as possible
during each channel use. For this scenario, spatial multiplexing thus performs better than OSTBC.
\item  More generally, for all $\beta$ values for which (\ref{eq:through_lowbeta_general}) is
accurate (i.e., the $\beta$-dependent term in
(\ref{eq:through_lowbeta_general}) is not necessarily negligible),
we can derive the following conditions:
\begin{itemize}
\item  For spatial multiplexing with MRC and ZF receivers, the optimal number of data streams is given by
\begin{align} \label{eq:MoptDefn}
M^{\rm opt} = \min\left(\max\left(\lfloor x\rfloor,1\right),N\right)
\,
\end{align}
where, for MRC, $x$ is the solution to
 \begin{align}\label{eq:MRC_cond_smallbeta}
 \frac{ 1}{ \lambda \pi p r_{\rm tr}^4 \beta
^{\frac{2}{\alpha}} } &=\frac{\Gamma\left(N-\frac{2}{\alpha}\right)} { \Gamma(N) }{\frac{\Gamma\left(x -1+ \frac{2}{\alpha}
\right)
 }{\Gamma(x-1)}}\left(1+\frac{2x }{\alpha(x-1)}\right) \; ,
 \end{align}
whilst, for ZF, it is the solution to
 \begin{align}\label{eq:ZF_cond_smallbeta}
&\frac{1}{\lambda \pi p r_{\rm tr}^4 \beta ^{\frac{2}{\alpha}} }
 =
\frac{\Gamma\left(N-x+1-\frac{2}{\alpha}\right)}{ \Gamma(N-x+1) }
{\frac{\Gamma\left(x -1+ \frac{2}{\alpha}\right)
}{\Gamma(x-1)}} \left(1 +\frac{2x}{\alpha(x-1)}
+\frac{2(x-1)}{\alpha(N-x+1)}\right) \; .
\end{align}

Note that the right hand sides of both
(\ref{eq:MRC_cond_smallbeta}) and (\ref{eq:ZF_cond_smallbeta}) are
increasing in $x$, hence $M_{\rm opt}$ is unique.
\item As $N \to \infty$,
\begin{align} \label{eq:MoptAsym}
\frac{M_{\rm opt}}{N} \to \varrho^{\rm opt}
\end{align}
where, for MRC,
\begin{align}\label{eq:MRC_cond_largeM_smallbeta}
& \varrho^{\rm opt}=\frac{1}{\beta} \left(\frac{\alpha }{ \lambda \pi p
r_{\rm tr}^4(\alpha+2)  } \right)^{\frac{\alpha}{2}}    \; ,
 \end{align}
whist, for ZF, $\varrho^{\rm opt}$ is the solution to
\begin{align}\label{eq:ZF_cond_largeM_smallbeta}
\varrho^{\rm opt} &=\frac{ \frac{1}{\beta}\left(\frac{\alpha}{\lambda \pi p
r_{\rm tr}^4  (\alpha+2) } \right)^{\frac{\alpha}{2}} \left(1-\varrho^{\rm
opt}\right)}{\left(1 +\left(\frac{\varrho^{\rm opt}}{1-\varrho^{\rm
opt}}\right)\left(\frac{2}{\alpha+2}
\right)\right)^{\frac{\alpha}{2}}} \; .
\end{align}
\end{itemize}
These equations confirm the intuition that, for both receivers, as
the node density is increased, less transmit antennas should be
used.  This is because adding more nodes to the network increases
the aggregate interference, and thus, it is better for each node to
transmit will less data streams in order to ``balance'' the overall
network interference. This behavior was also observed for SINR
operating values as high as $3$ dB in Figs.\
\ref{fig:mrc_throughput_changingM} and
\ref{fig:zf_throughput_changingM}.
\item The results above also reveal that less transmit antennas should
be used if the SINR operating value $\beta$ increases, a phenomenon
which is confirmed in Fig.\ \ref{fig:compare_throughput_changeSINR}.
Moreover, the optimal number of transmit antennas is \emph{at
least} as many for MRC as for ZF. For the example shown in Fig.\
\ref{fig:compare_throughput_changeSINR}, for MRC, all transmit
antennas should be used when $\beta$ is below $-3$ dB; whereas for
ZF, $\beta$ must be below $-10$ dB.
\item By noting that $R(M)<1$, it is clear that spatial multiplexing performs better than OSTBC.
Moreover, for spatial multiplexing, since $\frac{\Gamma\left(m-\frac{2}{\alpha}\right)}{ \Gamma(m) }$ is
decreasing in $m$, for $M>2$, MRC achieves a higher throughput than ZF.
\end{itemize}
In general, these results indicate  that spatial multiplexing with MRC, with all transmit antennas active, is the most favorable scheme for low SINR operating values.

\section{Transmission Capacity}

In this section, we turn to the analysis of transmission capacity.
In general, we find that for both spatial multiplexing and OSTBC, an
exact analysis of the transmission capacity is intractable due to
the complexity involved with inverting the exact outage probability
expressions. One exception is the case of spatial multiplexing with
ZF receivers with $M=N$, for which an exact expression for the
transmission capacity is obtained from (\ref{eq:cdf_general_simple})
and (\ref{eq:tc_def}) as
\begin{align}\label{eq:tc_zf_exact}
{\rm c}_{\rm ZF}(\epsilon) 
 = \frac{N(1-\epsilon)}{(\beta R^\alpha)^{\frac{2}{\alpha}} \eta(N)}\left(\log\left(\frac{1}{1-\epsilon}\right) - \frac{\beta r_{\rm tr}^\alpha N}{\rho}\right) \; .
\end{align}
For all other scenarios, we focus on studying the transmission
capacity for small outage levels, which is representative of
practical systems. Under these conditions, based on
(\ref{eq:SINR_general}), we present the following key theorem which,
after substituting the parameters in Table \ref{table:gamma_param},
yields closed-form expressions for the transmission capacity of the
spatial multiplexing and OSTBC schemes.

\begin{theorem}\label{lemm:tc_general}
If the SINR takes the general form in (\ref{eq:SINR_general}), then the
transmission capacity as $\epsilon - {\rm F}^{\rm SU} (\beta) \to 0$ can be written as
\begin{align}\label{eq:tc_general}
{\rm c}(\epsilon) &=\frac{\zeta \theta^{\frac{2}{\alpha}} \Gamma(n)
\Gamma(m)}  {\pi   \beta^{\frac{2}{\alpha}}
\Omega^{\frac{2}{\alpha}}\Gamma\left(n+\frac{2}{\alpha}\right)
\Gamma\left(m-\frac{2}{\alpha}\right) } \frac{\left(\epsilon- {\rm F}^{\rm SU}
(\beta) \right)^{+} }{{\rm E}_Y\left[e^{-\frac{\beta(Y+1)}{\theta}}
{}_1
F_1\left(1-m,1+\frac{2}{\alpha}-m,\frac{\beta(Y+1)}{\theta}\right)\right]}
+ O\left(\left(\left(\epsilon - {\rm F}^{\rm SU} (\beta)\right)^2\right)^+\right)
\end{align}
where ${}_1 F_1(\cdot,\cdot,\cdot)$ is the Kummer confluent
hypergeometric function, and the notation $(\cdot)^+$ implies
$(a)^{+} = \max (a, 0)$.  Also, the expectation in
(\ref{eq:tc_general}) can be written as
\begin{align}\label{eq:hypergeometric_TC}
{\rm E}_Y [ \cdots ] &=
\frac{e^{-\frac{\beta}{\theta}}}{\Gamma\left(m-\frac{2}{\alpha}\right)
} \sum_{\ell=0}^{m-1} \binom{m-1}{\ell} \left(\frac{\beta
}{\theta}\right)^{\ell} \Gamma\left(m-\ell-\frac{2}{\alpha}\right) \sum_{\tau=0}^{\ell} \binom{\ell}{\tau} {\rm
E}_Y\left[e^{-\frac{\beta Y}{\theta}} Y^\tau\right] ,
 \end{align}
and ${\rm F}^{\rm SU} (\beta)$ represents the outage probability with \emph{no
multi-node interference}, computed as
\begin{align}
{\rm F}^{\rm SU} (\beta) = {\rm Pr} \left( \frac{ W }{Y + 1} \leq
\beta\right) = 1 - e^{-\frac{\beta}{\theta}} \sum_{k=0}^{m-1} \frac{
\left(\frac{\beta}{\theta}\right)^k}{k!}  \sum_{\tau=0}^k
\binom{k}{\tau} {\rm E}_Y\left[ e^{-\frac{\beta Y}{\theta}} Y^\tau
\right] \; .
\end{align}
The remaining expectations are given in closed-form in
(\ref{eq:exp_poly_selfint_term}).
\end{theorem}
\begin{proof}
Follows by taking a Taylor expansion of (\ref{eq:cdf_general})
around $\lambda=0$, and then finding the inverse of the resulting
expression w.r.t.\ $\lambda p$.
\end{proof}

Note that the factor ${\rm F}^{\rm SU} (\beta)$ represents the outage
probability of a single user MIMO system, for which outages are
caused by self-interference and AWGN. Due to the
additional multi-node interference in ad hoc networks, any specified
outage constraint $\epsilon$ which falls below ${\rm F}^{\rm SU} (\beta)$ can
never be met, and therefore the transmission capacity in such cases
is zero. This phenomenon is illustrated in Fig.\
\ref{fig:mrc_tc_outage}, where the outage probability is plotted
versus intensity for spatial multiplexing with MRC receivers, for
different antenna configurations. For the results shown, an outage
probability of less than $0.2$ can never be achieved when $M \ge 2$,
due to the effects of AWGN and self-interference, which ensures that
$\epsilon>0$ when $\lambda=0$.  Note however, that for OSTBC and
spatial multiplexing with ZF receivers, there is no
self-interference, and ${\rm F}^{\rm SU} (\beta)$ in this case accounts for
outages due to AWGN only.

\begin{figure}[tb!]
\centerline{\includegraphics[width=0.7\columnwidth]{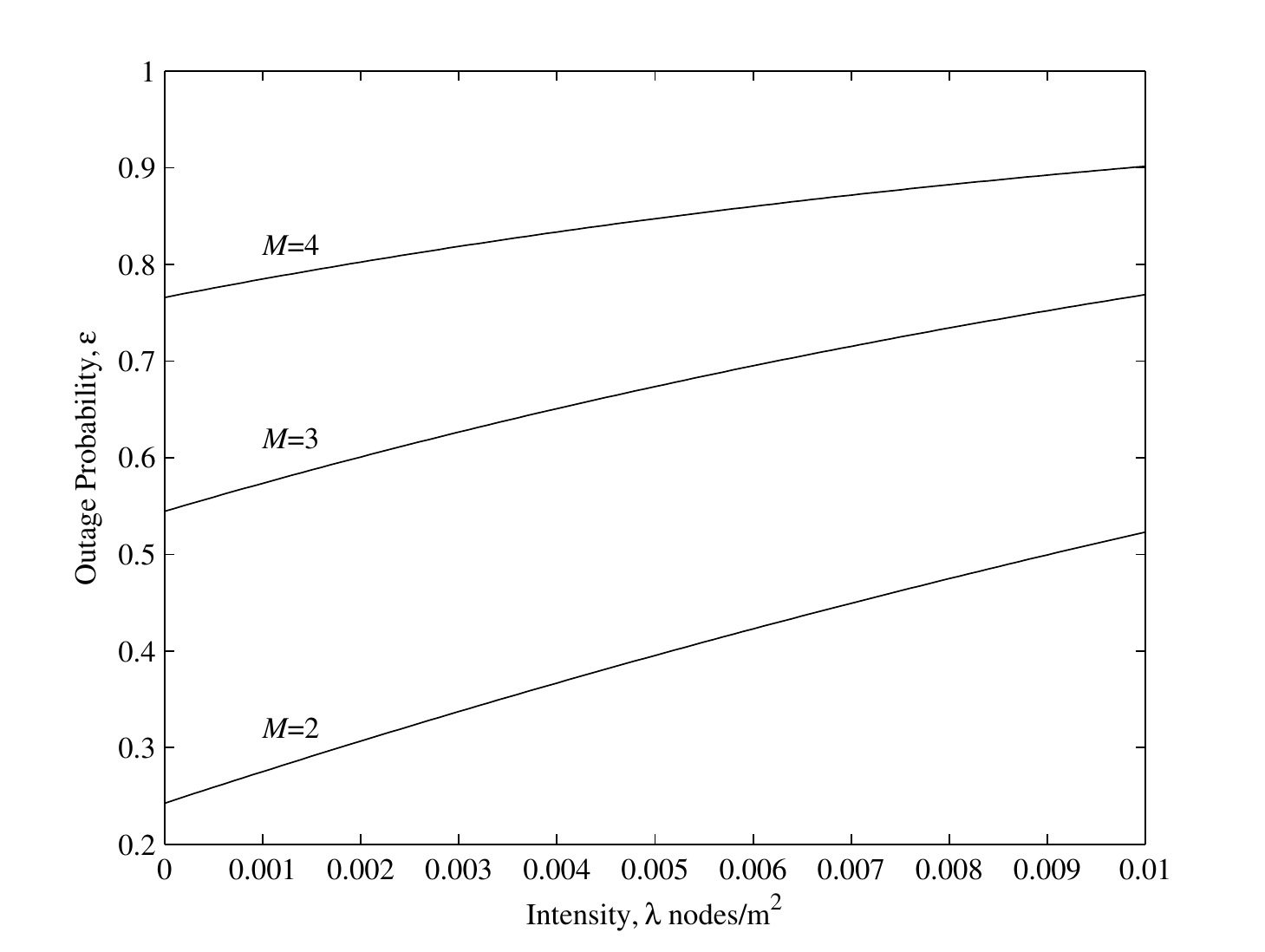}}
\caption{Outage probability vs.\ intensity of slotted ALOHA with spatial multiplexing using MRC receivers, and with $N=4$, ${r_{\rm tr}}=3$ m, $\alpha=4.23$, $\beta=3$ dB, $\rho=30$ dB and $p=1$.}
\label{fig:mrc_tc_outage}
\end{figure}

To compute the transmission capacity achieved by spatial
multiplexing and OSTBC, we substitute the relevant parameters from
Table \ref{table:gamma_param} into (\ref{eq:tc_general}). The
accuracy of our transmission capacity expression is confirmed in
Fig.\ \ref{fig:zf_tc_confirm_outage}, which plots the transmission
capacity vs.\ outage probability $\epsilon$ for spatial multiplexing
with ZF receivers. We see that for outage probabilities as high as
$\epsilon=0.1$, our expression is accurate. Although not shown,
similar accuracy has been observed for the other schemes also.

\begin{figure}[tb!]
\centerline{\includegraphics[width=0.7\columnwidth]{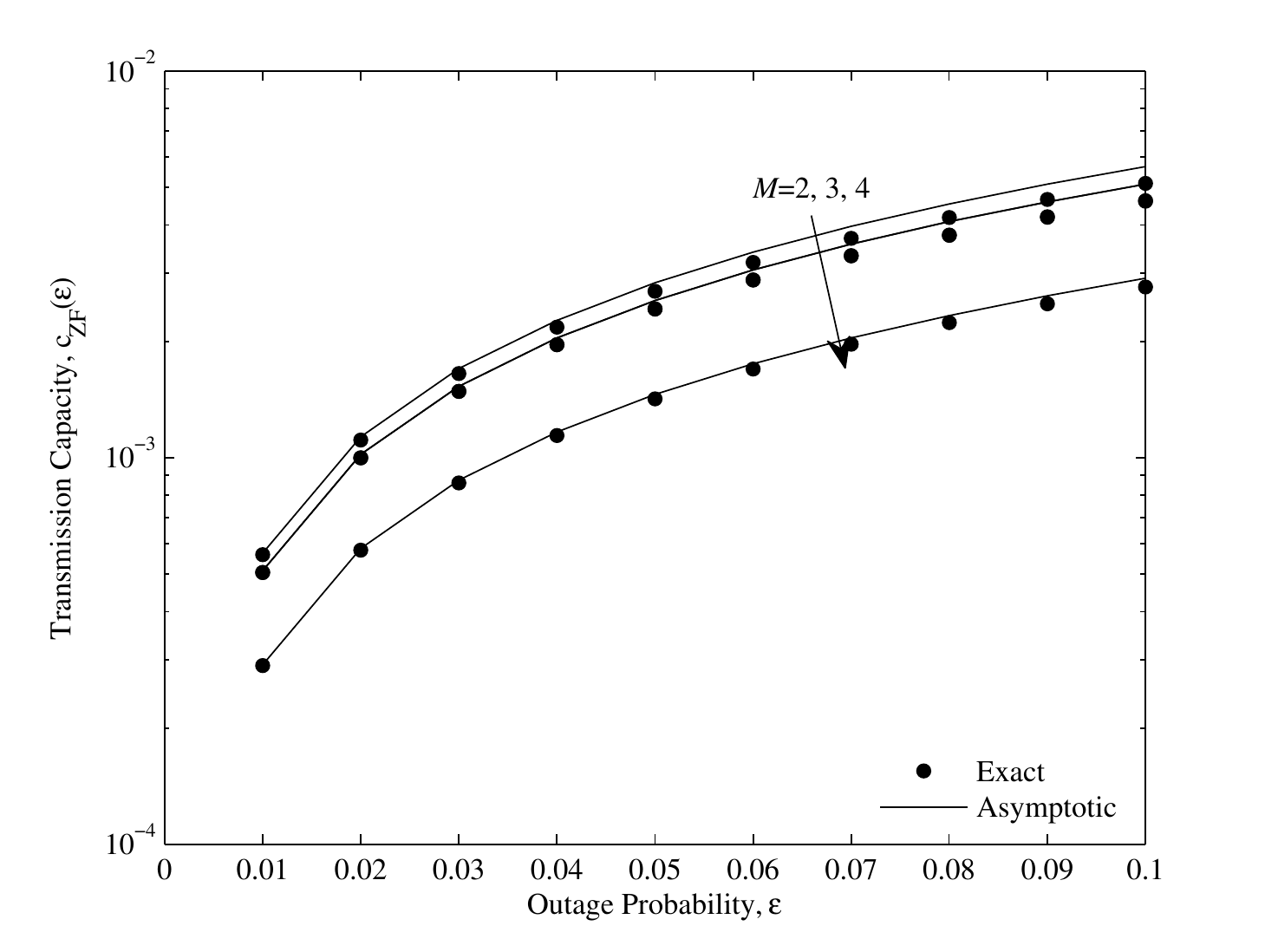}}
\caption{Transmission capacity vs.\ outage probability of slotted ALOHA using  ZF receivers, and with ${r_{\rm tr}}=3$ m, $\beta=3$ dB, $\alpha=4$, $N=4$ and $p=1$.}
\label{fig:zf_tc_confirm_outage}
\end{figure}

For large numbers of antennas, we present the following corollary for the transmission capacity.
\begin{corollary}\label{corr:tc_scaling}
As $M,N \to \infty$, the transmission capacity as $\epsilon - {\rm F}^{\rm SU} (\beta) \to 0$ satisfies
\begin{align}\label{eq:tc_general_scale}
\frac{{\rm c}(\epsilon)}{\zeta} &= \left(\frac{\left(m \theta-\beta - \beta u \Upsilon\right)^{+}}{n \Omega}\right)^{\frac{2}{\alpha}} \frac{\left(
\epsilon - {\rm F}^{\rm SU} (\beta) \right)^{+}}{ {\pi r_{\rm tr}^2  \beta^{\frac{2}{\alpha}}
} } + O\left(\left(\left(
\epsilon - {\rm F}^{\rm SU} (\beta)\right)^2 \right)^{+}\right)\; .
\end{align}
\end{corollary}
\begin{proof}
See Appendix \ref{app:tc_scaling}.
\end{proof}
(\ref{eq:tc_general_scale}) implies that if $\epsilon > {\rm F}^{\rm SU} (\beta)$ and $\frac{\left(m \theta-\beta - \beta Y\right)^{+}}{n \Omega}$ converges to a constant for large $M,N$, then the
transmission capacity will scale \emph{linearly} with $\zeta$. In the
following, we will use these results to investigate how the
transmission capacity scales with the number of antennas, for the
spatial multiplexing and OSTBC systems which we consider.  Under
some conditions, we will show that linear scaling is indeed possible.

\subsection{Transmission Capacity of Spatial Multiplexing}

Here we use \emph{Corollary \ref{corr:tc_scaling}} to gain insights
into the scaling behavior of the transmission capacity for spatial
multiplexing with MRC and ZF receivers, as given by the following
corollaries.

\begin{corollary} \label{th:MRC_TransCap}
As $N \to \infty$ with $M=\kappa N$ where $0<\kappa \le 1$, the
transmission capacity of spatial
multiplexing with MRC receivers as $\epsilon-H \left(\beta - \bar{\beta}_{\rm MRC}\right)\to 0$ behaves as
\begin{align}\label{eq:tc_mrc_general_scale}
& \frac{{\rm c}_{\rm MRC}(\epsilon)}{N} \to \kappa \left(\frac{r_{\rm tr}^\alpha}{\rho} + 1 \right)^{\frac{2}{\alpha}} \left( \left(\bar{\beta}_{\rm MRC}-\beta\right)^{+}
\right)^{\frac{2}{\alpha}} \frac{\left(\epsilon-H \left(\beta - \bar{\beta}_{\rm MRC}\right)\right)^{+}}  {\pi r_{\rm tr}^2  \beta^{\frac{2}{\alpha}}}  + O\left(\left(\left(\epsilon-H \left(\beta - \bar{\beta}_{\rm MRC}\right)\right)^2\right)^{+} \right)
\end{align}
where
\begin{align}
\bar{\beta}_{\rm MRC} =  \frac{1}{\kappa\left(\frac{r_{\rm tr}^\alpha}{\rho} + 1\right)}
\end{align}
and $H(\cdot)$ is the Heaviside step function.
\end{corollary}
\begin{proof}
Follows by substituting the relevant parameters for MRC receivers in
Table \ref{table:gamma_param} into (\ref{eq:tc_general_scale}), and
noting that \cite[Eq. (8)]{tse99} $\frac{W}{Y+1}$ converges in probability to  $\frac{1}{\kappa\left(\frac{ r_{\rm
tr}^\alpha}{\rho} + 1\right)}$ as $N \to \infty$  with $M= \kappa N$.
\end{proof}

\begin{corollary} \label{th:ZF_TransCap}
As $N \to \infty$ with $M=\kappa N$ where $0<\kappa < 1$,  the
transmission capacity of spatial multiplexing with ZF receivers as $\epsilon- H\left(\beta-\bar{\beta}_{\rm ZF}\right) \to 0$ behaves as
\begin{align}\label{eq:tc_zf_general_scale}
\frac{{\rm c}_{\rm ZF}(\epsilon)}{N} 
&\to \kappa
\left(\frac{ r_{\rm tr}^\alpha }{\rho}\right)^{\frac{2}{\alpha}} \left( \left( \bar{\beta}_{\rm ZF}-\beta\right)^{+}\right)^{\frac{2}{\alpha}} \frac{\left(\epsilon- H
\left(\beta-\bar{\beta}_{\rm ZF}\right) \right)^{+}}{ \pi r_{\rm tr}^2  \beta^{\frac{2}{\alpha}}} + O\left(\left(\left(\epsilon- H\left(\beta-\bar{\beta}_{\rm ZF}\right) \right)^2\right)^{+} \right)
\end{align}
where
\begin{align}
\bar{\beta}_{\rm ZF} = \frac{\rho}{r_{\rm tr}^\alpha}
\left(\frac{1}{\kappa}-1\right)\; .
\end{align}
\end{corollary}
\begin{proof}
Follows by substituting the relevant parameters for ZF receivers in
Table \ref{table:gamma_param} into (\ref{eq:tc_general_scale}), and
noting that \cite[Theorem 7.2]{tse99} $W$ converges in probability to  $\frac{\rho}{ r_{\rm tr}^\alpha}
\left(\frac{1}{\kappa}-1\right) $ as $N \to \infty$ with $M=\kappa N$.
\end{proof}

At low outage probabilities $\epsilon$, these results imply that
for MRC\footnote{The notation $g(N)= \Theta(f(N))$ means that
$\phi_1 f(N) < g(N) < \phi_2 f(N)$ for $N \to \infty$, where
$\phi_1$ and $\phi_2$ are constants independent of $N$.},
\begin{align}\label{eq:TC_linear_MRC}
{\rm c}_{\rm MRC}(\epsilon) = \left\{
\begin{array}{ll}
\Theta\left( N \right) , & \bar{\beta}_{\rm MRC} > \beta \\
0 , & {\rm otherwise}
\end{array}
\right.
\end{align}
while for ZF,
\begin{align}\label{eq:TC_linear_ZF}
{\rm c}_{\rm ZF}(\epsilon) = \left\{
\begin{array}{ll}
\Theta \left( N \right) , &
   \bar{\beta}_{\rm ZF}  > \beta \\
0 , & {\rm otherwise}
\end{array}
\right. .
\end{align}
Moreover, from (\ref{eq:tc_mrc_general_scale}) and (\ref{eq:tc_zf_general_scale}), assuming that $\bar{\beta}_{\rm MRC} > \beta$ and $\bar{\beta}_{\rm ZF} > \beta$, the transmission capacity can be approximated for large $M$ and $N$ as
\begin{align}\label{eq:tc_zf2}
& {\rm c}_{\rm MRC}(\epsilon) \approx M^{1-\frac{2}{\alpha}} \left( N - M \beta-\frac{M \beta r_{\rm tr}^\alpha}{\rho}
\right)^{\frac{2}{\alpha}}  \frac{\epsilon}{\pi r_{\rm tr}^2  \beta^{\frac{2}{\alpha}}}, \\
& {\rm c}_{\rm ZF}(\epsilon) \approx M^{1-\frac{2}{\alpha}}
\left( N-M-\frac{ M \beta r_{\rm tr}^\alpha}{\rho}\right)^{\frac{2}{\alpha}} \frac{\epsilon}{\pi r_{\rm tr}^2  \beta^{\frac{2}{\alpha}}} \notag  \; .
\end{align}
From the preceding equations, we observe the following:
\begin{itemize}
\item Under the assumption that the operating SINR $\beta$ is sufficiently
small such that $\bar{\beta}_{\rm MRC} > \beta$ and $
\bar{\beta}_{\rm ZF} > \beta$, both receivers achieve a transmission capacity scaling which is \emph{linear} in $N$.
\item ZF receivers are preferable over MRC receivers for sufficiently high SINR operating values, i.e., $\beta >1$, and vice-versa for small SINR operating values.
\item Spatial multiplexing with MRC receivers achieves the same performance as a spatial multiplexing system with ZF receivers, where the ZF receivers utilize $N-M \beta +1$ degrees of freedom (d.o.f.) to boost the signal power and $M\beta-1$ d.o.f.\ to cancel the self-interference from $M \beta - 1$ self-interfering streams.
\end{itemize}

This preceding discussion only tells part of the story,
and to gain a clearer picture we must also investigate the
conditions on $\beta$ in (\ref{eq:TC_linear_MRC}) and
(\ref{eq:TC_linear_ZF}) which ultimately dictate when linear scaling
is achievable for each receiver. Comparing these conditions, we find
that
\begin{align} \label{eq:ZFMRC_Condition}
\bar{\beta}_{\rm MRC} < \bar{\beta}_{\rm ZF} \quad \Leftrightarrow
\quad   \frac{\rho}{r_{\rm tr}^\alpha} > \frac{ \kappa }{ 1 - \kappa
} \; .
\end{align}
This result is interesting, since it gives insight into which
receiver, MRC or ZF, can achieve linear scaling ``more easily'',
i.e., for a larger range of operating SINRs $\beta$, in terms of the
average received SNR $\rho r_{\rm tr}^{-\alpha}$ and the antenna
parameter $\kappa$.  In particular, (\ref{eq:ZFMRC_Condition}) will
be satisfied when either the average received SNR is sufficiently
high or $\kappa$ is sufficiently small (i.e., $M \ll N$), in which
case, ZF will achieve linear scaling over a wider range of operating
SINRs than MRC.  This is intuitive, since by increasing the average
received SNR, the SINR of ZF grows proportionately, while the SINR
of MRC approaches a constant. Moreover, by decreasing $\kappa$, the
signal power of ZF will converge to the signal power of MRC, while
the self-interference power of MRC approaches a constant. On the
other hand, if the average received SNR is sufficiently low or
$\kappa$ is sufficiently high (i.e., $M \approx N$), then MRC will
achieve linear scaling for a wider range of $\beta$ values than ZF.

In general, from a transmission capacity perspective, these results
imply that for networks operating with low SNR it is highly
desirable to employ MRC receivers in favor or ZF. If the SNR is not
low, however, the choice is less clear, and depends largely on the
specific network operating conditions.  In particular, if the SNR is
large enough that (\ref{eq:ZFMRC_Condition}) is satisfied, then ZF
will have the advantage of achieving linear scaling more easily than
MRC. Regardless of this point, however, if $\beta$ is chosen
sufficiently small, then the condition (\ref{eq:ZFMRC_Condition})
becomes irrelevant, both schemes will achieve linear scaling.  Thus, the specific
operating SINR $\beta$, as well as the SNR, play a critical role in
deciding which spatial multiplexing receiver is most advantageous.

From (\ref{eq:TC_linear_MRC}) and (\ref{eq:TC_linear_ZF}), we
know that if $M=\kappa N$, then linear scaling of the
transmission capacity is achieved for both receivers. A natural
question which arises, however, is whether one can do \emph{even
better} than linear scaling, if $M$ is not confined to vary linearly
with $N$. This is answered in the following proposition:
\begin{proposition}\label{prop:TCZF}
For spatial multiplexing systems with either MRC or ZF receivers,
linear scaling of the transmission capacity, i.e.,\ ${\rm c}(\epsilon)
=\Theta(N)$, achieved when $M = \kappa N$ with $0<\kappa<1$, is the
\emph{best possible scaling}. Moreover, if $M$ varies sub-linearly
with $N$, then the transmission capacity scales sub-linearly also.
\end{proposition}
\begin{proof}
See Appendix \ref{app:TCZF_scaling_proof}.
\end{proof}

Note that this result is certainly not obvious, since in general,
the transmission capacity is typically maximized by selecting
$\kappa < 1$; i.e., by choosing only a subset of antennas for
transmission. Thus, it is not immediately clear whether the optimal $M$ should grow in
proportion to $N$, or at some slower (i.e., sub-linear) rate.

We now consider the optimal antenna configuration which maximizes the transmission capacity scaling, given in the following proposition.
\begin{proposition}\label{prop:opt_antenna}
The antenna configuration which optimizes the transmission capacity
is
\begin{align}\label{eq:opt_antenna}
M = \max \left( 1, \min\left(\left\lfloor N\kappa^{*} \right\rfloor,N \right) \right)
\end{align}
where
\begin{align}
\kappa^{*} = \left\{
\begin{array}{ll}
\left(1-\frac{2}{\alpha}\right) \frac{1}{\beta\left(1+\frac{r_{\rm tr}^\alpha}{\rho}  \right)}  , & \quad \quad {\rm for \; MRC} \\
\left(1-\frac{2}{\alpha}\right) \frac{1}{1+ \frac{ r_{\rm tr}^\alpha \beta}{\rho}} , & \quad \quad {\rm for \; ZF}
\end{array}
\right. .
\end{align}
\end{proposition}
From this optimal antenna configuration, we observe the following:
\begin{itemize}
\item The optimal number of data streams is higher for MRC than ZF for sufficiently small SINR operating values $\beta$, i.e., $\beta < 1$, and vice-versa for sufficiently high SINR operating values.
\item As $\alpha \to 2$, $\kappa^* \to 0$, and the optimal number of data streams decreases. This can be explained by noting that low path loss exponents correspond to scenarios where all multi-node interference, including far-away interference, is significant. In this scenario,  the number of data streams should be decreased to reduce the impact of multi-node interference.
\item As $\alpha$ goes large, the impact of multi-node interference becomes negligible, and the optimal number of data streams is dependent on the specific network operating conditions. For example, when MRC receivers are employed, transmitting the maximum number of data streams is optimal for sufficiently low $\beta$, and vice-versa for high $\beta$. For ZF receivers,  transmitting the maximum number of data streams is optimal for high average received SNR $\rho r_{\rm tr}^{-\alpha}$.  This can be observed in Fig.\ \ref{fig:zf_tc}, which
shows the transmission capacity of spatial multiplexing with ZF for
different $M$ at high transmit SNR $\rho$.  The ``Analytical" curves in the figure are based on
(\ref{eq:tc_general}).
\end{itemize}

\begin{figure}[tb!]
\centerline{\includegraphics[width=0.7\columnwidth]{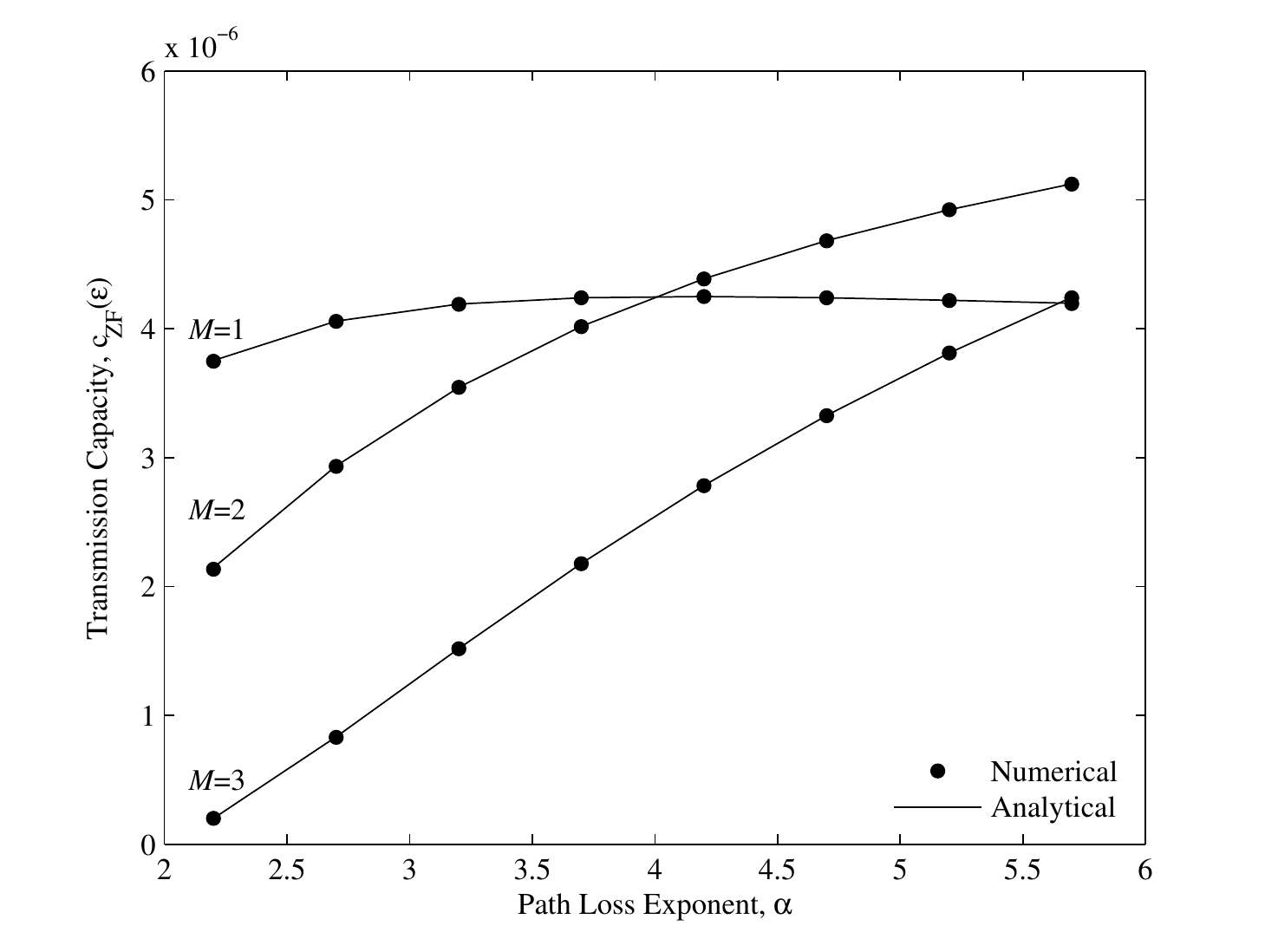}}
\caption{Transmission capacity vs.\ path loss exponent of slotted ALOHA with spatial multiplexing using ZF receivers in the interference-limited regime, and with ${r_{\rm tr}}=3$ m, $\beta=3$ dB, $\epsilon=0.0001$ and $p=1$.}
\label{fig:zf_tc}
\end{figure}

\emph{Proposition \ref{prop:opt_antenna}} also provides a
useful design criteria for choosing how many data streams should be
transmitted for a given path loss exponent. For example, for ZF, when
$\alpha=4$ and at high transmit SNR $\rho$, we see in (\ref{eq:opt_antenna}) that the optimal number
of transmitted data streams should be half the number of receive
antennas. The optimal number of transmit antennas was also
considered in \cite{govindasamy07}, which considered the use of MMSE
receivers to cancel the interference from the closest transmitting
nodes and a single-link performance measure. In contrast, we
consider the use of ZF receivers to cancel interference from the
corresponding transmitter only, and a \emph{network} performance
measure.

\subsubsection{Comparison with Multi-Node Interference
Cancelation (Single-Input Multiple-Output with Partial Zero-Forcing)}

Linear scaling was also previously shown to occur in \cite{jindal09}
for single-stream transmission, i.e.,\ $M=1$, by employing a
\emph{partial zero forcing} (PZF) scheme, where the receive antennas
were used to simultaneously cancel the interference from the $k_{\rm
PZF}$ closest interfering nodes and boost the signal power from the
corresponding transmitter.  Specifically, the transmission capacity
of the PZF scheme at high transmit SNR $\rho$ was shown to scale as
\begin{align}\label{eq:tc_PZF}
{\rm c}_{\rm PZF}(\epsilon) = \Theta\left(k_{\rm PZF} \left(\frac{N-k_{\rm PZF}}{k_{\rm PZF}}\right)^{\frac{2}{\alpha}} \right)\;.
\end{align}
Linear scaling was then achieved by setting $k_{\rm PZF} = \kappa N$
where $0 < \kappa<1$. At high transmit SNR $\rho$, comparing (\ref{eq:tc_zf2}) and
(\ref{eq:tc_PZF}), we see that setting $k_{\rm PZF}=M$ results in
the PZF scheme achieving the same scaling as our spatial
multiplexing with ZF receivers scheme\footnote{Note that similar comparisons can be made with our spatial multiplexing with MRC receivers scheme.}.

To understand the reason behind this same linear scaling result, we
explore some similarities between our spatial multiplexing with ZF receivers scheme
and the PZF scheme in \cite{jindal09}. For both schemes, it can be
shown that the signal component contributes to the transmission
capacity scaling by a factor of $N^{\frac{2}{\alpha}}$, and is
achieved by allocating $N(1-\kappa)$ d.o.f.\ to
boost the signal power, where $0 < \kappa <1$. The key difference,
however, lies with how the remaining $N \kappa$ d.o.f.\ are used.
For our spatial multiplexing with ZF receivers scheme, the $N
\kappa$ d.o.f.\ are used to cancel interference from the
corresponding transmitting node, while for the PZF scheme, they are
used to cancel interference from the $k_{\rm PZF}=N \kappa$ closest
interferers. However, it turns out that the resulting interference
component of both schemes contributes to the transmission capacity
through the \emph{same} scaling factor, which is
$N^{1-\frac{2}{\alpha}}$. To see why, it is convenient to express
the transmission capacity (\ref{eq:tc_def}) in the alternative form
\begin{align}\label{eq:TC_alt}
{\rm c}_{\rm ZF}(\epsilon) =  \lambda^{*}(\epsilon)(1 - \epsilon)
\end{align}
where $\lambda^{*}(\epsilon)$ is the inverse of $\epsilon={\rm
F}_{\rm ZF}\left(\beta;\frac{\lambda}{M}\right)$ taken w.r.t.\
$\lambda$. This alternate form can be obtained by noting that
\begin{align}
\epsilon = {\rm F}(\beta; \lambda/M) \quad \Rightarrow \quad \lambda =  M
{\rm F}^{-1}(\epsilon),
\end{align}
which has the same form as (\ref{eq:tc_def}) for small $\epsilon$.
(\ref{eq:TC_alt}) indicates that from a transmission capacity
perspective, increasing the number of data streams $M$ has the same
effect as decreasing the density of transmitting nodes by a factor
of $M$.

From the modified transmission capacity definition
(\ref{eq:TC_alt}), we can now explain why for both schemes the
interference component contributes to the transmission capacity by
the same scaling factor. First, for spatial multiplexing with ZF receivers, the average distance from the receiver
to the closest interfering node,  where the interfering nodes are distributed with density $\lambda/M$, can be shown to scale as $\sqrt{M}$ \cite[Eq. (9)]{haenggi05} whereas for
PZF (after cancelation), it scales as $\sqrt{k_{\rm PZF}}$
\cite{jindal09}. Thus, with $M=k_{\rm PZF}$, these coincide. Second,
for both schemes, the average number of interfering data streams per
unit area is the same, given by $\lambda$. For spatial multiplexing
with ZF, this can be seen by noting that $M$ data streams are being
transmitted by $\lambda/M$ nodes per unit area.

Summing up, the key point of the discussion above is that the same
linear transmission capacity scaling can be achieved whether we
employ MIMO with ZF receivers, or SIMO with PZF.  Clearly, the
choice as to which scheme to employ in practice will depend on the
specific network design. When there is the capability for employing
multiple antennas on all nodes, MIMO with ZF will be more
attractive, since it has lower complexity and less stringent
requirements on the level of channel knowledge at the receivers. On
the other hand, if single-antenna transmitters are required (e.g.,
due to size limitations, relevant for sensor networks), then SIMO
with PZF will be appropriate. A key implementation issue for such
networks, however, is how to accurately estimate the multi-node
interference which is needed to perform interference cancelation at
each receiver. Whilst very well motivated in theory, some important
questions still remain as to the practical feasibility of this
approach.

\subsection{Transmission Capacity of OSTBC}

We now consider the transmission capacity of OSTBC at low outage probabilities $\epsilon$.  In this case,
analogous to (\ref{eq:tc_zf2}), we consider the case where $N$ grows
large, and also assume that $\beta$ and $\epsilon$ are selected such
that $\epsilon > {\rm F}^{\rm SU} (\beta)$. Moreover, we  consider high transmit SNRs $\rho$, and thus, substituting the
relevant parameters from Table \ref{table:gamma_param} into
(\ref{eq:tc_general}), we get
\begin{align}\label{eq:tc_ostbc2}
\tilde{{\rm c}}_{\rm OSTBC} (\epsilon) &= \Theta\left(g(M) N^{\frac{2}{\alpha}} \right)
\end{align}
where
\begin{align}
g(M) = R(M) M^{\frac{2}{\alpha}} \frac{\Gamma\left(\frac{N_I(M)}{M}\right)}{\Gamma\left(\frac{N_I(M)}{M} + \frac{2}{\alpha}\right)}
\end{align}
where we have made explicit the dependence of the code rate $R$ and
$N_I$ on $M$. (Note that this result applies for arbitrary values of
$M$.)  In general, there is no simple relationship characterizing the dependence of $R$ and
$N_I$ on $M$. Thus, to proceed, here we consider the extreme
scenarios pertaining to minimum-rate and maximum-rate codes.

For minimum-rate codes, i.e., cyclic diversity systems, $R(M)=1/M$
and $N_I(M) = M$. Thus, it follows that $\tilde{{\rm c}}_{\rm OSTBC}
(\epsilon) = \Theta\left(M^{\frac{2}{\alpha}-1} N^{\frac{2}{\alpha}}
\right)$, which is obviously maximized if $M = 1$, yielding
$\tilde{{\rm c}}_{\rm OSTBC} (\epsilon) = \Theta\left(
N^{\frac{2}{\alpha}} \right)$.  This special case simply corresponds
to a standard SIMO system.

For maximum-rate codes, e.g., the Alamouti code, it is known from
\cite{liang03} that for $M$ even
\begin{align}
\frac{1}{2} < R(M) \le \frac{M+2}{2M} ,  \qquad N_I(M) =\frac{M(M+2)}{2}
\end{align}
and for $M$ odd
\begin{align}
& \frac{1}{2} < R(M) \le \frac{M+3}{2(M+1)} ,  \qquad\frac{M(M+1)}{2} \le
N_I(M) < \frac{M(M+2)}{2} \; .
\end{align}
Thus, using these bounds and the inequality $\frac{ m^x}{\Gamma(1-x)} \le \frac{\Gamma(m)}{\Gamma(m-x)} \le m^x$ \cite{hunter07}, we have $g^{\rm lb}(M) \leq g(M) \leq g^{\rm ub}(M)$, where
\begin{align}\label{eq:gub_OSTBC}
& g^{\rm lb}(M) := \frac{1}{2\Gamma\left(1+\frac{2}{\alpha}\right)} \left(\frac{2
M}{M+2}\right)^{\frac{2}{\alpha}} \notag \\
& g^{\rm ub}(M) :=
\begin{cases}
& \left(\frac{2 M}{M+2}\right)^{\frac{2}{\alpha}-1} , \hspace{1.8cm} {\rm even} \; M \\
& \frac{M+3}{2(M+1)} \left(\frac{2 M}{M+1}\right)^{\frac{2}{\alpha}} , \hspace{1cm} {\rm odd} \; M \\
\end{cases}
\; .
\end{align}

From this, it is clear that, regardless of whether $M$ is kept fixed
or allowed to vary with $N$, the best achievable scaling is
$\tilde{{\rm c}}_{\rm OSTBC} (\epsilon) =
\Theta\left(N^{\frac{2}{\alpha}} \right)$. This result, combined
with that for minimum-rate codes above, indicates that OSTBC can
only achieve \emph{sub-linear} scaling\footnote{A similar scaling
result was obtained in \cite{hunter07}, based on the corresponding
SINR approximation illustrated in Figs.\
\ref{fig:ostbc_outage_compare_nt4}--\ref{fig:ostbc_outage_compare_nt4Ns2}.},
in contrast to the spatial multiplexing systems considered
previously.

We now consider the question of how to select $M$ to maximize the
transmission capacity scaling for large $N$.  This is
tantamount to maximizing the leading multiplicative factor $g(M)$. For maximum-rate codes, we will focus on $g^{\rm up}(M)$, since it serves as an accurate approximation for $g(M)$.
It is obvious that $g^{\rm up}(M)$ is maximized by selecting $M=2$ for $M \ge 2$. Thus, the key finding is that for OSTBC, either the Alamouti code or SIMO transmission should be employed. Although it is hard to obtain a simple design rule which specifies which of these two coding schemes will perform better under certain network conditions,  we can gain insights by numerical analysis. In particular, Fig.\ \ref{fig:ostbc_tc} plots the transmission capacity achieved by OSTBC for different $M$.  Results are presented for $N=3$, and the codes used for $M=2$ and $M=3$ are the maximum-rate codes given by (\ref{eq:alamouti_param}) and (\ref{eq:ostbc_code2}) respectively. The ``Analytical" curves are based on (\ref{eq:tc_general}), and are clearly seen to match with the ``Numerical" curves, obtained by taking $\rho$ large in the outage probability expression obtained by substituting the relevant parameters from Table \ref{table:gamma_param} into (\ref{eq:cdf_general}), and subsequently  solving for $\lambda(\epsilon)$ using numerical techniques. We see that the Alamouti code performs better than SIMO for $\alpha <4.7$, and vice-versa for $\alpha >4.7$. 

\begin{figure}[tb!]
\centerline{\includegraphics[width=0.7\columnwidth]{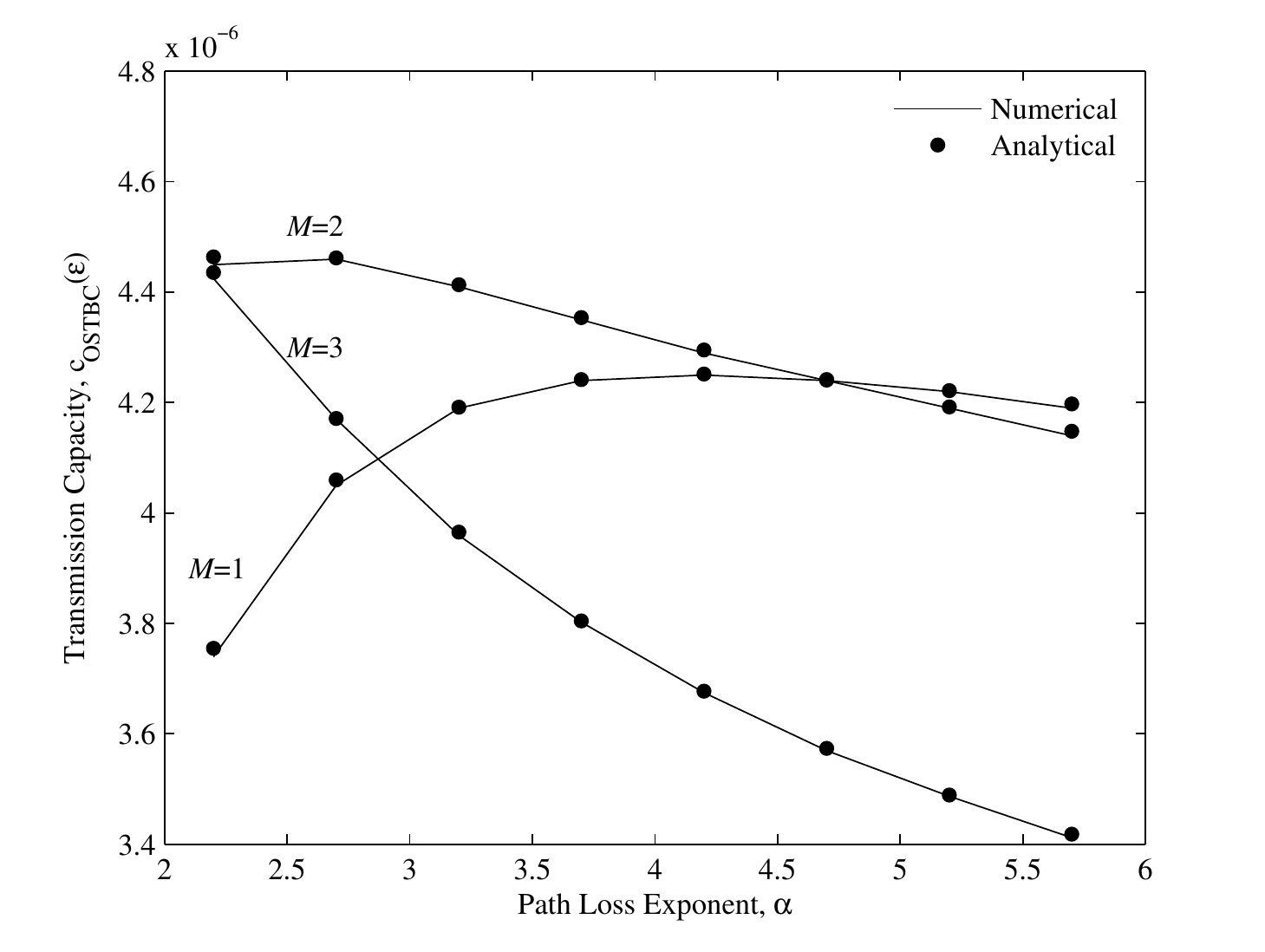}}
\caption{Transmission capacity vs.\ path loss exponent of slotted ALOHA with OSTBC, and with ${r_{\rm tr}}=3$ m, $\beta=3$ dB, $\epsilon=0.0001$ and $p=1$.}
\label{fig:ostbc_tc}
\end{figure}

\subsection{Comparison}

For large $N$ and assuming $\epsilon > {\rm F}^{\rm SU} (\beta)$ in (\ref{eq:tc_general}),  spatial multiplexing  achieves a higher transmission capacity scaling than OSTBC, at high transmit SNR $\rho$, if
\begin{align}\label{eq:zf_ostbc_tc_compare}
\begin{array}{ll}
& R(M)  < \frac{M}{ \Gamma\left(1+\frac{2}{\alpha}\right)}\left(\frac{N_I(M)}{M^3}\left(1-\frac{M}{N}\right) \right)^{\frac{2}{\alpha}}  \quad \quad {\rm for \; ZF} \\
& R(M)  < \frac{M}{ \Gamma\left(1+\frac{2}{\alpha}\right)}\left(\frac{N_I(M)}{ M^3}\left(1-\frac{M\beta}{N}\right)\right)^{\frac{2}{\alpha}} \hspace{0.4cm} {\rm for \; MRC} \; .
\end{array}
\end{align}
Clearly, these conditions depend on the specific OSTBC employed.  As
expected, we see that by incorporating a higher rate code, whilst
keeping all other parameters fixed, the possibility of the condition
(\ref{eq:zf_ostbc_tc_compare}) being satisfied is decreased.  Moreover, by noting that $N_I(M) < M^3$ for all practical OSTBC codes, the condition in
(\ref{eq:zf_ostbc_tc_compare}) becomes more likely to be satisfied
as the path loss exponent $\alpha$ increases; indicating that
environments with higher path loss exponents are more beneficial for
spatial multiplexing,
compared with OSTBC; and vice-versa for low path loss environments.

Figs.\  \ref{fig:tc_compare_changeN} and \ref{fig:tc_compare_changeM} show the transmission capacity achieved by spatial multiplexing with ZF receivers and OSTBC for different antenna configurations. The curves for both spatial multiplexing and OSTBC are based on (\ref{eq:tc_general}). The results show that the relative transmission capacity of the transmission schemes depends on both the antenna configuration, the path loss exponent $\alpha$ and the SINR operating value. In particular, we observe that OSTBC performs best at small $\alpha$ (e.g.,\ $\alpha=2$, corresponding to free space environments), whereas spatial multiplexing performs the best at large $\alpha$ (e.g.,\ $\alpha=6$, corresponding to indoor environments). These observations agree with our analytical conclusions put forth above, based on (\ref{eq:zf_ostbc_tc_compare}).

We also observe that  ZF receivers are preferable over MRC receivers for a sufficiently high SINR operating value, i.e., $\beta=0$ dB, and vice-versa for a low SINR operating value, i.e., $\beta=-5$ dB. This is because for high SINR operating values, the outage probability when there is no multi-node interference, i.e., the outage probability of the single user MIMO system, is higher for MRC than ZF, and offsets any positive gains resulting from the signal power that MRC has over ZF receivers. However, for sufficiently low SINR operating values, and as indicated by previous analysis, the outage probability when there is no multi-node interference is negligible for both receivers, and thus MRC is preferable over ZF.

\begin{figure}[t!]
\centerline{\includegraphics[width=0.7\columnwidth]{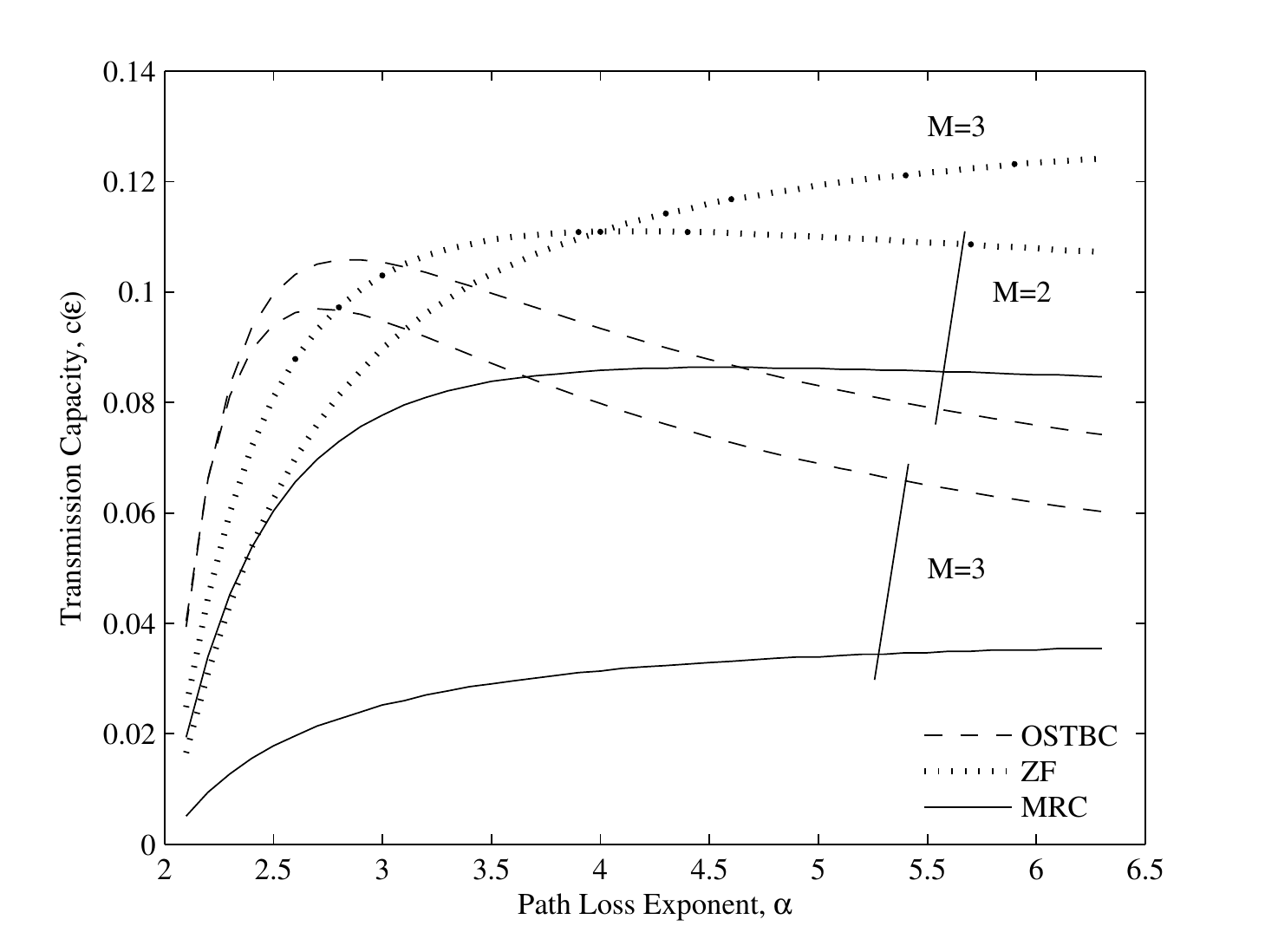}}
\caption{Transmission capacity vs.\ path loss exponent of slotted ALOHA with spatial multiplexing and OSTBC, and with $N=5$, ${r_{\rm tr}}=1$ m, $\beta=0$ dB, $N=5$, $\rho=30$ dB, $\epsilon=0.15$ and $p=1$.}
\label{fig:tc_compare_changeN}
\end{figure}

\begin{figure}[t!]
\centerline{\includegraphics[width=0.7\columnwidth]{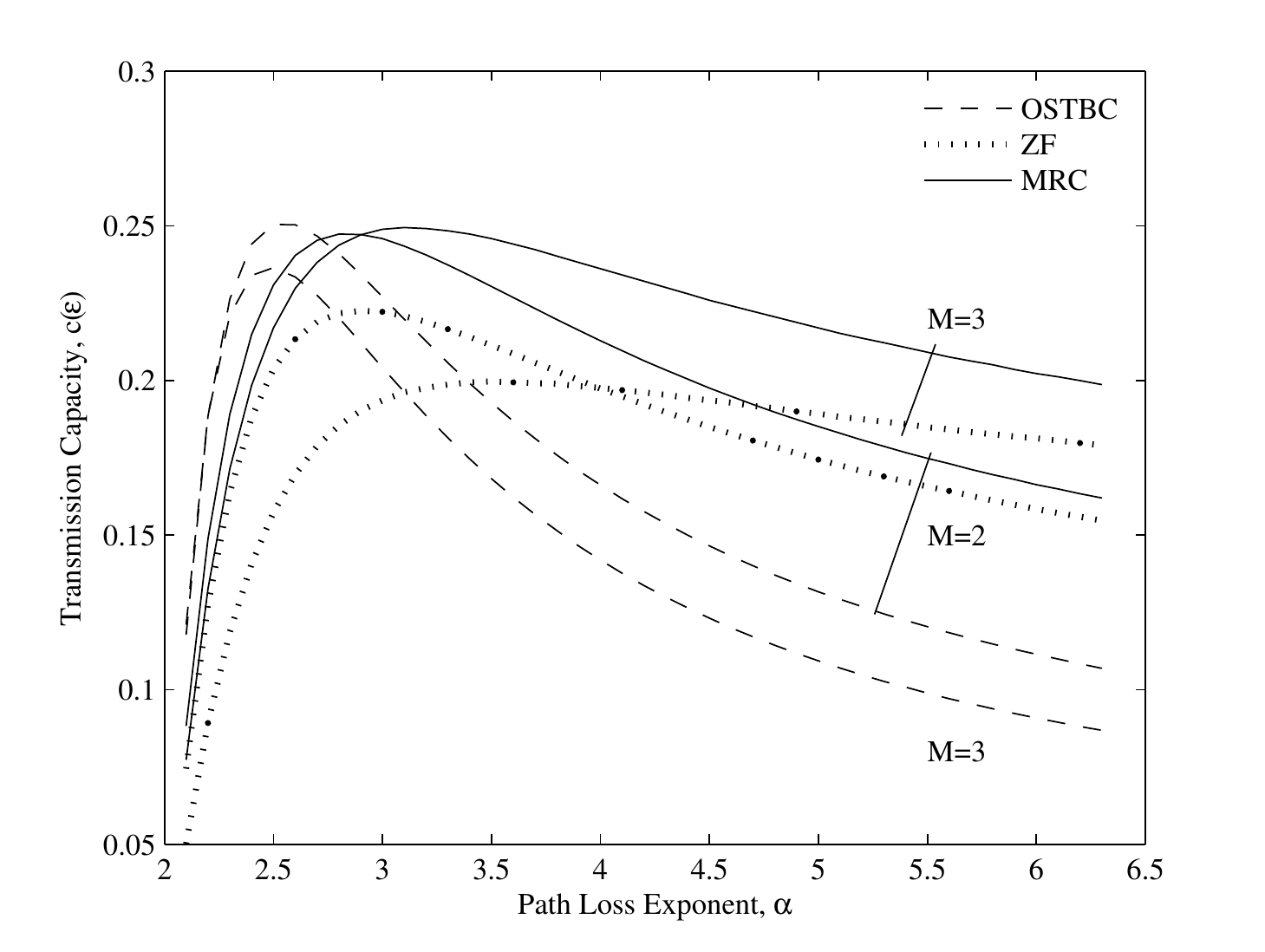}}
\caption{Transmission capacity vs.\ path loss exponent of slotted ALOHA with spatial multiplexing and OSTBC, and with $N=5$, ${r_{\rm tr}}=1$ m, $\beta=-5$ dB, $N=5$, $\rho=30$ dB, $\epsilon=0.15$ and $p=1$.}
\label{fig:tc_compare_changeM}
\end{figure}


\section{Network Throughput Comparison between Slotted ALOHA and Coordinated Access protocol}\label{sec:CA}

In this section, we compare the multi-antenna transmission schemes employing a slotted ALOHA MAC protocol, described in Section \ref{sec:sysmodel}, with a baseline single antenna transmission scheme employing a CA MAC protocol. It is worth noting that the CA protocol is an idealized protocol, since the overhead involved in achieving full coordination is prohibitive in practice for ad hoc networks. We will show that the simple multi-antenna slotted ALOHA MAC protocol can in fact
yield even higher throughputs than the tightly
scheduled CA MAC protocol in various practical scenarios. This shows that not only does the use of multiple antennas allow for a decentralized MAC, but it can actually also improve throughput in certain scenarios. We start by describing the CA protocol model.

\subsection{Coordinated Access Protocol Model}

For the CA protocol, transmissions are scheduled such that
interference is minimized at each receiver. Since the strongest
interferers are those closest to the receivers, we consider a CA
scheduling protocol which removes the closest interferers within a
guard zone around each receiver. Note that, although removing these
strong interferers leads to a gain in throughput, this gain comes at
the cost of significant overhead and synchronization requirements;
something which is clearly undesirable in an ad hoc network setting.

Scheduling nodes to minimize interference has been considered for a
long time, however there have been few results where the
transmitting nodes form a PPP. In \cite{venkataraman06}, a circular
guard zone of radius $r_{\rm gz}$ around the typical receiver was
considered, but neglected a guard zone around every other receiver
in the network. Recently in \cite{hasan07}, transmissions were
scheduled only if the receiver has no interferers within a circular
guard zone also of radius $r_{\rm gz}$. However, the throughput
suffers when the network becomes dense, as the probability of a
receiver with no interferers within its guard zone decreases.

In this section, we consider an alternative model where transmissions
are scheduled in a TDMA manner.  As with other common CA protocols,
the throughput of the TDMA scheme we consider does not go to zero as
the network becomes dense. For analytical tractability, we assume
each receiver is located at the center of a non-overlapping square
of side length $2 r_{\rm gz}$, arranged according to a lattice
structure. The transmitting nodes, distributed according to a PPP, which lie in
the same square are each assigned a time slot of the same duration
as the slots used in the slotted ALOHA protocol.  We assume a
dynamic slot allocation where the number of slots in a particular
square is equal to the number of transmitters in that square. We
note that although this may not be completely practical due to the
deterministic placement of each receiver, this model correctly
captures the fact that each receiver is spatially separated such
that interference is minimized, and is very useful for comparison
purposes.

The CA protocol is illustrated in Fig.\ \ref{fig:TDMA_plot}, where
we show a network before and after scheduling.
 The scheduled
transmitter transmits to the receiver at the center of the
corresponding square during the assigned time slot. In Fig. \ref{fig:TDMA_plot}, the receivers are shown as crosses, while transmitters as dots. As shown, after applying the CA protocol, there is at most one transmitter in each square, which transmits to the receiver at the center of the square.

\begin{figure*}[t!]
\centerline{\includegraphics[width=1\columnwidth]{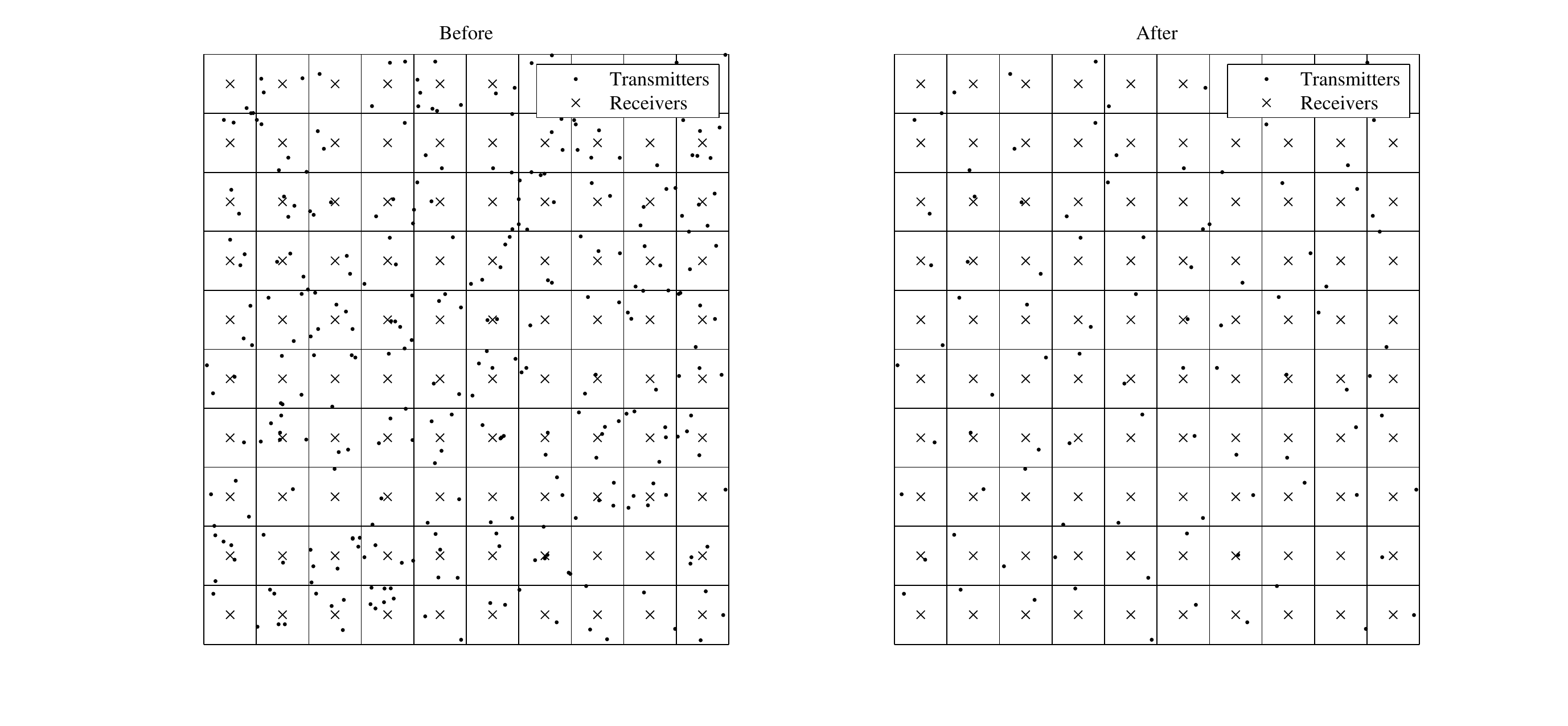}}
\caption{Coordinated access protocol model: before and after scheduling with $\lambda=3$ ${\rm m}^2$.}
\label{fig:TDMA_plot}
\end{figure*}

Under this model, the following theorem presents a new closed-form
upper bound for the network throughput.
\begin{theorem}\label{lem:throughput_CA}
The throughput of the CA MAC protocol employing single antennas is upper bounded by
\begin{align}\label{eq:throughput_CA}
& {\rm T}_{\rm CA} \le e^{-\frac{\beta r_{\rm tr}^\alpha}{\rho}}
\lambda_{\rm CA} \exp \left\{ - \pi\lambda_{\rm CA} \left( \frac{2
(\beta r_{\rm tr}^{\alpha})^{\frac{2}{\alpha}} \Gamma\left(\frac{2}{\alpha}\right)\Gamma\left(1 - \frac{2}{\alpha}\right)}{\alpha}
 - \left(\frac{2r_{\rm gz}}{\sqrt{\pi}}\right)^2 \;
_2 F_1 \left(\frac{2}{\alpha},1;1 + \frac{2}{\alpha};-\left(\frac{2
r_{\rm gz}}{ \sqrt{\pi} r_{\rm tr}}\right)^{\alpha}\frac{1}{\beta}
\right) \right) \right\}
\end{align}
where $\lambda_{\rm CA} = (1 - e^{-4 \lambda r_{\rm gz}^2})/(4
r_{\rm gz}^2)$, and ${}_2F_1(\cdot,\cdot;\cdot;\cdot)$ is the Gauss
hypergeometric function \cite[Eq. (9.6.2)]{abramowitz70}.

\begin{proof}
See the Appendix \ref{app:throughput_CA}.
\end{proof}
\end{theorem}

Note that the same throughput can also be obtained by considering an
alternative but similar system where the transmitting nodes around
the square guard zone transmit with probability $(1 - e^{-4 \lambda
r_{\rm gz}^2})/(4 \lambda r_{\rm gz}^2)$ at each time slot. Thus the
transmission probability of the CA protocol adaptively adjusts to
the intensity of transmitting nodes.  This is in contrast to the
fixed transmission probability employed in slotted ALOHA.

For the case of dense networks, the throughput of the CA protocol
approaches a constant as $\lambda \to \infty$, and is given by substituting $\lambda_{\rm CA} = \lambda_{\rm CA}^{\infty} = 1/4 r_{\rm gz}^2$ into (\ref{eq:throughput_CA}). This constant value reflects the fact that the CA
protocol ensures that only one node within a square of length $2
r_{\rm gz}$ is allowed to transmit.  This
models, for example, a system using a TDMA protocol when a large
number of transmitting nodes causes the throughput to reach a
saturation point, and the only transmitting nodes are the ones
allocated to the time slot corresponding to the current transmission
period.

Fig.\ \ref{fig:throughput_CA} presents the throughput  vs.\ the
guard zone parameter $r_{\rm gz}$ for different intensities. We see
that the analytical upper bounds in (\ref{eq:throughput_CA})
accurately match the Monte Carlo simulated curves in most cases. We
also notice that there is an optimal value of $r_{\rm gz}$ which
maximizes the throughput of the system. In addition, we see that the dense network throughput, obtained by substituting $\lambda_{\rm CA}^{\infty}$ into (\ref{eq:throughput_CA}), closely
matches the Monte Carlo simulated curves for transmitting node
intensities as small as 0.03.

\begin{figure}[t!]
\centerline{\includegraphics[width=0.7\columnwidth,keepaspectratio]{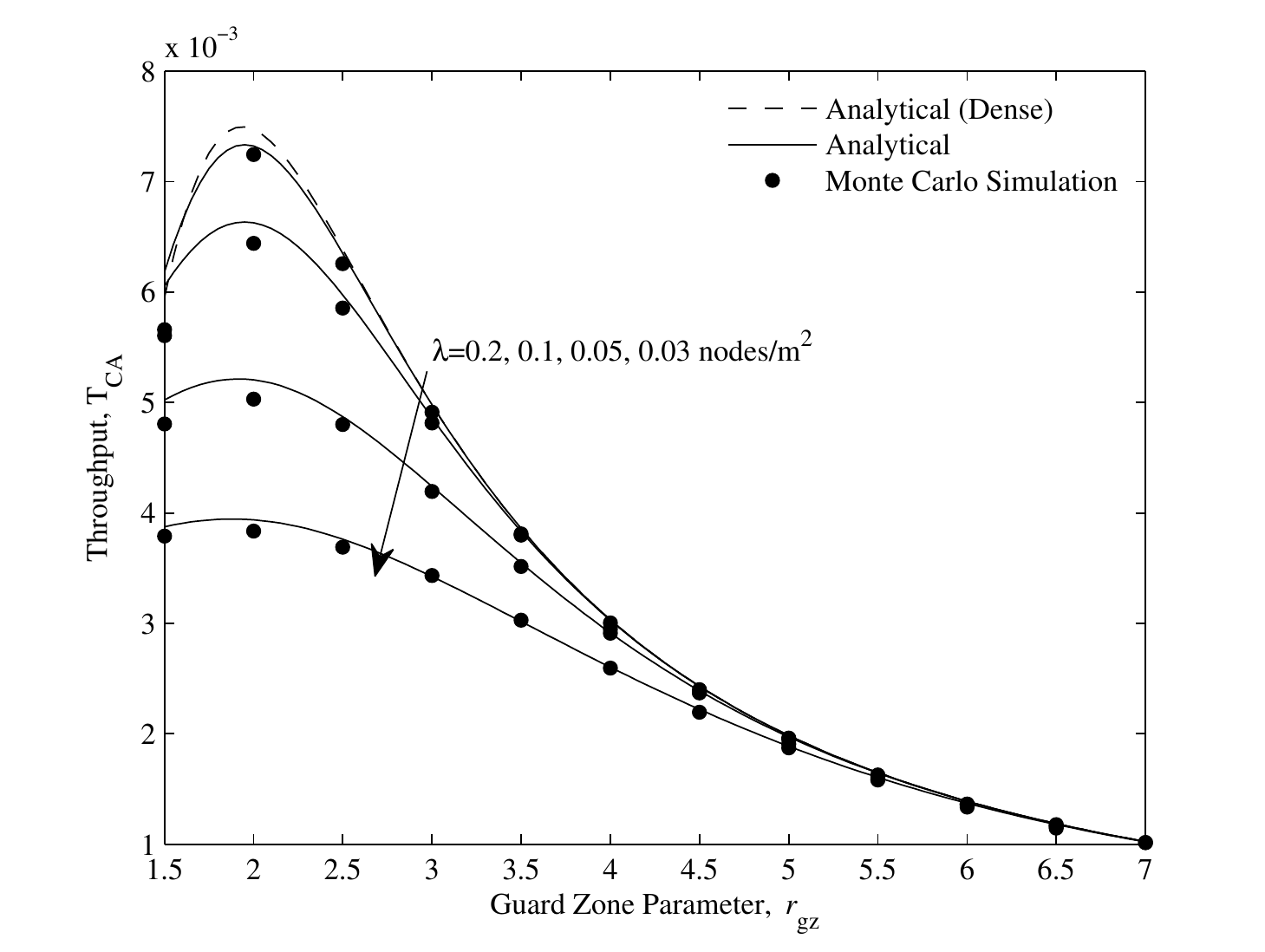}}
\caption{Throughput of the CA protocol vs.\ guard zone parameter with
$\rho=10$ dB, $r_{\rm tr}=1.5$ m, $\alpha=4$ and $\beta=5$ dB.}
\label{fig:throughput_CA}
\end{figure}

\subsection{Comparison}\label{sec:compare}

In this section, we compare the slotted ALOHA MAC protocol employing multiple antennas with the CA MAC protocol
employing single antennas. When considering the CA protocol we choose the
guard zone parameter $r_{\rm gz}$ which maximizes the throughput for
each value of $\lambda$. Also, when considering the slotted ALOHA protocol we only utilize one transmit antenna and multiple receive antennas. Note that our analysis in Sections \ref{sec:per} and \ref{sec:analysis_compare} shows that utilizing more transmit antennas for spatial multiplexing and OSTBC may lead to a higher throughput in different networking scenarios. However, as our purpose is to show the practicality of the slotted ALOHA protocol, we consider the simple case when only one transmit antenna is utilized.  Note that in this scenario, the performance of spatial multiplexing and OSTBC are the same. We assume both protocols use
the same initial intensity of transmitting nodes $\lambda$, thereby
allowing for a fair comparison.

Fig.\ \ref{fig:compare_SA_CA} presents the intensities and
transmission probabilities for which the throughput of slotted ALOHA
is greater than for the CA protocol.  Results are presented for
different antenna configurations, and the respective throughput
values are calculated from (\ref{eq:through_def}) and
(\ref{eq:throughput_CA}). Note that the throughput for slotted ALOHA can also be obtained from  (\ref{eq:through_def}). We clearly see that increasing the number
of receive antennas leads to an increase in the intensity range for which
the slotted ALOHA throughput is greater than the CA
throughput.

\begin{figure}[t!]
\centerline{\includegraphics[width=0.7\columnwidth,keepaspectratio]{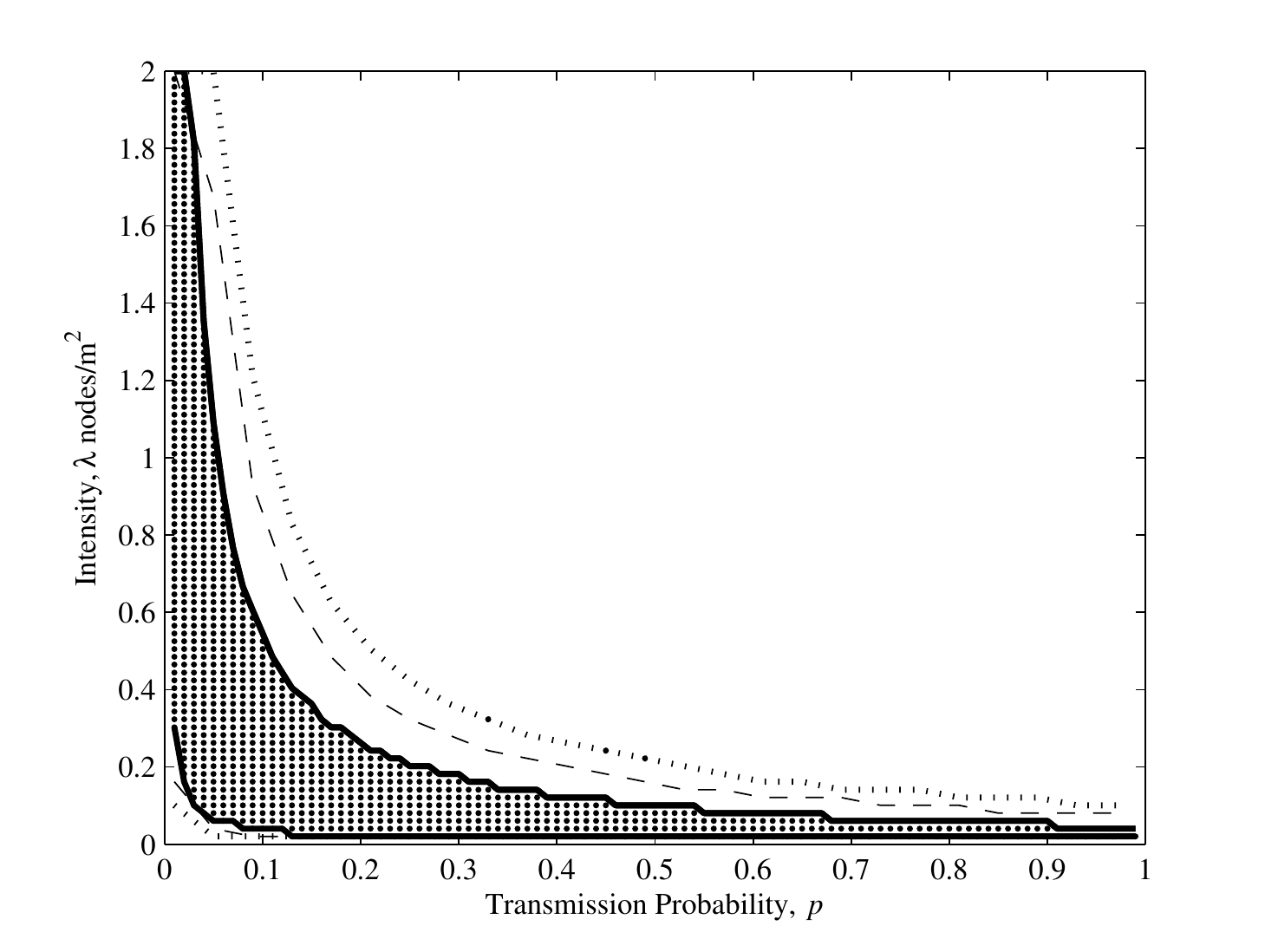}}
\caption{Comparison between slotted ALOHA and the CA protocol. The shaded area
between the solid lines shows the region where the throughput of
slotted ALOHA with $M=1$ in
(\ref{eq:through_def}) is greater than the CA protocol in
(\ref{eq:throughput_CA}); for $N=2$.  The dashed lines shows that the
region expands for $N=3$, and the dotted lines shows further
expansion for $N=4$. Results are shown for $M=1$, $\rho=10$ dB, $r_{\rm
tr}=2$ m, $\alpha=3$ and $\beta=5$ dB.}\label{fig:compare_SA_CA}
\end{figure}


In addition, Fig.\ \ref{fig:compare_SA_CA} shows that slotted ALOHA
outperforms the CA protocol for nearly all transmission probability
values when the network is sparse. However, as the intensity
increases, the transmission probability has to decrease in order for
the throughput of slotted ALOHA to be greater than the CA
protocol. Although not shown in Fig. \ref{fig:compare_SA_CA}, the
throughput is an increasing function of the number of antennas.

\section{Conclusion}

Open-loop point-to-point transmission schemes and slotted ALOHA MAC are practical choices for ad hoc networks, due to their relatively low feedback requirements and decentralized structure, respectively. Within this framework we have analyzed important network performance measures for spatial multiplexing transmission with MRC and ZF receivers, and for OSTBC, using tools from stochastic geometry. We derived new expressions for the outage probability, which were subsequently used to evaluate the network throughput and transmission capacity.

Based on our analysis, we demonstrated that the optimal number of transmit antennas, and the relative performance of the spatial multiplexing and OSTBC systems is dependent on a number of system parameters, including the chosen SINR operating value and node intensity. In particular, from a throughput perspective, our analysis has led to the following design guidelines:
\begin{itemize}
\item For dense networks and high SINR operating values, cyclic antenna diversity codes utilizing the maximum number of transmit antennas is the preferred scheme.
\item For low SINR operating values, spatial multiplexing with MRC receivers, while transmitting with the maximum number of data streams, is the preferred scheme.
\end{itemize}
For applications operating under a strict outage constraint, the transmission capacity is a preferable performance measure.  From a transmission capacity perspective, the following were observed:
\begin{itemize}
\item For sufficiently high path loss exponents, spatial multiplexing is preferred over OSTBC, and vice-versa for low path loss exponents.
\item In the large antenna $N$ regime, and if the self-interference is sufficiently smaller than the SINR operating value:
    \begin{itemize}
        \item MRC and ZF receivers can achieve \emph{linear} scaling, while OSTBC can only achieve \emph{sub-linear} scaling.
    \item For OSTBC, either the Alamouti scheme or SIMO is the optimal scheme, depending on the path loss exponent and number of receive antennas.
    \end{itemize}
\end{itemize}

Finally, we demonstrated the practicality of the multi-antenna slotted ALOHA MAC system in ad hoc networks, by comparing it with a baseline single antenna CA MAC system. We showed the interesting result that in many practical scenarios, the multi-antenna slotted ALOHA MAC system can actually outperform the idealized single antenna CA MAC in terms of network throughput. Further, we demonstrated that increasing the number of antennas can have the benefit of increasing the range of node intensities for which the use of slotted ALOHA MAC outperforms the idealized single antenna CA MAC.

\section{Acknowledgments}
The authors would like to thank Prof. Nihar Jindal for useful
discussions regarding the linear scaling of spatial multiplexing.

\begin{appendix}

\subsection{Derivation of SINR Approximation for OSTBC}\label{app_ostbc_sinr_bound}

We require the following lemmas:

\begin{lemma}\label{lem:exp_dist_equivalent}

Let $\mathbf{h} \dis \mathcal{CN}_{N \times 1}\left(\mathbf{0}_{N
\times 1}, \mathbf{I}_N \right)$ and $\mathbf{g} \dis
\mathcal{CN}_{N \times 1}\left(\mathbf{0}_{N \times 1}, \mathbf{I}_N
\right)$ be independent random variables. Then
\begin{align}\label{eq:exp_ind}
\frac{|\mathbf{h}^\dagger \mathbf{g}|^2}{||\mathbf{h}||^2} \dis {\rm Gamma}(0,1)
\end{align}
is independent of $\mathbf{h}$.
\end{lemma}
\begin{proof}
See \cite{shah00}.
\end{proof}

\begin{lemma}\label{lem:add_term_bound}
Let $\mathbf{h} \dis \mathcal{CN}_{N \times 1}\left(\mathbf{0}_{N
\times 1}, \mathbf{I}_N \right)$, \newline $\mathbf{g} \dis \mathcal{CN}_{N
\times 1}\left(\mathbf{0}_{N \times 1}, \mathbf{I}_N \right)$, $x
\dis \mathcal{CN}(0,1)$ and $y \dis \mathcal{CN}(0,1)$, all mutually
independent. Then
\begin{align}\label{eq:lem_cond_inequality}
|\mathbf{h}^\dagger \mathbf{g}|^2 <_{\rm st} |\mathbf{h}^\dagger \mathbf{g} + xy|^2
\end{align}
where
\begin{align}
A < _{\rm st} B \; \; \Longrightarrow \; \; {\rm F}_A(z) > {\rm
F}_B(z) , \quad {\rm for} \; z \in \mathds{R}^{+},
\end{align}
for random variables $A$ and $B$ with c.d.f.s ${\rm F}_A(\cdot )$
and ${\rm F}_B(\cdot )$ respectively.
\end{lemma}
\begin{proof}
Setting $X_1=\mathbf{h}^\dagger \mathbf{g}$, the distribution of
$X_1$ conditioned on $\mathbf{h}$ is $\mathcal{CN} \left(0,
||\mathbf{h}||_F^2\right)$. Thus, $|X_1|^2$ conditioned on
$\mathbf{h}$ has the same distribution as  $||\mathbf{h}||_F^2 Y_1$,
where $Y_1 \dis {\rm Exp}(1)$. Similarly, setting
$X_2=\mathbf{h}^\dagger \mathbf{g} + xy$, $|X_2|^2$ conditioned on
$\mathbf{h}$ has the same distribution as $\left(||\mathbf{h}||_F^2
+ |x|^2\right) Y_2$, where $Y_2 \dis {\rm Exp}(1)$.
(\ref{eq:lem_cond_inequality}) now follows.
\end{proof}

\begin{lemma}\label{lem:sum_gamma_equiv_dist}
Let $X_i \dis {\rm Gamma}\left(k_\ell,\theta\right)$, for
$i=1,\ldots,L$, all mutually independent. Then
\begin{align}
\sum_{i=1}^L X_i \dis {\rm Gamma}\left( \sum_{i=1}^L k_i ,\theta\right) \; .
\end{align}
\end{lemma}
\begin{proof}
See \cite{papoulis02}.
\end{proof}

\begin{lemma}\label{lem:gamma_increase_shape}
Let $X_1 \dis {\rm Gamma}(k_1, \theta)$ and $X_2 \dis {\rm Gamma}(k_2, \theta)$. If $k_1 > k_2$, then
\begin{align}
{\rm E}\left[X_1^n \right] > {\rm E}\left[X_2^n \right]
\end{align}
for any $n \in \mathds{Z}^+$.
\end{lemma}
\begin{proof}
The $n$th moment of $X \dis {\rm Gamma}\left(k,\theta\right)$  is
\begin{align}\label{eq:gamma_moments}
{\rm E}\left[X^n \right] &= \frac{\Gamma(n+k) \theta^{n} }{\Gamma(k)
}  \; .
\end{align}
The proof follows upon showing the positivity of the derivative of
(\ref{eq:gamma_moments}) w.r.t.\ $k$.
\end{proof}

\begin{lemma}\label{lemm:gamma_approx_OSTBC}
Let $\mathbf{h} \dis \mathcal{CN}_{L \times 1}\left(\mathbf{0}_{L
\times 1}, \mathbf{I}_L \right)$, $\mathbf{g} \dis \mathcal{CN}_{L
\times 1}\left(\mathbf{0}_{L \times 1}, \mathbf{I}_L \right)$ and
$\mathbf{y} \dis \mathcal{CN}_{T \times 1}\left(\mathbf{0}_{T \times
1}, \mathbf{I}_T \right)$, all mutually independent. Then $Z= \frac{
|\mathbf{h}^\dagger \mathbf{g}|^2}{ ||\mathbf{h}||^2 +
||\mathbf{y}||^2} $ has approximately the same distribution as
\begin{align}
Z^{\approx} \dis {\rm Gamma}\left(\frac{L}{L+T},1\right) .
\end{align}
\end{lemma}
\begin{proof}
We justify this approximation by first noting that
${\rm E}[Z] = {\rm E}[Z^{\approx}]$. Also, applying \emph{Lemma \ref{lem:exp_dist_equivalent}} and
\emph{Lemma \ref{lem:gamma_increase_shape}} to $Z$ yields: $Z^{\rm
lb} < Z \le Z^{\rm ub}$, where $Z^{\rm lb} \dis \lim_{\varsigma \to
0} {\rm Gamma}\left(\varsigma,1\right)$ and $Z^{\rm ub} \dis {\rm
Gamma}\left(1,1\right)$. Now, $Z \to Z^{\rm lb}$ when $L \to 0$ and
$Z \to Z^{\rm ub}$ when $L \to L+T \Rightarrow T \to 0$, implying
that $Z \to Z^{\approx}$ at the two extremes: $\frac{L}{L+T} \to 0$
and $\frac{L}{L+T} \to 1$. For values away from these two extremes,
our approximation is motivated by noting that for a fixed $L+T$,
applying \emph{Lemma \ref{lem:add_term_bound}} to $Z$ and
\emph{Lemma \ref{lem:gamma_increase_shape}} to $Z^{\approx}$
reveals, respectively, that $Z$ and $Z^{\approx}$ both increase with
$L$.
\end{proof}

We are now in a position to derive our approximation for the OSTBC
SINR. We first outline the main issues which make dealing with the
exact SINR (\ref{eq:ostbc_sinr_exact}) difficult, and provide some
justification for our approximations which are used to overcome
these.
\begin{itemize}
\item The first challenge is caused by mutual dependencies between random variables in the SINR. In particular, the normalized interference power, $\mathcal{K}_{\ell,\sum}$, is comprised of a sum of dependent random variables. As in \cite{hunter07,choi07}, we will neglect this mutual dependence, which is justified in \cite{choi07} by noticing that the common terms in $\mathcal{K}_{\ell,\sum}$ are multiplied by independent random variables, so they should be ``nearly independent''.  In addition, the numerator and denominator in (\ref{eq:ostbc_sinr_exact}) are dependent. Once again, as in \cite{choi07}, we neglect this dependence, with the argument that there are sufficient independent terms in the numerator and denominator that the dependence is weak.
\item
The second challenge is the difficulty in computing the exact
marginal distribution of $\mathcal{K}_{\ell,\sum}$ analytically,
even under the independence assumption discussed above. This appears
intractable, and to address this we seek an accurate approximation
by fitting a gamma distribution. Such distributions are known to be
exact for various cases and are expected to be quite accurate in
general. They are also analytically friendly, allowing us to invoke
the useful properties given in the lemmas above.
\end{itemize}

To most clearly illustrate our methods, we will first present our
derivations for a specific OSTBC code; namely, the OSTBC code in
(\ref{eq:ostbc_code3}). We will then give a brief discussion to
highlight how the derivation extends to general OSTBC codes. It is
convenient to define the operator ${\rm vecdim}(\cdot )$, which
takes as input the inner product of two vectors with equal
dimension, and returns the dimension of these vectors.

\subsubsection{Derivation for the OSTBC Code in (\ref{eq:ostbc_code3})} \label{subsec:DerivEx}

Without loss of generality, we focus on decoding $x_{0,1}$. The
$\ell$th row of $\varsigma_k(\mathbf{V}) $ and $\mathbf{M}_k$ are
given respectively by
\begin{align}\label{eq:varsigma3}
& \varsigma_k(\mathbf{V})_{\ell} = \left[\mathbf{V}_{\ell,1}, \hspace{0.2cm} \mathbf{V}_{\ell,2}^*, \hspace{0.2cm}-\mathbf{V}_{\ell,3}^*, \hspace{0.2cm} -\mathbf{V}_{\ell,4}\right] \hspace{0.2cm} {\rm and} \notag \\ &
\left(\mathbf{M}_k\right)_{\ell} = \left[h_{0,\ell,1}^*, \hspace{0.2cm}  h_{0,\ell,2}, \hspace{0.2cm}  h_{0,\ell,3}, \hspace{0.2cm}  h_{0,\ell,4}^* \right] \; .
\end{align}
Substituting (\ref{eq:ostbc_code3}) and (\ref{eq:varsigma3}) into
(\ref{eq:data_est_ostbca}), the data estimate for $x_{0,1}$ can be
written as
\begin{align}\label{eq:data_est_ostbc3}
\hat{x}_{0,1} = ||\mathbf{H}_0||_F^2 x_{0,1} +  \sum_{D_{\ell} \in \Phi} \sqrt{\frac{1}{|D_{\ell} |^\alpha}} I_{\ell}+ \sum_{k=1}^{N}\left(h_{0,k,1}^*n_{k,1} + h_{0,k,2} n_{k,2}^* - h_{0,k,3} n_{k,3}^* - h_{0,k,4}^* n_{k,4} \right) \notag \\
\end{align}
where
\begin{align}\label{eq:int_ostbc3}
I_{\ell}  &= W_{\ell,1} x_{\ell,1} + W_{\ell,2}
x_{\ell,2}+Z_{\ell,2} x_{\ell,2}^* +W_{\ell,3} x_{\ell,3}+Z_{\ell,3}
x_{\ell,3}^* \; ,
\end{align}
with
\begin{align}\label{eq:wizi1}
W_{\ell,1}& =\sum_{k=1}^{N} \left(h_{0,k,1}^* h_{\ell,k,1} + h_{0,k,2}
h_{\ell,k,2}^*  +h_{0,k,3} h_{\ell,k,3}^* + h_{0,k,4}^*
h_{\ell,k,4}\right) \; ,
\end{align}
\begin{align}\label{eq:wizi2}
& W_{\ell,2}=\sum_{k=1}^{N}\left(h_{0,k,1}^* h_{\ell,k,2} - h_{0,k,2}
h_{\ell,k,1}^* \right) \notag \\
& Z_{\ell,2}=
\sum_{k=1}^{N}\left(-h_{0,k,3} h_{\ell,k,4}^* + h_{0,k,4}^*
h_{\ell,k,3} \right) \; ,
\end{align}
\begin{align}\label{eq:wizi3}
& W_{\ell,3}=\sum_{k=1}^{N}\left(h_{0,k,1}^* h_{\ell,k,3}  - h_{0,k,3}
h_{\ell,k,1}^* \right) \notag \\
&  Z_{\ell,3}=
\sum_{k=1}^{N}\left(h_{0,k,2} h_{\ell,k,4}^* - h_{0,k,4}^*
h_{\ell,k,2} \right)\; .
\end{align}
Note that the $W_{\ell,q}$'s and $Z_{\ell,q}$'s can each be
interpreted as the inner product of two equal-length vectors; in
each case, with one vector having elements drawn from either $\pm
\mathbf{H}_0$ or $\pm \mathbf{H}_0^\dagger$, and the other vector
having elements drawn from either $\pm \mathbf{H}_\ell$ or $\pm
\mathbf{H}_{\ell}^\dagger$. For example, $W_{\ell,2}$ in
(\ref{eq:wizi2}) can be written as
\begin{align}
W_{\ell,2} = \mathbf{w}_{0} \cdot \mathbf{w}_{\ell}
\end{align}
where $\mathbf{w}_{0}$ and $ \mathbf{w}_{\ell}$ are $2N \times 1$
vectors given respectively by
\begin{align}\label{eq:ostbc_weight_w0}
\mathbf{w}_0 = [h_{0,1,1}^{*} , \, -h_{0,1,2} , \, \ldots \, ,
h_{0,N,1}^{*} , \, -h_{0,N,2}]^T
\end{align}
and
\begin{align}\label{eq:ostbc_weight_well}
\mathbf{w}_{\ell} = [h_{\ell,1,2} , \, h_{\ell,1,1}^* , \, \ldots ,
\, h_{\ell,N,2} , \, h_{\ell,N,1}^*]^T \; .
\end{align}
With this interpretation, it is clear that ${\rm
vecdim}(W_{\ell,1})=4N$ and ${\rm vecdim}(W_{\ell,2})={\rm
vecdim}(Z_{\ell,2})={\rm vecdim}(W_{\ell,3})={\rm
vecdim}(Z_{\ell,3})=2N$.

From (\ref{eq:data_est_ostbc3}), the SINR can be written as
\begin{align}\label{eq:ostbc_SINR_code3}
\gamma_{{\rm OSTBC},1} 
&=\frac{\frac{\rho}{R M r_{\rm tr}^\alpha}
||\mathbf{H}_0||_F^2}{\frac{\rho}{RM} \sum_{D_{\ell} \in
\Phi}\frac{1}{|D_{\ell}|^\alpha} \mathcal{K}_{\ell,\sum} + 1 } \;
\end{align}
where $\mathcal{K}_{\ell,\sum}=\frac{{\rm
E}_{x_{\ell,1},x_{\ell,2},x_{\ell,3}}\left[I_{\ell}
I_{\ell}^*\right]}{||\mathbf{H}_0||_F^2}$.  Note that
\begin{align}\label{eq:int_alt_code3}
\mathcal{K}_{\ell,\sum} &= \mathcal{K}_{\ell,1} +\mathcal{K}_{\ell,2} +\mathcal{K}_{\ell,3}
\end{align}
where
\begin{align}\label{eq:k1_code3}
\mathcal{K}_{\ell,1} = \frac{ |W_{\ell,1}|^2}{ ||\mathbf{H}_0||_F^2}
\end{align}
and, for $q = 2, 3$,
\begin{align}\label{eq:kj}
\mathcal{K}_{\ell,q} &= \frac{{\rm E}_{x_{\ell,q}}\left[\bigl|W_{\ell,q} x_{\ell,q}+Z_{\ell,q} x_{\ell,q}^*\bigr|^2\right] }{||\mathbf{H}_0||_F^2} \notag \\
&=\frac{|W_{\ell,q}|^2  +|Z_{\ell,q}|^2  }{||\mathbf{H}_0||_F^2} .
\end{align}
It is clear that all three terms in (\ref{eq:int_alt_code3}) are
mutually dependent. Moreover, there is dependence within each term,
since $|W_{\ell,q}|^2$, $|Z_{\ell,q}|^2$, and $||\mathbf{H}_0||_F^2$
involve common elements.  As stated previously, we neglect these
dependencies. Thus applying \emph{Lemma
\ref{lemm:gamma_approx_OSTBC}} to
$\frac{|W_{\ell,q}|^2}{||\mathbf{H}_0||_F^2}$ and
$\frac{|Z_{\ell,q}|^2}{||\mathbf{H}_0||_F^2}$, followed by applying
\emph{Lemma \ref{lem:sum_gamma_equiv_dist}} reveals that
$\mathcal{K}_{\ell,q}$ is approximately distributed as
\begin{align}
{\rm Gamma}\left({\rm vecdim} (W_{\ell,q}) + {\rm vecdim}
(Z_{\ell,q}),1\right)
\end{align}
for all $q = 1, 2, 3$.  (Note that for the particular code
considered here, $Z_{\ell,1}=0$.) Now, from (\ref{eq:ostbc_code3})
and (\ref{eq:int_ostbc3}), it can be seen that
\begin{align}
\sum_{q=1}^{N_s} {\rm vecdim}(W_{\ell,q})+{\rm vecdim}(Z_{\ell,q}) = N N_I \, ,
\end{align}
and we can therefore deduce that $\mathcal{K}_{\ell,\sum}$ is
approximately distributed as ${\rm Gamma}\left( N_I /M , N_I
\right)$.
Finally, with the aforementioned assumption that the numerator and
denominator of the SINR (\ref{eq:ostbc_SINR_code3}) are independent,
we obtain the desired approximation.

\subsubsection{Extension to General OSTBC Codes}

Here we give some details of how to generalize the previous
derivation for arbitrary OSTBC codes. Since this extension is quite
straightforward, the discussion is kept brief. The challenge, once
again, is to approximate the distribution of
$\mathcal{K}_{\ell,\Sigma}=\frac{{\rm
E}_{\mathbf{X}_{\ell}}\left[||\mathbf{M}_k \odot \varsigma
\left(\mathbf{H}_{\ell} \mathbf{X}_{\ell}
\right)||_1^2\right]}{||\mathbf{H}_0||_F^2}$ in
(\ref{eq:ostbc_sinr_exact}). As before, it can be shown that
\begin{align}\label{eq:k_general}
\mathcal{K}_{\ell,\Sigma} =\sum_{q=1}^{N_s} \mathcal{K}_{\ell,q} ,
\quad \quad  \mathcal{K}_{\ell,q} = \frac{|W_{\ell,q}|^2
+|Z_{\ell,q}|^2 }{||\mathbf{H}_0||_F^2} \;
\end{align}
where $W_{\ell,q}$ and $Z_{\ell,q}$ are each composed of the inner
product of two equal-length vectors; one vector having elements
drawn from $\pm \mathbf{H}_0$ or $\pm \mathbf{H}_0^\dagger$ without
repetition, the other having elements drawn from $\pm
\mathbf{H}_\ell$ or $\pm \mathbf{H}_\ell^\dagger$ without
repetition. As evident from the previous example, the specific set
of selected elements as well as the dimension of the vectors
comprising the inner products are dependent on the particular OSTBC
code employed.

Under the same independence assumptions as before, we apply
\emph{Lemma \ref{lemm:gamma_approx_OSTBC}} to
$\frac{|W_{\ell,q}|^2}{||\mathbf{H}_0||_F^2}$ and
$\frac{|Z_{\ell,q}|^2}{||\mathbf{H}_0||_F^2}$, followed by applying
\emph{Lemma \ref{lem:sum_gamma_equiv_dist}}, to find that
$\mathcal{K}_{\ell,q}$ is approximately distributed as
\begin{align}
{\rm Gamma}\left({\rm vecdim} (W_{\ell,q}) + {\rm vecdim}
(Z_{\ell,q}),1\right) \, .
\end{align}
It follows that $\mathcal{K}_{\ell,\Sigma}$ is approximately
distributed as ${\rm Gamma}\left( N_I / M, 1\right)$. The desired
SINR approximation is then obtained by invoking the assumption that
the numerator and denominator and independent, as before.

\subsection{Proof of Theorem \ref{lemm_general_cdf}}\label{app:general_cdf_proof}

We first rewrite the c.d.f.\ of $\gamma$ as
\begin{align}
{\rm F}_\gamma(\beta) &= {\rm Pr} \left(\frac{W}{Y + \sum_{\ell \in \Phi} |X_\ell|^{-\alpha} \Psi_{\ell 0}+ 1} \le \beta\right) \notag \\
&  = 1-{\rm Pr} \left(\sum_{\ell \in \Phi} |X_\ell|^{-\alpha} \Psi_{\ell 0} \le \frac{W}{\beta} - Y-1\biggr|A  \right) {\rm Pr}(A)  -{\rm Pr} \left(\sum_{\ell \in \Phi} |X_\ell|^{-\alpha} \Psi_{\ell 0} \le \frac{W}{\beta} - Y-1\biggr|\bar{A}  \right) {\rm Pr}(\bar{A})\notag \\
&= 1-{\rm Pr} \left(\sum_{\ell \in \Phi} |X_\ell|^{-\alpha} \Psi_{\ell 0} \le \frac{W}{\beta} - Y-1\biggr|A  \right) {\rm Pr}(A)
\end{align}
where $A$ denotes the event $\frac{W}{\beta} - Y>1$ and $\bar{A}$ denotes the complement of $A$. Applying \cite[Eq. (113)]{weber07} along with some algebraic manipulation, the c.d.f.\ of $\gamma$, conditioned on $W$ and $Y$, can be written as
\begin{align}
{\rm F}_{\gamma|A,W,Y}(\beta)
&= 1 - \sum_{k=0}^\infty \frac{\left(- \pi p \lambda \Gamma\left(1-\frac{2}{\alpha}\right) \beta^{\frac{2}{\alpha}} \right)^k}{k! \Gamma\left(1-\frac{2k}{\alpha}\right)}   \left( \left(W - Y \beta-\beta\right)^{-\frac{2}{\alpha}}{\rm E} \left[\Psi^{\frac{2}{\alpha}}\right]\right)^k {\rm Pr}(A) \notag \;
\end{align}
where $\Psi\dis\Psi_{\ell 0}$ for all $\ell$. Averaging with respect to (w.r.t.) $W$ gives
\begin{align}\label{eq:cdf_conditionY}
{\rm F}_{\gamma|Y}(\beta)&= 1 - \sum_{k=0}^\infty \frac{\left(- \pi p \lambda \Gamma\left(1-\frac{2}{\alpha}\right) \beta^{\frac{2}{\alpha}}  {\rm E}\left[\Psi^{\frac{2}{\alpha}}\right]\right)^k}{k! \Gamma\left(1-\frac{2k}{\alpha}\right)} \mathcal{I}_1
\end{align}
where
\begin{align}\label{eq:proof_integral}
\mathcal{I}_1 &= \int_{\beta(Y+1)}^\infty \left(w - Y \beta-\beta\right)^{-\frac{2k}{\alpha}} f_W(w) {\rm d} w \notag \\
&= \frac{e^{-\frac{\beta(Y+1)}{\theta}}}{\Gamma(m) \theta^m} \sum_{\ell=0}^{m-1} \binom{m-1}{\ell} (\beta(Y+1))^{\ell} \Gamma\left(m-\ell-\frac{2k}{\alpha}\right) \theta^{m-\ell-\frac{2k}{\alpha}} \; .
\end{align}
Substituting (\ref{eq:proof_integral}) into (\ref{eq:cdf_conditionY}), we have
\begin{align}\label{eq:cdf_conditionY2}
& {\rm F}_{\gamma|Y}(\beta)= 1 - \frac{e^{-\frac{\beta(Y+1)}{\theta}}}{\Gamma(m) } \sum_{\ell=0}^{m-1} \binom{m-1}{\ell} \left(\frac{\beta(Y+1)}{\theta}\right)^{\ell} \sum_{k=0}^\infty \frac{\left(- \frac{\pi p \lambda \Gamma\left(1-\frac{2}{\alpha}\right) \beta^{\frac{2}{\alpha}}  {\rm E}\left[\Psi^{\frac{2}{\alpha}}\right]}{\theta^{\frac{2}{\alpha}}}\right)^k \Gamma\left(m-\ell-\frac{2k}{\alpha}\right)}{k! \Gamma\left(1-\frac{2k}{\alpha}\right)} \; .
\end{align}
The remaining challenge is to remove the infinite series in (\ref{eq:cdf_conditionY2}). To this end, we apply the identity \cite[Eq.\ (24.1.3)]{abramowitz70}
\begin{align}
\frac{\Gamma\left(m-\ell-\frac{2k}{\alpha}\right)}{\Gamma\left(1-\frac{2k}{\alpha}\right)} = (-1)^{m-\ell-1} \sum_{i=0}^{m-\ell-1} s(m-\ell,i+1) \left(\frac{2k}{\alpha}\right)^i
\end{align}
and Dobi\'{n}ski's Formula \cite{rota64}, which gives
\begin{align}\label{eq:cdf_conditionY3}
{\rm F}_{\gamma|Y}(\beta)&= 1 - \frac{(-1)^{m-1} e^{- \frac{\pi p \lambda \Gamma\left(1-\frac{2}{\alpha}\right) \beta^{\frac{2}{\alpha}}  {\rm E}\left[\Psi^{\frac{2}{\alpha}}\right]}{\theta^{\frac{2}{\alpha}}}} e^{-\frac{\beta(Y+1)}{\theta}}}{\Gamma(m) } \sum_{\ell=0}^{m-1} \binom{m-1}{\ell} \left(-\frac{\beta(Y+1)}{\theta}\right)^{\ell} \notag \\
& \hspace{1cm} \sum_{i=0}^{m-\ell-1} s(m-\ell,i+1) \left(\frac{2}{\alpha}\right)^i \sum_{j=0}^i S(i,j) \left(- \frac{\pi p \lambda \Gamma\left(1-\frac{2}{\alpha}\right) \beta^{\frac{2}{\alpha}}  {\rm E}\left[\Psi^{\frac{2}{\alpha}}\right]}{\theta^{\frac{2}{\alpha}}}\right)^j \; .
\end{align}
Averaging w.r.t.\ $Y$, we have
\begin{align}\label{eq:cdf_general_proof}
{\rm F}_{\gamma}(\beta)&= 1 - \frac{(-1)^{m-1}  e^{- \frac{\pi p \lambda \Gamma\left(1-\frac{2}{\alpha}\right) \beta^{\frac{2}{\alpha}}  {\rm E}\left[\Psi^{\frac{2}{\alpha}}\right]}{\theta^{\frac{2}{\alpha}}}} e^{-\frac{\beta}{\theta}}}{\Gamma(m) } \sum_{\ell=0}^{m-1} \binom{m-1}{\ell} \left(-\frac{\beta}{\theta}\right)^{\ell}  \sum_{i=0}^{m-\ell-1} s(m-\ell,i+1)\\
& \hspace{0.2cm}\times \left(\frac{2}{\alpha}\right)^i \sum_{j=0}^i S(i,j) \left(- \frac{\pi p \lambda \Gamma\left(1-\frac{2}{\alpha}\right) \beta^{\frac{2}{\alpha}}  {\rm E}\left[\Psi^{\frac{2}{\alpha}}\right]}{\theta^{\frac{2}{\alpha}}}\right)^j \mathcal{I}_2  \notag
\end{align}
where
\begin{align}\label{eq:integral2}
\mathcal{I}_2 &= \int_0^\infty e^{-\frac{\beta y}{\theta}} \left(y+1\right)^{\ell}  {\rm f}_Y(y) {\rm d} y\notag \\
&= \frac{1}{\Gamma(u) \Upsilon^u}\sum_{\tau=0}^{\ell} \binom{\ell}{\tau} \left(\frac{\beta }{\theta}+\frac{1}{\Upsilon}\right)^{-\tau-u} \Gamma(\tau+u) \; .
\end{align}
Finally, by noting that
\begin{align}\label{eq:psi_exp}
{\rm E}\left[\Psi^{\frac{2}{\alpha}}\right] 
&= \frac{\Omega^{\frac{2}{\alpha}}\Gamma\left(n+\frac{2}{\alpha}\right) }{\Gamma(n)}
\end{align}
and substituting (\ref{eq:integral2}) and (\ref{eq:psi_exp}) into (\ref{eq:cdf_general_proof}), we obtain the desired result.

\subsection{Proof of \emph{Corollary \ref{corr:tc_scaling}}}\label{app:tc_scaling}

We first note that by taking a Taylor expansion of (\ref{eq:cdf_general})
around $\lambda=0$, and then finding the inverse of the resulting
expression w.r.t.\ $\lambda p$, the transmission capacity can be written in the general form
\begin{align}\label{eq:TC_alt_general}
c(\epsilon) = \frac{\zeta \left(\epsilon-{\rm F}^{\rm SU}(\beta) \right)^{+} }{ \pi \beta^{\frac{2}{\alpha}} r_{\rm tr}^2 {\rm E}\left[\Psi^{\frac{2}{\alpha}}\right] g\left({\rm SNR}_0\right) } + O\left(\left(\left(\epsilon-{\rm F}^{\rm SU}(\beta)\right)^2 \right)^{+} \right)
\end{align}
where
\begin{align}\label{eq:gSNR}
g\left({\rm SNR}_0\right)&=\int_{\beta(Y+1)}^\infty \left(w-Y \beta- \beta \right)^{-\frac{2}{\alpha}} f_{W}\left(w\right)  {\rm d} w  \; .
\end{align}
To proceed, we note that if $X \dis {\rm Gamma}(n,\vartheta)$, then $\lim_{n \to \infty} \frac{X-n \vartheta}{\sqrt{n} \vartheta} \dis \mathcal{CN}(0, 1)$. Now if $\lim_{n \to \infty} \sqrt{n} \vartheta \to 0$ and $\lim_{n \to \infty} n \vartheta \to c$ where $c$ is a constant, then $\lim_{n \to \infty} X \to c$. We see from Table \ref{table:gamma_param} that as $N \to \infty$ with $M=\kappa N$ where $0<\kappa<1$, $\sqrt{m} \Theta \to 0$, $\sqrt{u} \Upsilon \to 0$ and $\sqrt{n} \Omega \to 0$, and thus
\begin{align}\label{eq:largeXPsi}
W \to m \Theta, \qquad Y \to u \Upsilon  \qquad  {\rm and} \qquad \Psi \to n \Omega \; .
\end{align}
The result follows by substituting (\ref{eq:largeXPsi}) into (\ref{eq:TC_alt_general}).

\subsection{Proof of Proposition \ref{prop:TCZF}}\label{app:TCZF_scaling_proof}

Let $M = f(N)$.
For large $M$ and $N$,
\begin{align} \label{eq:ZFMRC_transCapScale}
c &= \Theta\left(f(N) \left(\frac{N}{f(N)} -\xi
\right)^{\frac{2}{\alpha}} \right)
\end{align}
where
\begin{align}
 \xi =
\begin{cases}
1  + \frac{\beta r_{\rm tr}^\alpha}{\rho} , & {\rm for \, ZF}  \\
\beta  + \frac{\beta r_{\rm tr}^\alpha}{\rho} , & {\rm for \, MRC}
\end{cases}
\; .
\end{align}
Since $M \le N$, $f(N)$ must satisfy $f(N) \le N$. Therefore, $f(N)$
must scale \emph{linearly or sub-linearly} with $N$. If $f(N)$ is
linear in $N$, then we know that the transmission capacity is linear
in $N$ also. However, if $f(N)$ scales sub-linearly, then the
transmission capacity (\ref{eq:ZFMRC_transCapScale}) becomes
\begin{align}\label{eq:proof_TCZF}
c &= \Theta\left(N^{\frac{2}{\alpha}} f^{1-\frac{2}{\alpha}}(N)
\right) \;
\end{align}
which is also sub-linear in $N$.

\subsection{Proof of Lemma \ref{lem:throughput_CA}}\label{app:throughput_CA}

The intensity $\lambda_{\rm CA}$ after applying the CA protocol is
obtained by noting that at most one node is allowed to transmit
in each square of length $2 r_{\rm gz}$, and $1-e^{-4 \lambda r_{\rm
gz}^2}$ is the percentage of receivers with at least one node inside
the square guard zone. By denoting the new aggregate interference
after applying the CA protocol as $\rho \mathcal{I}_{\Phi_{{\rm
CA}}}$, and noting that the typical transmitter-receiver channel, $|h_0|^2$, is exponentially distributed, we
can write the throughput as
\begin{align}\label{eq:throughput_CA_proof}
{\rm T}_{\rm CA} & = \lambda_{\rm CA} (1-{\rm F}(\beta))
\notag \\
&= \lambda_{\rm CA}\int_0^{\infty} {\rm P}\left( |h_0|^2  \ge
\zeta (\rho x + 1) \right)
f_{\mathcal{I}_{\Phi_{{\rm CA}}}}(x) {\rm d} x \notag\\
&= e^{-\zeta} \lambda_{\rm CA} \int_0^{\infty} e^{-\zeta \rho x}
f_{\mathcal{I}_{\Phi_{{\rm CA}}}}(x)
{\rm d} x \notag \\
&= e^{-\zeta} \lambda_{\rm CA} \exp \left\{ -\lambda_{\rm CA}
\int_{\mathds{R}^2} 1-  {{\rm E}}_{h_0} \left[ e^{-\rho \zeta |h_0|^2 /|t|^{\alpha}} \right] {\rm d} t\right\} \notag \\
&= e^{-\zeta} \lambda_{\rm CA}  \exp \left\{ -\lambda_{\rm CA}
\int_{\mathds{R}^2} \frac{\rho \zeta}{|t|^{\alpha}+\rho \zeta}{\rm d}
t\right\}
\end{align}
where $\zeta = \frac{\beta r_{\rm tr}^{\alpha}}{\rho }$ and
$f_{\mathcal{I}_{\Phi_{{\rm CA}}}}(x)$ is the p.d.f.\ of $\Phi_{\rm CA}$. Note that the third line in (\ref{eq:throughput_CA_proof}) was obtained by applying standard results in stochastic geometry (see e.g.,\ \cite[pp. 231]{stoyan95}). Note that the integral in
(\ref{eq:throughput_CA_proof}) is over the entire infinite plane, except for
a square guard zone centered at the origin of length $2 r_{\rm gz}$.
Due to the difficulty in evaluating an exact closed-form solution to this integral, we obtain an
upper bound by assuming a circular guard zone of radius $2 r_{\rm
gz}/\sqrt{\pi}$ centered at the origin. The use of a circular guard zone as an upper bound can be seen by first noting
that the same number of interferers, on average, are removed from within both
the square and circular guard zones. Second, by denoting $S$ as the interferers which do not fall within both guard zones, then the interferers in $S$ which are enclosed by the circle of radius $2 r_{\rm gz}/\sqrt{\pi}$ are closer to the origin than the interferers in $S$ which are enclosed by the square
of side length $2 r_{\rm gz}$. This implies that the use of the circular guard zone removes interferers which are closer to the typical receiver, than the interferers which are removed using the square guard zone.
The throughput upper bound is thus given by
\begin{align}\label{eq:throughput_CA_proof2}
{\rm T}_{\rm CA} \le e^{-\zeta} \lambda_{\rm CA} \exp \left\{ -2 \pi
\lambda_{\rm CA}
 \int_{{\frac{2 r_{\rm gz}}{\sqrt{\pi}}}}^{\infty}
\frac{\rho \zeta r}{r^{\alpha}+\rho \zeta } {\rm d} r \right\}
\; .
\end{align}
Finally, we solve the integral in (\ref{eq:throughput_CA_proof2})
using \cite[Eq. (3.259.2)-(3.259.3)]{gradshteyn65} followed by some algebraic manipulation.

\end{appendix}

\IEEEpeerreviewmaketitle


\begin{thebibliography}{10}
\providecommand{\url}[1]{#1}
\def\UrlFont{\rmfamily}
\providecommand{\newblock}{\relax}
\providecommand{\bibinfo}[2]{#2}
\providecommand\BIBentrySTDinterwordspacing{\spaceskip=0pt\relax}
\providecommand\BIBentryALTinterwordstretchfactor{4}
\providecommand\BIBentryALTinterwordspacing{\spaceskip=\fontdimen2\font plus
\BIBentryALTinterwordstretchfactor\fontdimen3\font minus
  \fontdimen4\font\relax}
\providecommand\BIBforeignlanguage[2]{{%
\expandafter\ifx\csname l@#1\endcsname\relax
\typeout{** WARNING: IEEEtran.bst: No hyphenation pattern has been}%
\typeout{** loaded for the language `#1'. Using the pattern for}%
\typeout{** the default language instead.}%
\else
\language=\csname l@#1\endcsname
\fi
#2}}

\bibitem{tse05}
D.~Tse and P.~Viswanath, \emph{Fundamentals of Wireless Communication},
  1st~ed.\hskip 1em plus 0.5em minus 0.4em\relax New York: Cambridge, 2005.

\bibitem{louie08}
R.~H.~Y. Louie, M.~R. McKay, and I.~B. Collings, ``Maximum sum-rate of {MIMO}
  multiuser scheduling with linear receivers,'' \emph{{IEEE} Trans. Commun.},
  vol.~57, no.~11, pp. 3500--3510, Nov. 2009.

\bibitem{larsson03}
E.~G. Larsson and P.~Stoica, \emph{Space-Time Block Coding for Wireless
  Communications}.\hskip 1em plus 0.5em minus 0.4em\relax United Kingdom:
  Cambridge University Press, 2003.

\bibitem{forenza06b}
A.~Forenza, M.~R. McKay, I.~B. Collings, and R.~W. {Heath Jr.}, ``Switching
  between {OSTBC} and spatial multiplexing with linear receivers in spatially
  correlated {MIMO} channels,'' in \emph{Proc.\ of IEEE Veh. Tech. Conf.
  (VTC)}, Melbourne, Australia, May 2006, pp. 1387--1391.

\bibitem{weber05}
S.~P. Weber, X.~Yang, J.~G. Andrews, and G.~de~Veciana, ``Transmission capacity
  of wireless ad hoc networks with outage constraints,'' \emph{{IEEE} Trans.
  Inform. Theory}, vol.~51, no.~12, pp. 4091--4102, Dec. 2005.

\bibitem{weber07}
S.~P. Weber, J.~G. Andrews, and N.~Jindal, ``The effect of fading, channel
  inversion, and threshold scheduling on ad hoc networks,'' \emph{{IEEE} Trans.
  Inform. Theory}, vol.~53, no.~11, pp. 4127--4149, Nov. 2007.

\bibitem{hasan07}
A.~Hasan and J.~G. Andrews, ``The guard zone in wireless ad hoc networks,''
  \emph{{IEEE} Trans. Wireless Commun.}, vol.~6, no.~3, pp. 897--906, Mar.
  2007.

\bibitem{govindasamy07}
S.~Govindasamy, D.~W. Bliss, and D.~H. Staelin, ``Spectral efficiency in
  single-hop ad-hoc wireless networks with interference using adaptive antenna
  arrays,'' \emph{{IEEE} J. Select. Areas Commun.}, vol.~25, no.~7, pp.
  1358--1369, Sept. 2007.

\bibitem{huang08}
\BIBentryALTinterwordspacing
K.~Huang, J.~G. Andrews, R.~W. {Heath Jr.}, D.~Guo, and R.~A. Berry, ``Spatial
  interference cancellation for multi-antenna mobile ad-hoc networks,'' 2008.
  [Online]. Available: \url{http://arxiv.org/abs/0807.1773v1}
\BIBentrySTDinterwordspacing

\bibitem{jindal09}
N.~Jindal, J.~G. Andrews, and S.~P. Weber, ``Rethinking {MIMO} for wireless
  networks: Linear throughput increases with multiple receive antennas,'' in
  \emph{Proc.\ of IEEE Int. Conf. on Commun. (ICC)}, Dresden, Germany, June
  2009, pp. 1--6.

\bibitem{vaze09}
R.~Vaze and R.~W. {Heath Jr.}, ``Optimal use of multiple antennas in ad-hoc
  networks: Transmission capacity perspective,'' 2009, [Online]. Available:
  \url{http://users.ece.utexas.edu/~vaze/Research/TC-MIMO/TC-MIMOVersion4.pdf}.

\bibitem{hunter07}
A.~M. Hunter, J.~G. Andrews, and S.~P. Weber, ``Transmission capacity of ad hoc
  networks with spatial diversity,'' \emph{{IEEE} Trans. Wireless Commun.},
  vol.~7, no.~12, pp. 5058--5071, Dec. 2008.

\bibitem{hunter08conf}
A.~M. Hunter and J.~G. Andrews, ``Adaptive rate control over multiple spatial
  channels in ad hoc networks,'' in \emph{Proc.\ of Modeling and Optimization
  in Mobile, Ad Hoc, and Wireless Networks and Workshops (WiOPT)}, Berlin,
  Germany, Apr. 2008, pp. 469--474.

\bibitem{kountouris09}
M.~Kountouris and J.~G. Andrews, ``Transmission capacity scaling of {SDMA} in
  wireless ad hoc networks,'' in \emph{Proc.\ of IEEE Int.\ Symposium on Info.
  Theory}, Taormina, Italy, Nov. 2009, pp. 534--538.

\bibitem{stamatiou07}
K.~Stamatiou, J.~G. Proakis, and J.~R. Zeidler, ``Evaluation of {MIMO}
  techniques in {FH-MA} ad hoc networks,'' in \emph{Proc.\ of IEEE Global
  Telecommun. Conf. (GLOBECOM)}, Washington D.C., USA, Nov. 2007.

\bibitem{stoyan95}
D.~Stoyan, W.~S. Kendall, and J.~Mecke, \emph{Stochastic Geometry and its
  Applications}, 2nd~ed.\hskip 1em plus 0.5em minus 0.4em\relax England: John
  Wiley and Sons, 1995.

\bibitem{mckay06_thesis}
\BIBentryALTinterwordspacing
M.~R. McKay, ``Random matrix theory analysis of multiple antenna communication
  systems,'' {PhD} Thesis, {S}chool of Electrical and Information Engineering,
  University of Sydney, Australia. 2006. [Online]. Available:
  \url{http://ihome.ust.hk/~eemckay/thesis.pdf}
\BIBentrySTDinterwordspacing

\bibitem{james64}
A.~T. James, ``Distributions of matrix variates and latent roots derived from
  normal samples,'' \emph{Ann. Math. Statist.}, vol.~35, no.~2, pp. 475--501,
  1964.

\bibitem{alamouti98}
S.~M. Alamouti, ``A simple transmit diversity technique for wireless
  communications,'' \emph{{IEEE} J. Select. Areas Commun.}, vol.~16, no.~8, pp.
  1451--1458, Oct. 1998.

\bibitem{paulraj03}
A.~Paulraj, R.~Nabar, and D.~Gore, \emph{Introduction to Space-Time Wireless
  Communications}.\hskip 1em plus 0.5em minus 0.4em\relax United Kingdom:
  Cambridge University Press, 2003.

\bibitem{tarokh99}
V.~Tarokh, H.~Jafarkhani, and A.~R. Calderbank, ``Space-time block codes from
  orthogonal designs,'' \emph{{IEEE} Trans. Inform. Theory}, vol.~45, no.~5,
  pp. 1456--1467, July 1999.

\bibitem{liang03}
X.-B. Liang, ``Orthogonal designs with maximal rates,'' \emph{{IEEE} Trans.
  Inform. Theory}, vol.~49, no.~10, pp. 2468--2503, Oct. 2003.

\bibitem{choi07}
W.~Choi, N.~Himayat, S.~Talwar, and M.~Ho, ``The effects of co-channel
  interference on spatial diversity techniques,'' in \emph{Proc.\ of IEEE
  Wireless Commun.\ and Networking Conf.\ (WCNC)}, Hong Kong, Mar. 2007, pp.
  1936--1941.

\bibitem{abramowitz70}
M.~Abramowitz and I.~A. Stegun, \emph{Handbook of Mathematical Functions with
  Formulas, Graphs, and Mathematical Tables}, 9th~ed.\hskip 1em plus 0.5em
  minus 0.4em\relax New York: Dover Publications, 1970.

\bibitem{baccelli06}
F.~Baccelli, B.~Blaszcyszyn, and P.~Muhlethaler, ``An {ALOHA} protocol for
  multihop mobile wireless networks,'' \emph{{IEEE} Trans. Inform. Theory},
  vol.~52, no.~2, pp. 421--436, Feb. 2006.

\bibitem{tse99}
D.~N.~C. Tse and S.~V. Hanly, ``Linear multiuser receivers: effective
  interference, effective bandwidth and user capacity,'' \emph{{IEEE} Trans.
  Inform. Theory}, vol.~45, no.~2, pp. 641--657, Mar. 1999.

\bibitem{haenggi05}
M.~Haenggi, ``On distances in uniform random networks,'' \emph{{IEEE} Trans.
  Inform. Theory}, vol.~51, no.~10, pp. 3584--3586, Oct. 2005.

\bibitem{venkataraman06}
J.~Venkataraman, M.~Haenggi, and O.~Collins, ``Shot noise models for outage and
  throughput analyses in wireless ad hoc networks,'' in \emph{Proc.\ of
  Military Commun.\ Conf. (MILCOM)}, Washington, D.C., Oct. 2006, pp. 1--7.

\bibitem{shah00}
A.~Shah and A.~M. Haimovich, ``Performance analysis of maximal ratio combining
  and comparison with optimum combining for mobile radio communications with
  co-channel interference,'' \emph{IEEE Trans. on Veh. Technol.}, vol.~49,
  no.~4, pp. 1454--1463, July 2000.

\bibitem{papoulis02}
A.~Papoulis and S.~U. Pillai, \emph{Probability, Random Variables and
  Stochastic Processes}, 4th~ed.\hskip 1em plus 0.5em minus 0.4em\relax North
  America: McGraw-Hill, 2002.

\bibitem{rota64}
G.-C. Rota, ``The number of partitions of a set,'' \emph{The American Monthly
  Magazine}, vol.~71, no.~5, pp. 498--504, May 1964.

\bibitem{gradshteyn65}
I.~S. Gradshteyn and I.~M. Ryzhik, \emph{Table of Integrals, Series, and
  Products}, 4th~ed.\hskip 1em plus 0.5em minus 0.4em\relax San Diego, CA:
  Academic, 1965.

\end{thebibliography}
\end{document}